[1994/12/01]
\documentclass[]{article}
\title{Optimal and Scalable Methods to Approximate the Solutions of Large-Scale Bayesian Problems: Theory and Application to Atmospheric Inversions and Data Assimilation}

\author{N. Bousserez and Daven K. Henze}
\date{2016/09/19}

\usepackage{amsmath,amsthm,amssymb,amsbsy,algorithm,algorithmic}
\usepackage{graphicx,footnote}
\usepackage{enumerate}
\usepackage{color}
\usepackage{lineno}
\usepackage{natbib}
\usepackage[margin=1.3in]{geometry}

\chardef\bslash=`\\ 

 \theoremstyle{plain} 
\newtheorem{thm}{Theorem}[section]
\newtheorem{cor}[thm]{Corollary}
\newtheorem{lem}[thm]{Lemma}
\newtheorem{prop}[thm]{Proposition}

\theoremstyle{definition}
\newtheorem{defn}[thm]{Definition}%

\theoremstyle{remark}
\newtheorem{rem}{\textbf{Remarks}}[section]

\newcommand{\st}{\sigma}

\newcommand{\eval}[2][\right]{\relax
  \ifx#1\right\relax \left.\fi#2#1\rvert}

\begin{document}
\maketitle
\markboth{Optimal and Scalable Methods to Approximate Bayesian Solutions}
{Optimal and Scalable Methods to Approximate Bayesian Solutions}
\renewcommand{\sectionmark}[1]{}

\begin{abstract}
This paper provides a detailed theoretical analysis of methods to approximate the solutions of high-dimensional ($>10^6$) linear Bayesian problems. An optimal low-rank projection that maximizes the information content of the Bayesian inversion is proposed and efficiently constructed using a scalable randomized SVD algorithm. Useful optimality results are established for the associated posterior error covariance matrix and posterior mean approximations, which are further investigated in a numerical experiment consisting of a large-scale atmospheric tracer transport source-inversion problem. This method proves to be a robust and efficient approach to dimension reduction, as well as a natural framework to analyze the information content of the inversion. Possible extensions of this approach to the non-linear framework in the context of operational numerical weather forecast data assimilation systems based on the incremental 4D-Var technique are also discussed, and a detailed implementation of a new Randomized Incremental Optimal Technique (RIOT) for 4D-Var algorithms leveraging our theoretical results is proposed.
\end{abstract}
\section{Introduction}
The Bayesian approach to inverse problems consists of updating a prior probability distribution of a quantity of interest conditioned on some physically-related observations. The conditioned probability distribution is called the posterior distribution. The Bayesian framework has been widely adopted to solve geophysical problems. For large-scale non-linear problems, such as those encountered in atmospheric modeling, sampling techniques (e.g., Markov Chain Monte-Carlo) to estimate the posterior distribution require prohibitively numerous integrations of the forward model that relates the inferred variables to the observations   \citep{Tarantola05}. Alternatively, when the forward model is linear and the probability distributions for the prior and the observations are Gaussian, the posterior probability distribution is Gaussian and can be fully characterized by its mean (i.e., the maximum-likelihood) and its error covariance matrix, for which analytical expressions are available. However, explicitly calculating the posterior error covariance matrix remains a challenging task in inverse problems for which the dimensions of the matrix are very large \citep{bousserez2015improved}. As an example, the typical number of optimized variables in data assimilation (DA) systems for numerical weather prediction (NWP) is $\sim10^8$, which corresponds to a posterior error covariance matrix of dimension $10^8 \times 10^8$. In such cases, the posterior error covariance matrix cannot be explicitly represented in computer memory, and appropriate low-rank approximations are needed to extract useful information. Low-rank estimations of the posterior error covariance matrix are also useful to compute other quantities of interest such as the model resolution matrix and the Degree Of Freedom for Signal (DOFS), which characterize the information content of the inversion \citep{rodgers2000inverse,Tarantola05}. 

 The variational approach to solving large-scale inverse problems employs tangent-linear and adjoint models with iterative gradient-based optimization algorithms to compute the maximum-likelihood of the posterior distribution. The potential of state-of-the-art optimization algorithms to provide posterior error covariance estimates as a by-product of the minimization have long been recognized (e.g., \citet{thacker1989role,Rabier92,Nocedal06,Muller05}). For large-scale problems, such optimization algorithms are usually halted before full convergence, effectively only approximating the solution. Although the convergence properties of these approximations have been investigated in previous numerical experiments (e.g., \citet{bousserez2015improved}), a theoretical analysis of their optimality with respect to the information content of the inversion has yet to be performed. Another approach to make large-scale inverse problems tractable is the use of ensembles to approximate the error covariances of the system. Such methods, which can be either stochastic (e.g., the Ensemble Kalman Filter (EnKF)) or deterministic (i.e., square-root formulations such as the Ensemble Adjusted Kalman Filter (EAKF)) have the advantage that they do not require the use of an adjoint model. However, the small number of ensembles used results in severe rank-deficiencies for the associated error covariance matrices. Sophisticated filtering localization techniques that help mitigate this sampling noise are the subject of intense research activities in ensemble-based DA (e.g., \citet{menetrier2015linear,andersonlei13}).

Besides implicit (i.e., incomplete variational minimizations) and explicit (i.e., ensemble-based) low-rank approximations, another approach to approximate the solution of large-scale Bayesian problems consists of performing a prior dimensional reduction of the system. In the context of linear problems, such methods can allow one to explicitly form the matrices associated with the inverse problem and to analytically compute the posterior mean and posterior error covariance. In the atmospheric inversion community, several studies have focused on designing objective methods to construct reduced spaces, so as to optimize some criteria related to the information content of the problem. In \citet{bocquet2011bayesian}, a rigorous multi-scale approach to dimension reduction is proposed, wherein an optimal aggregation scheme is defined by constructing a grid of tiles for which the associated reduced Bayesian problem has maximum DOFS. However, the optimization of the grid can be computationally expensive, and the method requires explicitly calculation of the Jacobian of the system, which is not feasible for high-dimensional systems. Moreover, the optimality of the solution is guaranteed only for a specific class of reductions, namely grid aggregation methods. More recently, \citet{turner2015balancing} proposed a method to construct an optimal reduced basis set of Gaussian-mixture functions to analytically solve large-scale atmospheric source inversion problems. Their algorithm consists of an incremental construction of the basis where the posterior error variances in observation space are recomputed at each iteration (i.e., for each new dimension) until a minimum total error is reach. The cost associated with computing the posterior error variances at each iteration makes this approach poorly scalable and not suitable for problems where the optimal basis needs to be constructed in a timely manner.  Similarly to \citet{bocquet2011bayesian}, their approach also lacks generality by restricting the analysis to a specific class of basis (namely, the Gaussian-mixture functions).

Recently, \citet{spantini2015optimal} presented a detailed theoretical analysis of optimal low-rank approximations of the posterior mean and posterior error covariance matrix for linear Bayesian problems. They show that the proposed approximations are defined in the subspace that maximizes the observational constraints, which is measured as the relative gain in information in the posterior with respect to the prior information. Interestingly, this method can reconcile theoretical optimality and computational scalability, since in practice the low-rank optimal approximations can be efficiently constructed by applying matrix-free singular value decomposition (SVD) routines to the so-called prior-preconditioned Hessian of the quadratic cost function. Furthermore, when high-performance computing is required, the use of recently developed randomized SVD methods allows to fully parallelize the algorithm and to implement and scalable approach to low-rank approximation for large-scale Bayesian problems (e.g., \citet{bui2012extreme}). 

In this paper, we provide a detailed theoretical analysis of the Bayesian approximation problem in the context of optimal projections. This approach has the advantage of producing approximations to the posterior mean and posterior error covariance matrix that are consistent with each other, i.e., they are both approximations to the full-dimensional posterior solutions and exact solutions to a projected low-rank Bayesian problem. Our mathematical developments generalize the theoretical framework of \citet{bocquet2011bayesian} and allow us to construct a projection that maximizes the DOF among all low-rank projections. This maximum-DOFS projection yields posterior mean and posterior error covariance approximations similar to those proposed in  \citet{spantini2015optimal}, for which we provide additional interpretations and optimality results. Moreover, although \citet{spantini2015optimal} identified their optimal approximations as the solutions of a projected Bayesian problem with maximum observational information with respect to the prior, we note that they did not rigorously demonstrate this result by omitting to analyze the so-called representativeness error. This error, which quantifies the impact of the unobserved subspace in the inversion as a result of the dimension reduction, has been taken into account in our proofs. For the first time, this optimal approximation method is applied to a large-scale atmospheric-transport source inversion problem using a highly-scalable randomized SVD algorithm. Finally, we investigate new links between the maximum-DOFS approximations and preconditioned conjugate-gradient (CG) algorithms embedded in non-linear Gauss-Newton minimization methods such as incremental 4D-Var in operational DA systems for NWP. This enables us to propose an improved incremental 4D-Var algorithm leveraging both our theoretical optimality results and the efficiency of randomized SVD algorithms.

Section \ref{theory} of this paper presents the theory and formalism of the optimal low-rank projection problem and provides useful optimality results for the associated approximations of the posterior mean and posterior covariance matrices. Section \ref{implementation} discusses the practical construction of the optimal approximations and describes in details a randomized SVD algorithm that allows highly-scalable implementation of the method. In Section \ref{num_exp}, we present a numerical experiment to illustrate the theoretical results established in Section \ref{theory} and test the computational performance of the randomized SVD approach to implement the optimal low-rank approximations. Our example consists of a high-dimensional atmospheric-transport source inversion problem using a large dataset of satellite observations. Finally, in Section \ref{link_da} we investigate the links between the proposed optimal approximations and variational optimization algorithms used in current operational DA systems for NWP, and we propose a new Randomized Incremental Optimal Technique (RIOT) for 4D-Var based on our findings. 

\section{Theory}
\label{theory}
\subsection{The Bayesian Problem}
\subsubsection{Finding the Maximum Likelihood}
\label{form_prel}
Here we shall review the Bayesian inversion approach to finding the maximum likelihood of a set of random variables, given some prior probability distribution functions (pdf) on these variables and on a set of physically-related observations, adopting the notations generally used in the numerical weather prediction community. Formally, the vector of observations, $\mathbf{y}$, is related to the so-called control vector $\mathbf{x}$ through a forward model operator, $H$:
 \begin{eqnarray}
 \label{eq1}
\mathbf{y}=H(\mathbf{x}),
\end{eqnarray}
where $\mathbf{x}\in E$, $\mathbf{y}\in F$, $H: E \rightarrow F$, and $E$ and $F$ are the control space (of dimension $n$) and the observation space (of dimension $p$), respectively.

Assuming Gaussian pdfs for the prior ($\mathbf{x}^b$) and the observations, with covariance error matrices $\mathbf{B}$ and $\mathbf{R}$, respectively, the maximum likelihood can be obtained by minimizing the following cost function:
\begin{eqnarray}
 \label{eq3}
J(\mathbf{x}) =\frac{1}{2}(H(\mathbf{x})-\mathbf{y})^T\mathbf{R}^{-1}(H(\mathbf{x})-\mathbf{y})+\frac{1}{2}(\mathbf{x}-\mathbf{x}^{b})^T\mathbf{B}^{-1}(\mathbf{x}-\mathbf{x}^{b}).
\end{eqnarray}
An analytical solution of (\ref{eq3}) can be expressed as: 
\begin{eqnarray}
 \label{post_update2}
\mathbf{x}^a=\mathbf{x}^b+\left(\mathbf{B}^{-1}+\mathbf{H}^T\mathbf{R}^{-1}\mathbf{H}\right)^{-1}\mathbf{H}^T\mathbf{R}^{-1}\left(\mathbf{y}-H(\mathbf{x}^b)\right)
\end{eqnarray}
where $\mathbf{H}$ is the Jacobian of the forward model. By applying the Shermann-Morrison-Woodbury formula to (\ref{post_update2}) \citep{sherman1949adjustment}, an alternative expression can be obtained: 
\begin{eqnarray}
 \label{eq2}
\mathbf{x}^a=\mathbf{x}^b+\mathbf{BH}^T\left(\mathbf{R+HBH}^T\right)^{-1}\left(\mathbf{y}-H(\mathbf{x}^b)\right),
\end{eqnarray}

Equations (\ref{eq2}) and (\ref{post_update2}) differ significantly in term of practical implementation. Eq. (\ref{eq2}) requires forming and inverting the ($p\times p$) matrix $\mathbf{R+HBH}^T$, while in Eq. (\ref{post_update2}) the ($n\times n$) matrix $\mathbf{B}^{-1}+\mathbf{H}^T\mathbf{R}^{-1}\mathbf{H}$ is inverted. The matrix $\mathbf{R+HBH}^T$ is called the matrix of innovation statistics, and it plays an important role in DA methods. 

For the common class of problems associated with a large control vector (e.g., $n>10^6$) but a small number of observations $(p<<n)$, provided that a tangent linear (i.e., an implicit $\mathbf{H}$) and an adjoint (i.e., an implicit $\mathbf{H}^T$) models are available, it may be possible to explicitly form and invert the matrix of innovation statistics and to compute the maximum likelihood exactly using Eq. (\ref{eq2}) (note that in this case $\mathbf{B}$ would need to be defined implicitly as well). This approach is at the core of the representer method \citep{bennett2005inverse}, which is similar to the so-called Physical Space Assimilation System (PSAS). Note that in practice $p$ adjoint and tangent-linear model integrations are required to extract the $p$ columns of $\mathbf{HBH}^T$. Although parallel implementation is possible to compute those $p$ columns, operational constraints (e.g., in NWP) and the limitation of computer resources may render this method impractical even for moderately large $p$ (e.g., $p>10^3$). In the case where $n$ is small enough, the maximum likelihood solution can be obtained following a similar approach, but using (\ref{post_update2}) instead of (\ref{eq2}). The formulation $(\ref{post_update2})$ is sometimes used in ensemble-based DA methods (e.g., EAKF) \citep{Anderson01}, where a small number of perturbed trajectories is used to produce a sample estimate of $\mathbf{B}^{-1}+\mathbf{H}^T\mathbf{R}^{-1}\mathbf{H}$. If neither of the analytical formulations (\ref{eq2}) and (\ref{post_update2}) can be directly used (e.g., if both $n$ and $p$ are very large), a variational optimization approach consisting of minimizing the cost function ($\ref{eq3}$) is usually the method of choice, provided an adjoint model is available. However, solutions obtained from iterative minimization techniques are often only approximations to the maximum likelihood solution, since in practice the iteration is halted before full convergence is reached.

In the present study, we shall assume that $n$ is very large ($n>10^6$) and propose optimal approximations to the Bayesian solution whose practical implementations present good scalability properties. The optimality criteria considered will rely on the information content of the inversion, whose rigorous definition is the object of the following Section.

\subsubsection{The Linear Case: Information Content and Incremental Formulation}
\label{info_content}

 In this study we shall assume that the forward model $H$ is linear, so that $H=\mathbf{H}$. The non-linear case will be treated in Section \ref{link_da}, which investigates operational DA assimilation methods in NWP. Assuming linearity for the forward model allows us to rigorously define and compute useful quantities characterizing the information content of the inversion \citep{rodgers2000inverse}. With a linear forward model, $\mathbf{H}$, the posterior distribution is Gaussian and the maximum likelihood is equal to the posterior mean.  If the linear approximation of the forward model is valid in a neighborhood of the maximum-likelihood, the local posterior pdf is approximately Gaussian, in which case the notion of posterior error covariance becomes (locally) meaningful. Moreover, for a linear forward model, the cost function defined in (\ref{eq3}) becomes quadratic, and the inverse of its Hessian at the minimum is equal to the posterior error covariance matrix, that is:
\begin{eqnarray}
 \label{eq4}
\mathbf{P}^a \equiv \overline{(\mathbf{x}^a-\mathbf{x}^t)(\mathbf{x}^a-\mathbf{x}^t)^T}= (\nabla ^2 J)^{-1}(\mathbf{x}^a)=(\mathbf{B}^{-1}+\mathbf{H}^T\mathbf{R}^{-1}\mathbf{H})^{-1},
\end{eqnarray}
where $\overline{\mathbf{x}}$ denotes the expectation of the random vector $\mathbf{x}$ and $\mathbf{x}^t$ represents the true state.
Another useful formulation for $\mathbf{P}^a$ can be obtained by applying the Shermann-Morrison-Woodbury formula to (\ref{eq4}):
\begin{eqnarray}
 \label{eq5}
\mathbf{P}^a = \mathbf{B}-\mathbf{BH}^T(\mathbf{HBH}^T+\mathbf{R})^{-1}\mathbf{HB}
\end{eqnarray}
This formula expresses $\mathbf{P}^a$ as a negative update of $\mathbf{B}$. The update term ($\mathbf{BH}^T(\mathbf{HBH}^T+\mathbf{R})^{-1}\mathbf{HB}$ can be interpreted as the posterior error reduction afforded by the observations. Another useful metric related to the information content of the problem is the DOFS, which quantifies the number of parameters independently constrained by the observations. It can be defined as the trace of the model resolution matrix (or averaging kernel) $\mathbf{A}$, which represents the sensitivity of the posterior mean to the true state \citep{rodgers2000inverse}:
\begin{eqnarray}
 \label{eq6}
\mathbf{A} \equiv \frac{\partial \mathbf{x}^a}{\partial \mathbf{x}^t} &= &\mathbf{Id}-\mathbf{P}^a\mathbf{B}^{-1} \\
 \label{eq7}
\text{DOFs}&=&\text{Tr}(\mathbf{A})
\end{eqnarray}

Finally, an additional useful formulation is to link the model resolution matrix to the posterior update term in (\ref{eq5}):
\begin{align}
 \label{eq8}
\mathbf{A}=\mathbf{BH}^T(\mathbf{HBH}^T+\mathbf{R})^{-1}\mathbf{H}
\end{align}
From (\ref{eq6}) we clearly see that $\mathbf{A}$ can be interpreted as a relative posterior error reduction.

Both $\mathbf{B}-\mathbf{P}^a$ and $\mathbf{A}$ characterize the information content of the inversion (in an absolute and relative sense, respectively) and will be central to our analysis. Since the triplet $(\mathbf{x}^{a}, \mathbf{P}^{a},\mathbf{A})$ fully characterizes the posterior pdf and the information content of the linear Bayesian problem, it shall be referred to as the solution of the Bayesian problem. It is worth noting that in our large-scale framework the matrices $\mathbf{P}^a$ and $\mathbf{A}$ cannot be computed directly nor represented explicitly in computer memory. Meaningful approximations of these quantities are therefore needed to properly interpret the statistical significance and the information content of the estimated posterior mean. Other applications include posterior sampling strategies (e.g., in cycling DA methods), where optimal and efficient approximations of the square-root for $\mathbf{P}^a$ are required (see Section \ref{opt_post_samp}).

In the case of a linear forward model, $\mathbf{H}$, the posterior update formula (\ref{eq2}) suggests a simplification of the problem by considering the increment $\delta \mathbf{x}\equiv\mathbf{x}-\mathbf{x}^b$ and innovation $\mathbf{d}\equiv\mathbf{y}-\mathbf{Hx}^b$ as the control and observation vector, respectively. The error statistics associated with the variables $\delta \mathbf{x}$ and $\mathbf{d}$ are the same as those associated with $\mathbf{x}$ and $\mathbf{y}$, respectively. Therefore, the previous equations defining the Bayesian solutions are unchanged when applying this change of variable. Note that in this incremental framework the prior is now the constant null vector  ($\delta\mathbf{x}^b=0$), whereas the true state is a random variable ($\delta\mathbf{x}^t=\mathbf{x}^t-\mathbf{x}^b$). In the rest of this paper, unless specified otherwise, the control vector $\mathbf{x}$ will be identified with the increment $\delta\mathbf{x}$, and the observation vector $\mathbf{y}$ replaced by the innovation vector $\mathbf{d}$.
 \subsubsection{Low-Rank Projections versus Low-Rank Approximations}
 \label{diff_proj_approx}
When approximating the solution of a large-scale Bayesian problem, a fundamental distinction needs to be made between low-rank projections and low-rank approximations.   A low-rank projections consists of restricting the Bayesian problem to a (small) subspace of the initial control space (by means of a projection operator). In this case a Bayesian problem of lower-rank is solved, and its solution can be considered to be an approximation of the initial problem in the sense that its posterior mean and posterior error covariance matrix converge to the true solutions as the reduced control space is (incrementally) increased. On the other hand, low-rank approximations of Bayesian problems belong to a more general class of methods that construct approximations of the posterior mean and posterior error covariance matrix of the initial high-dimensional problem, without the requirement of consistency between the approximated (low-rank) posterior mean and posterior error covariance (that is, they do not necessarily represent the posterior mean and corresponding posterior error covariance of a Bayesian problem). This distinction is important for interpretation as well as for applications of these methods. For instance, in a non-linear framework, a projection can be useful to define a low-rank version of a large-scale Bayesian problem to which MCMC sampling methods can be efficiently applied (e.g., \citet{cui2014likelihood}). Other approximations of the same rank that do not correspond to a projected Bayesian problem may provide better estimates of the true solution, but would not be suitable for this application. In our study, we will first describe the formalism of low-rank projections and provide an optimal projection for the Bayesian problem that maximizes the information content (i.e., the DOFS) of the inversion (see Section \ref{model_reduc}). We will then explore the link between the solutions of this optimal projected problem and optimal low-rank approximations of the posterior mean and posterior error covariance matrix (Section \ref{link_lr_approx}).

\subsection{Low-Rank Projections}
\label{model_reduc}

\subsubsection{General Formulation}
One way to reduce the computational cost associated with solving a large-scale Bayesian problem is to project the problem onto a small subspace, $E'\subset E$, of dimension $k<<n$. By construction, the projection restricts the posterior updates to the prior mean ($\mathbf{x}^b$) and to the prior error covariance matrix to the subspace $E'$, which effectively amounts to solving a problem of dimension $k$. The projection can be chosen so as to optimize some criteria, usually related to the information content of the inversion (e.g., maximum DOF or minimum posterior error covariance matrix for some norm). An important aspect of the projection is that it may induce an additional observational error if the observed subspace (i.e., the orthogonal of the kernel of $\mathbf{H}$) is not included in the range of the projector. This additional term is the so-called representativeness error. In the following we shall rigorously define the projected Bayesian problem and provide an analytical expression for the representativeness error.
  \begin{defn}[\textbf{Projected Bayesian Problem}]
Let us consider a Bayesian problem defined by $\mathcal{B}\equiv(E,F,\mathbf{H},\mathbf{B},\mathbf{R})$ (using the definitions in Section \ref{form_prel}), and a projection operator $\mathbf{\Pi}$ (i.e, $\mathbf{\Pi}^2=\mathbf{\Pi}$). The projected problem associated with $\mathbf{\Pi}$ is the Bayesian problem $\mathcal{B}_{\Pi}\equiv(E_\Pi,F,\mathbf{H}_\Pi,\mathbf{B}_\Pi,\mathbf{R}_{\Pi})$, where $E_\Pi=\{ \mathbf{\Pi x},\, \mathbf{x} \in E\} $, and $\mathbf{H}_\Pi$, $\mathbf{B}_\Pi$, and $\mathbf{R}_{\Pi}$ are the forward model, prior and observation error covariance matrices, respectively, in some basis of $E_\Pi$ and $F$.
 \end{defn}
 The observational error covariance matrix ($\mathbf{R}_{\Pi}$) of the projected problem can be expressed as a function of $\mathbf{B}$ and $\mathbf{\Pi}$:
 \begin{prop}[\textbf{Representativeness Error}]
  \label{agg_error}
The observational error covariance matrix $\mathbf{R}_{\Pi}$ for the projected Bayesian problem $\mathcal{B}_\mathbf{\Pi}=(E,F,\mathbf{H\Pi},\mathbf{B}_\Pi,\mathbf{R}_\Pi)$ can be expressed as the sum of the observational error covariance for the original Bayesian problem, $\mathbf{R}$, and a representativeness error, as follows:
  \begin{align}
  \label{eq:agg_err}
  \mathbf{R}_\Pi=\mathbf{R}+\mathbf{H}(\mathbf{B}+\mathbf{\Pi} \mathbf{B} \mathbf{\Pi}^T - \mathbf{B} \mathbf{\Pi}^T - \mathbf{\Pi} \mathbf{B})\mathbf{H}^T 
\end{align}
\end{prop}

 \begin{proof}
For the sake of clarity, below we distinguish between the control vector $\mathbf{x}$ and its associated increment $\delta \mathbf{x}\equiv\mathbf{x}-\mathbf{x}^b$, and the observation vector $\mathbf{y}$ and the corresponding \textit{innovation} $\mathbf{d}\equiv\mathbf{y}-\mathbf{Hx}^b$. In the incremental framework, the observational error covariance matrix can be written:
 \begin{align*}
 \mathbf{R}_\mathbf{\Pi} &\equiv \overline{(\mathbf{H}\mathbf{\Pi}\delta\mathbf{x}^t-\mathbf{d})(\mathbf{H}\mathbf{\Pi}\delta\mathbf{x}^t-\mathbf{d})^T} \\ \nonumber
&  =\overline{(\mathbf{H}\mathbf{\Pi}\mathbf{x}^t-\mathbf{H}\mathbf{\Pi}\mathbf{x}^b+\mathbf{H}\mathbf{x}^b-\mathbf{H}\mathbf{x}^t+\epsilon)(\mathbf{H}\mathbf{\Pi}\mathbf{x}^t-\mathbf{H}\mathbf{\Pi}\mathbf{x}^b+\mathbf{H}\mathbf{x}^b-\mathbf{H}\mathbf{x}^t+\epsilon)^T}   \nonumber
\end{align*}
Using the independence assumption between the errors in the observations and in the prior, one obtains:
\begin{align*}
&  \mathbf{R}_\mathbf{\Pi} = \mathbf{R}+\mathbf{H}(\mathbf{B}+\mathbf{\Pi} \mathbf{B} \mathbf{\Pi}^T - \mathbf{B} \mathbf{\Pi}^T - \mathbf{\Pi} \mathbf{B})\mathbf{H}^T
 \end{align*}
\end{proof}
Our goal is to find a projection that maximizes the DOFS or minimizes the posterior error covariance matrix of the Bayesian problem $\mathcal{B}_\Pi$, in some sense to be defined thereafter. In the following we describe a two-step approach to the optimal projection problem, wherein an appropriate decomposition is used to construct a class of projectors in which the optimal solutions must lie. This restriction to a particular class of projectors allows us to greatly simplify the problem, as we show in Section \ref{opt_proj}. In addition, the two-step method yields some interesting theoretical interpretations, and can be related to previous Bayesian dimension reduction approaches, as described in Sections \ref{2steps} and \ref{gen_change_var}.
 
\subsubsection{A Two-Step Approach}
\label{2steps}
The idea behind the projection approach is to solve a Bayesian problem of smaller dimension ($k$) than the original large-scale problem (i.e., one has $k<<n$), allowing fast (sometimes analytical) computation of its solution. In this Section, a factorization of rank-$k$ projectors is proposed to construct a Bayesian problem of dimension $k$ from which the solutions of the projected problem $\mathcal{B}_\Pi$ in the canonical basis of $E$ can be derived with simple transformations. Let us recall that any projector is defined by its null space and its range, and can be written:
\begin{align}
\label{def_proj}
\mathbf{\Pi}=\mathbf{I}\left( \mathbf{O}^T\mathbf{I}\right)^{-1}\mathbf{O}^T,
\end{align}
where $\mathbf{O}$ is a matrix whose columns form an orthonormal basis for the orthogonal of the null space of $\mathbf{\Pi}$, and $\mathbf{I}$ is a matrix whose columns span the range of $\mathbf{\Pi}$. In other words, $\mathbf{I}$ defines the subspace (of dimension $k$) onto which the Bayesian problem is projected, while $\mathbf{O}$ defines the direction of the projection. Another form for (\ref{def_proj}) can be derived as follows: 
  \begin{prop}[\textbf{Factorization of a Projector}]
  \label{eq_2steps}
  Any linear operator $\mathbf{\Pi}$ is a projector of rank $k$ if and only if it can be written as the product of two rank-$k$ matrices,  one of which is the left inverse of the other, i.e.: 
  \begin{align}
  \label{fact_1}
  \mathbf{\Pi}= \mathbf{\Gamma}^\star\mathbf{\Gamma},
\end{align}
where $ \mathbf{\Gamma}^\star$ and $\mathbf{\Gamma}$ two matrices of dimension $(n\times k)$ and $(k\times n)$, respectively, with maximum rank and:
  \begin{align}
    \label{fact_2}
 \mathbf{\Gamma}\mathbf{\Gamma}^\star=\mathbf{Id}_k
\end{align}
\end{prop}
 \begin{proof}
 Any projector $\mathbf{\Pi}$ of rank $k$ can be written $\mathbf{\Pi}=\mathbf{I}\left( \mathbf{O}^T\mathbf{I}\right)^{-1}\mathbf{O}^T$, where $\mathbf{I}$ (the range of $\mathbf{\Pi}$) and $\mathbf{O}$ (the orthogonal of the null space of $\mathbf{\Pi}$) are two matrices of rank $k$ and dimension $(n\times k)$. Defining $\mathbf{\Gamma}=\mathbf{O}^T$ and $\mathbf{\Gamma}^\star=\mathbf{I}(\mathbf{O}^T\mathbf{I})^{-1}$, we obtain $\mathbf{\Pi}=\mathbf{\Gamma}^\star\mathbf{\Gamma}$, with $\mathbf{\Gamma}$ and $\mathbf{\Gamma}^\star$ of dimension $(n\times k)$ and $(k\times n)$, respectively, and $\mathbf{\Gamma}\mathbf{\Gamma}^\star=\mathbf{Id}_k$. Now let us consider a linear operator $\mathbf{\Pi}=\mathbf{\Gamma}^\star\mathbf{\Gamma}$, with $\mathbf{\Gamma}\mathbf{\Gamma}^\star=\mathbf{Id}_k$. One has also $\mathbf{\Pi}=\mathbf{\Gamma}^\star (\mathbf{\Gamma}\mathbf{\Gamma}^\star)^{-1}\mathbf{\Gamma}$. Defining $\mathbf{O}=\mathbf{\Gamma}^T$ and $\mathbf{I}=\mathbf{\Gamma}^\star$, one sees that $\mathbf{\Pi}$ has the general form of a projector.
\end{proof}
\begin{rem}
In practice, $\mathbf{\Gamma}$ and $\mathbf{\Gamma}^\star$ can be derived from the range and the null space of the projection $\mathbf{\Pi}$ (see previous proof). Note that the decomposition $\mathbf{\Pi}= \mathbf{\Gamma}^\star\mathbf{\Gamma}$ is not unique. Indeed, $\mathbf{\Gamma}=(\mathbf{OP})^T$ and $\mathbf{\Gamma}^\star=\mathbf{IQ}((\mathbf{O}\mathbf{P})^T\mathbf{I}\mathbf{Q})^{-1}$ verify (\ref{fact_1}) and  (\ref{fact_2}) for all orthogonal matrices $\mathbf{P}$ and  $\mathbf{Q}$ (i.e., verifying $\mathbf{P}^T\mathbf{P}=\mathbf{Q}^T\mathbf{Q}=\mathbf{Id}$). This simply formalizes the fact that a projection is only defined by its null space and range, and is therefore basis-invariant.
\end{rem} 
The previous decomposition is now used to define a reduced Bayesian problem of dimension $k$ in order to be able to express the solution of the projected problem in the initial control space $E$.
 
 \begin{defn}[\textbf{Reduced Bayesian Problem}]
Let us consider a projected Bayesian problem defined by $\mathcal{B}_{\mathbf{\Pi}}\equiv(E_\Pi,F,\mathbf{H}_\Pi,\mathbf{B}_\Pi,\mathbf{R}_{\Pi})$. The \textit{reduced Bayesian problem} associated with  $\mathcal{B}_{\Pi}$ in the basis defined by the columns of $\mathbf{\Gamma}^\star$ is the problem $\mathcal{B}_\omega=(E_\omega,F,\mathbf{H\Gamma}^\star,\mathbf{\Gamma}\mathbf{B}\mathbf{\Gamma}^T,\mathbf{R}_{\Pi})$, where $E_\omega=\{\mathbf{\Gamma}\mathbf{x},\, \mathbf{x}\in E \}$ and $\mathbf{\Pi}=\mathbf{\Gamma}^\star\mathbf{\Gamma}$ is a rank-$k$ factorization of the projector $\mathbf{\Pi}$, as defined in Prop. \ref{eq_2steps}. 
\end{defn}

The couple $(\mathbf{\Gamma},\mathbf{\Gamma}^\star)$ characterizes the correspondence between the reduced control space $E_\omega$, on which the reduced problem $\mathcal{B}_\omega$ is defined, and the initial control space $E$. Note that the $k$ columns of $\mathbf{\Gamma}^\star$ form a basis for the subspace of $E$ on which the Bayesian problem is projected. Therefore, the vector $\mathbf{\Gamma x}$ represents the coordinates of $\mathbf{x}$ in the basis defined by $\mathbf{\Gamma}^\star$.
Using the posterior solution of the reduced Bayesian problem $\mathcal{B}_\omega$, one can provide an analytical solution for the projected problem $\mathcal{B}_{\mathbf{\Pi}}$ in the initial space $E$:

\begin{prop}[\textbf{Posterior Solution for a General Projection }]
 \label{proper:corresp}
Let us consider a projector $\mathbf{\Pi}=\mathbf{\Gamma}^\star\mathbf{\Gamma}$, factored according to Prop. \ref{eq_2steps}. We define the associated projected and reduced Bayesian problems, $\mathcal{B}_{\mathbf{\Pi}}\equiv(E_\Pi,F,\mathbf{H}_\Pi,\mathbf{B}_\Pi,\mathbf{R}_{\Pi})$and $\mathcal{B}_\omega=(E_\omega,F,\mathbf{H\Gamma}^\star,\mathbf{\Gamma}\mathbf{B}\mathbf{\Gamma}^T,\mathbf{R}_{\Pi})$, respectively.  One has:
\begin{eqnarray}
\label{eq:xa_proj}
\mathbf{x}^{a}_\Pi&=&\mathbf{\Gamma}^\star\mathbf{x}^{a}_\omega \\
\label{eq:Pa_proj}
\mathbf{P}^{a}_\Pi&=&{\mathbf{\Gamma}^\star} \mathbf{P}^{a}_\omega {\mathbf{\Gamma}^\star}^T\\
\label{eq:A_proj}
\mathbf{A}_\Pi&=&{\mathbf{\Gamma}^\star} {\mathbf{A}_\omega} {\mathbf{\Gamma}},
\end{eqnarray}
where  $\mathbf{x}^{a}_\Pi$ and $\mathbf{x}^{a}_\omega$ are the posterior mean of $\mathcal{B}_\Pi$ and $\mathcal{B}_\omega$, respectively, $\mathbf{P}^{a}_\Pi$ and $\mathbf{P}^{a}_\omega$ are the posterior error covariance matrices of $\mathcal{B}_\Pi$ and $\mathcal{B}_\omega$, respectively, and $\mathbf{A}_\Pi$ and $\mathbf{A}_\omega$ are the model resolution matrices of $\mathcal{B}_\Pi$ and $\mathcal{B}_\omega$, respectively. More precisely: 
\begin{eqnarray}
\label{eq:xa_proj1a}
\mathbf{x}^{a}_\Pi&=&\mathbf{\Pi}\mathbf{B}\mathbf{\Pi}^T\mathbf{H}^T ( \mathbf{H}\mathbf{\Pi}\mathbf{B}\mathbf{\Pi}^T\mathbf{H}^T+\mathbf{R}_\Pi )^{-1}\mathbf{d}\\
\label{eq:Pa_proj1a}
\mathbf{P}^{a}_\Pi&=&\mathbf{\Pi}  \mathbf{B}\mathbf{\Pi}^T- \mathbf{\Pi} \mathbf{B}\mathbf{\Pi}^T\mathbf{H}^T ( \mathbf{H}\mathbf{\Pi}\mathbf{B}\mathbf{\Pi}^T\mathbf{H}^T+\mathbf{R}_\Pi )^{-1}\mathbf{H} \mathbf{\Pi}\mathbf{B} \mathbf{\Pi}^T \\
\label{eq:A_proj1a}
\mathbf{A}_\Pi&=&\mathbf{\Pi}\mathbf{B}\mathbf{\Pi}^T\mathbf{H}^T ( \mathbf{H}\mathbf{\Pi}\mathbf{B}\mathbf{\Pi}^T\mathbf{H}^T+\mathbf{R}_\Pi )^{-1}\mathbf{H} \mathbf{\Pi}
\end{eqnarray}
Or equivalently:
\begin{eqnarray}
\label{eq:xa_proj1b}
\mathbf{x}^{a}_\Pi&=&{\mathbf{\Gamma}^\star}\left [ ( \mathbf{\Gamma} \mathbf{B} \mathbf{\Gamma}^T)^{-1} + {\mathbf{\Gamma}^\star}^T \mathbf{H}^T{\mathbf{R}_\Pi}^{-1}\mathbf{H}\mathbf{\Gamma}^\star \right ]^{-1} {\mathbf{\Gamma}^\star}^T \mathbf{H}^T{\mathbf{R}_\Pi}^{-1}\mathbf{d} \\
\label{eq:Pa_proj1b}
\mathbf{P}^{a}_\Pi&=&{\mathbf{\Gamma}^\star}\left [ ( \mathbf{\Gamma} \mathbf{B} \mathbf{\Gamma}^T)^{-1} + {\mathbf{\Gamma}^\star}^T \mathbf{H}^T{\mathbf{R}_\Pi}^{-1}\mathbf{H}\mathbf{\Gamma}^\star \right ]^{-1}{\mathbf{\Gamma}^\star}^T\\
\label{eq:A_proj1b}
\mathbf{A}_\Pi&=&\mathbf{\Pi}- {\mathbf{\Gamma}^\star}\left [ ( \mathbf{\Gamma} \mathbf{B} \mathbf{\Gamma}^T)^{-1} + {\mathbf{\Gamma}^\star}^T \mathbf{H}^T{\mathbf{R}_\Pi}^{-1}\mathbf{H}\mathbf{\Gamma}^\star \right ]^{-1}  ( \mathbf{\Gamma} \mathbf{B} \mathbf{\Gamma}^T)^{-1} \mathbf{\Gamma}
\end{eqnarray}
Here the solution $(\mathbf{x}^{a}_\Pi,\mathbf{P}^{a}_\Pi,\mathbf{A}_\Pi)$ is expressed in the canonical basis of the initial control space $E$.
\end{prop}
\begin{proof}
Let us define the two vector spaces $E_\mathbf{\Pi}=\{\mathbf{\Pi}\mathbf{x},\, \mathbf{x} \in E\}$ and $E_\omega=\{\mathbf{\Gamma} \mathbf{x},\, \mathbf{x} \in E\}$. One can easily verify that the following application defines an isomorphism between $E_\Pi$ and $E_\omega$:
\begin{align}
\begin{cases}
g: E_\Pi\rightarrow E_\omega  \\
\mathbf{x}_\mathbf{\Pi} \longmapsto \mathbf{x}_\omega = \mathbf{\Gamma} \mathbf{x}_\mathbf{\Pi} \\ \notag
g^{-1}: E_\omega \rightarrow  E_\mathbf{\Pi}\  \\
\mathbf{x}_\omega \longmapsto \mathbf{x}_\mathbf{\Pi} = \mathbf{\Gamma}^\star \mathbf{x}_\omega \notag
\end{cases}
\end{align}

 With this definition, $g$ associates any vector of $E_\Pi$ expressed in the canonical basis of $E$ to its coordinates in the basis formed by the columns of $\mathbf{\Gamma}^\star$, while $g^{-1}$ associates any vector of $E_\Pi$ expressed in the basis defined by $\mathbf{\Gamma}^\star$ to its coordinates in the canonical basis of $E$.  The prior error covariance matrix $\mathbf{B}_\Pi$ in the basis defined by $\mathbf{\Gamma}^\star$ is simply $\overline{(\mathbf{\Gamma}\mathbf{\Pi}(\mathbf{x}^t-\mathbf{x}^b)(\mathbf{\Gamma}\mathbf{\Pi}(\mathbf{x}^t-\mathbf{x}^b))^T}=\overline{(\mathbf{\Gamma}(\mathbf{x}^t-\mathbf{x}^b)(\mathbf{\Gamma}(\mathbf{x}^t-\mathbf{x}^b))^T}=\mathbf{\Gamma}\mathbf{B}\mathbf{\Gamma}^T$.  Likewise, the forward model $\mathbf{H}_\Pi$ expressed in the basis defined by $\mathbf{\Gamma}^\star$ is simply $\mathbf{H}\mathbf{\Gamma}^\star$, since: $\forall \mathbf{x}_\Pi \in E_\Pi,\, \mathbf{H}_\Pi\mathbf{x}_\Pi=\mathbf{H}\mathbf{\Pi}\mathbf{x}=\mathbf{H}\mathbf{\Gamma}^\star\mathbf{\Gamma}\mathbf{\Pi}\mathbf{x}=\mathbf{H}\mathbf{\Gamma}^\star\mathbf{\Gamma}\mathbf{x}=\mathbf{H}\mathbf{\Gamma}^\star\mathbf{x}_\omega$. Therefore, the Bayesian problems $\mathcal{B}_{\mathbf{\Pi}}\equiv(E_\Pi,F,\mathbf{H}_\Pi,\mathbf{B}_\Pi,,\mathbf{R}_{\mathbf{\Pi}})$ and $\mathcal{B}_\omega=(E_\omega,F,\mathbf{H\Gamma}^\star,\mathbf{\Gamma}\mathbf{B}\mathbf{\Gamma}^T,\mathbf{R}_{\Pi})$ are strictly equivalent ($\mathcal{B}_\omega$ is $\mathcal{B}_{\mathbf{\Pi}}$ expressed in the particular basis defined by $\mathbf{\Gamma}^\star$). Noting $\mathbf{x}^{a}_\omega$, $\mathbf{P}^{a}_\omega$ and $\mathbf{A}_\omega$ the posterior mean, posterior error covariance and model resolution matrix, respectively, of the Bayesian problem $\mathcal{B}_\omega$, one can directly obtain Eq. (\ref{eq:xa_proj}) and (\ref{eq:Pa_proj}) by applying $g^{-1}$. For Eq. (\ref{eq:A_proj}), we note that the model resolution matrix of the reduced problem $\mathcal{B}_\Pi$ in the canonical basis must verify: $\forall \mathbf{x} \in E,\, \mathbf{A}_\Pi \mathbf{\Pi}\mathbf{x}=\mathbf{\Gamma}^\star\mathbf{A}_\omega\mathbf{x}_\omega$. Using $\mathbf{x}_\omega=\mathbf{\Gamma}\mathbf{\Pi}\mathbf{x}$ in the right-hand side, one obtains $ \mathbf{A}_\Pi \mathbf{\Pi}\mathbf{x}=\mathbf{\Gamma}^\star\mathbf{A}_\omega\mathbf{\Gamma}\mathbf{\Pi}\mathbf{x}$, which by identification gives Eq. (\ref{eq:A_proj}). Formulas (\ref{eq:xa_proj1a})-(\ref{eq:A_proj1a}) and (\ref{eq:xa_proj1b})-(\ref{eq:A_proj1b}) are then obtained by replacing $(\mathbf{x}^{a}_\omega, \mathbf{P}^{a}_\omega,\mathbf{A}_\omega)$ in (\ref{eq:xa_proj})-(\ref{eq:A_proj}) using Eq. (\ref{eq2}), (\ref{eq5}), (\ref{eq8}), and Eq. (\ref{post_update2}), (\ref{eq4}), (\ref{eq6}), respectively.   
\end{proof}

\begin{rem}
As long as the dimension $p$ of the observation space $F$ allows for explicit construction and inversion of ($p\times p$) covariance matrices in $F$, formulas (\ref{eq:xa_proj1a})-(\ref{eq:A_proj1a}) or (\ref{eq:xa_proj1b})-(\ref{eq:A_proj1b}) can be used to compute the posterior mean and extract any column of the posterior error covariance or model resolution matrices for the projected problem. On the other hand, when $p$ is large, due to the presence of the representativeness error $\mathbf{R}_\Pi$ and the necessity to form the $(p \times p)$ matrix $\mathbf{H}(\mathbf{B}+\mathbf{\Pi} \mathbf{B} \mathbf{\Pi}^T - \mathbf{B} \mathbf{\Pi}^T - \mathbf{\Pi} \mathbf{B})\mathbf{H}^T$, it may not be practical to use either of the two formulations (\ref{eq:xa_proj1a})-(\ref{eq:A_proj1a}) or (\ref{eq:xa_proj1b})-(\ref{eq:A_proj1b}). Interestingly, as we will show in Section \ref{opt_proj}, one can define an optimal projection which has the remarkable property that it effectively avoids the need to know and evaluate the representativeness errors covariance matrix to compute the posterior solution.
 \end{rem}

\paragraph{Aggregations and Projections}
\label{rem:aggregation_proj}
In the multi-scale formalism presented by \citet{bocquet2011bayesian}, $\mathbf{\Gamma}$ and $\mathbf{\Gamma}^\star$ are called the \textit{aggregation} and \textit{prolongation} operators, respectively, and $\mathbf{\Gamma}\mathbf{\Gamma}^\star=\mathbf{Id}_k$ is an imposed stability condition. In their study the operator $\mathbf{\Gamma}$ consists of a weighted average of model grid cell parameters (e.g., atmospheric fluxes) and the objective is to solve an aggregated version of the initial Bayesian problem of smaller dimension, which corresponds to our reduced problem $\mathcal{B}_\omega$. Although Prop. \ref{proper:corresp} shows that the formulations for the aggregated problem $\mathcal{B}_\omega$ and the projected problem $\mathcal{B}_\Pi$ are theoretically equivalent, their interpretations are quite different. In the projection framework, the analysis is centered on choosing a subspace of $E$, $E_\Pi$, on which the Bayesian problem is solved, i.e., on the choice of  $\mathbf{\Gamma}^\star$, whose columns represent a basis of $E_\Pi$ (e.g., \citet{turner2015balancing}). On the other hand, in the aggregation framework, the problem focuses on the choice of an average operator to define an aggregated problem, i.e., on $\mathbf{\Gamma}$ . Likewise, in the projection framework the posterior solution is analyzed in $E_\Pi$, while in the aggregation framework the posterior solution for the aggregated control vector in $E_\omega$ is the meaningful quantity to interpret.  Note that, from Eq. (\ref{eq:A_proj}), the DOFS of the aggregated problem is the same as the DOFS of the associated projected problem expressed in the canonical basis of $E$, since $\mathrm{Tr}(\mathbf{A}_\Pi)=\mathrm{Tr}({\mathbf{\Gamma}^\star} {\mathbf{A}_\omega} {\mathbf{\Gamma}})=\mathrm{Tr}({\mathbf{A}_\omega} {\mathbf{\Gamma}}{\mathbf{\Gamma}^\star} )=\mathrm{Tr}({\mathbf{A}_\omega})$. In our analysis $\mathbf{\Gamma}$ is a general $(k\times n)$ operator, therefore we shall refer to it as a \textit{reduction} operator (instead of an aggregation operator) and refer to $\mathbf{\Gamma}^\star$ as a prolongation operator for the sake of consistency with the formalism of \citet{bocquet2011bayesian}.

\subsubsection{A Generalized Change of Variable for Linear Bayesian Problems}
\label{gen_change_var}
We now turn to the problem of optimizing the choice of the projection $\mathbf{\Pi}$. The factorization of the projection described in the previous section will be used to construct our optimal solution in two steps.  The first step consists, for a given reduction operator $\mathbf{\Gamma}$, of finding a prolongation operator $\mathbf{\Gamma}^\star$ that minimizes the representativeness error  $\mathbf{R}_{\Pi}$ of the projected problem. The following Theorem provides such an optimal prolongation operator $\mathbf{\Gamma}^\star$ as a function of $\mathbf{\Gamma}$:
\begin{thm}[\textbf{Optimal Prolongation}]
\label{thm:best_prolong}
For any reduction operator $\mathbf{\Gamma}$, there exists a prolongation operator $\mathbf{\Gamma}_{opt}^\star$ such that the representativeness error is minimum w.r.t. the L\"owner partial ordering. More specifically, one has:
\begin{align}
\label{best_prolongation}
&\forall \mathbf{\Gamma} \in \mathcal{M}_{k,n}(\mathbb{R}),\, \exists  {\mathbf{\Gamma}_{opt}^\star} \in \mathcal{M}_{n,k}(\mathbb{R}) \,|\, \forall {\mathbf{\Gamma}^\star}\in \mathcal{M}_{n,k}: 
\mathbf{R}_{\Pi_\text{opt}} \le \mathbf{R}_{\Pi}   ,
\end{align}
where $\mathbf{\Pi}=\mathbf{\Gamma}^\star\mathbf{\Gamma}$, $\mathbf{\Pi_\text{opt}}=\mathbf{\Gamma}_{opt}^\star\mathbf{\Gamma}$, $\mathcal{M}_{m,n}$ represents the space of $(m \times n)$ real matrices, and the symbol $\le$ denotes the L\"owner partial ordering within the set of real positive definite matrices.\newline
Moreover, one has:
\begin{eqnarray}
\label{best_prolong_eq}
\mathbf{\Gamma}^\star_{opt} =\mathbf{B}\mathbf{\Gamma}^T(\mathbf{\Gamma}\mathbf{B}\mathbf{\Gamma}^T)^{-1} 
\end{eqnarray}
\end{thm}
\begin{proof}
Let us rewrite the observational error covariance for the projected problem using the decomposition (\ref{eq_2steps}): 
\begin{align}
\mathbf{R}_\mathbf{\Pi} = \mathbf{R}+\mathbf{H}(\mathbf{B}+\mathbf{\Gamma}^\star\mathbf{\Gamma} \mathbf{B}\mathbf{\Gamma}^T{\mathbf{\Gamma}^\star}^T - \mathbf{B} \mathbf{\Gamma}^T{\mathbf{\Gamma}^\star}^T -\mathbf{\Gamma}^\star\mathbf{\Gamma}  \mathbf{B})\mathbf{H}^T
\end{align}
From Lemma (\ref{eq:ineq_proj}), it is clear that minimizing $\mathbf{R}_\mathbf{\Pi}$ is equivalent to minimizing the matrix $\mathbf{\Delta B}=\mathbf{B}+\mathbf{\Gamma}^\star\mathbf{\Gamma} \mathbf{B}\mathbf{\Gamma}^T{\mathbf{\Gamma}^\star}^T - \mathbf{B} \mathbf{\Gamma}^T{\mathbf{\Gamma}^\star}^T -\mathbf{\Gamma}^\star\mathbf{\Gamma}  \mathbf{B}$. Fixing $\mathbf{\Gamma}$, we note that the solution (best prolongation $\mathbf{\Gamma}_{opt}^\star$) to this minimization problem is also the Best Linear Unbiased Estimator (BLUE) of the following problem:
\begin{align*}
 \label{demo_blue}
&\mathrm{Arg} \min_{{\mathbf{\Gamma}^\star}} \mathrm{Tr} \overline{({\mathbf{x}} - \mathbf{x}^t)(\mathbf{x} - \mathbf{x}^t)^T}  ,\\
&\text{with } 
\begin{cases}
&\mathbf{x} = \mathbf{x}_b +{\mathbf{\Gamma}^\star} (\mathbf{y}-\mathbf{\Gamma} \mathbf{x}_b) \\ \notag
&\mathbf{y} = \mathbf{\Gamma}\mathbf{x}^t \\ \notag
&  \overline{({\mathbf{x}^b} - \mathbf{x}^t)(\mathbf{x}^b - \mathbf{x}^t)^T} = \mathbf{B}
\end{cases}
\end{align*}
The BLUE solution to this problem, and therefore the optimal prolongation operator $\mathbf{\Gamma}^\star$, is given by $\mathbf{\Gamma}_{opt}^\star=\mathbf{B}\mathbf{\Gamma}^T(\mathbf{\Gamma}\mathbf{B}\mathbf{\Gamma}^T)^{-1}$. The posterior error covariance matrix of the BLUE analysis is precisely $\Delta \mathbf{B}$, and it is minimum in the sense of the  L\"owner partial ordering among all linear estimator (i.e., among all prolongation operators) (e.g., \citet{isotalo2008blue}),  which proves (\ref{best_prolongation}).
\end{proof}
 
\begin{rem}
The optimal prolongation $\mathbf{\Gamma}_{opt}^\star$ was first proposed by \citet{bocquet2011bayesian}, where it was derived from a Bayesian perspective exploiting the prior information. In our approach this result is obtained simply by minimizing the representativeness error. 
\end{rem}
Theorem \ref{thm:best_prolong} is a strong optimality result, since the optimality w.r.t. the L\"owner partial ordering implies that the representativeness error is minimum in any direction of the observation space. This leads to the following important optimality result: 

\begin{cor}
\label{opt_inf_cont}
For any reduction operator $\mathbf{\Gamma}$, the prolongation operator $\mathbf{\Gamma}_{opt}^\star=\mathbf{B}\mathbf{\Gamma}^T(\mathbf{\Gamma}\mathbf{B}\mathbf{\Gamma}^T)^{-1} $ minimizes the Fisher measurement information matrix w.r.t. the L\"owner partial ordering, i.e.:
\begin{align}
&\forall \mathbf{\Gamma} \in \mathcal{M}_{k,n}(\mathbb{R}),\, \exists  {\mathbf{\Gamma}_{opt}^\star} \in \mathcal{M}_{n,k}(\mathbb{R}) \,|\, \forall {\mathbf{\Gamma}^\star}\in \mathcal{M}_{n,k}: 
\mathbf{H}^T\mathbf{R}_{\Pi_\text{opt}}^{-1}\mathbf{H} \le \mathbf{H}^T\mathbf{R}^{-1}_{\Pi} \mathbf{H}  ,
\end{align}
where $\mathbf{\Pi}=\mathbf{\Gamma}^\star\mathbf{\Gamma}$ and $\mathbf{\Pi_\text{opt}}=\mathbf{\Gamma}_{opt}^\star\mathbf{\Gamma}$.
\end{cor}
\begin{proof}
This follows by applying consecutively Lemmas (\ref{eq:ineq_inv}) and (\ref{eq:ineq_proj}) to (\ref{best_prolongation}).
\end{proof}
 Theorem \ref{thm:best_prolong} and the optimal prolongation (\ref{best_prolong_eq}) yield different interpretations depending on the application. In the context of aggregation (see \ref{rem:aggregation_proj}), once an aggregation operator $\mathbf{\Gamma}$ (e.g., a weighted average of model grid-cells) has been chosen, the optimal prolongation (\ref{best_prolong_eq}) should be constructed and used together with the (reduced) forward model $\mathbf{H\Gamma}^\star$. On the other hand, in the context of a low-rank projection, once a subspace for the range of the projector has been chosen, Eq. (\ref{best_prolong_eq}) imposes an optimal direction for the projection. More precisely, if the columns of the matrix $\mathbf{I}$ represent a basis for the range of the projector $\mathbf{\Pi}$, then an optimal direction is defined by $\mathbf{D} = \mathbf{I} - \mathbf{B}^{-1}\mathbf{I}$ (i.e., $\mathbf{O}= \mathbf{B}^{-1}\mathbf{I}$ in (\ref{def_proj})). As an example of application of those concepts, we note that, in \citet{turner2015balancing}, the computation of the Gaussian Mixture Model (GMM) basis defining the range of the projection is performed simultaneously with the computation of the direction of the projection, or equivalently, the operators $\mathbf{\Gamma^\star}$ and $\mathbf{\Gamma}$ are constructed all at once. The fact that the resulting projection does not belong to the class of optimal projection defined in Prop. \ref{canon_form} is revealed by the presence of suboptimal features, such as a non-trivial minimum in the total posterior error variance for a rank $k<n$ (see interactive discussion of \citet{turner2015balancing}). Based on our results, one approach to improve the method proposed by \citet{turner2015balancing} would be to use the computed Gaussian Mixture Model (GMM) basis as range for the projection, but to replace their reduction operator (i.e., their weight matrix $\mathbf{W}$) by one that corresponds to an optimal direction for the projection.

 The following Proposition provides another useful optimality result for interpreting the optimal prolongation $\mathbf{\Gamma}_{opt}^\star$:
 \begin{prop}
 \label{prop:geom_interp_prolong}
 For any given reduction operator $\mathbf{\Gamma}$, the associated optimal prolongation operator $\mathbf{\Gamma}_{opt}^\star$ defines a rank-$k$ projection $\mathbf{\Pi}_{opt}$ that minimizes, over all rank-$k$ projections, the Frobenius distance between any square-root of the prior error covariance and its projection, i.e.:
 \begin{align}
  \| \mathbf{L}-\mathbf{\Pi}_{opt}\mathbf{L} \|_F=\min_{\mathbf{\Pi} \in \mathcal{P} }\| \mathbf{L}-\mathbf{\Pi}\mathbf{L}\|_F,
 \end{align}
 where $\mathbf{\Pi}_{opt}=\mathbf{\Gamma}_{opt}^\star\mathbf{\Gamma}$, $\mathcal{P}=\{\mathbf{\Pi}\in\mathcal{M}_n |\,\mathbf{\Pi}^2=\mathbf{\Pi},\, \mathrm{rank}(\mathbf{\Pi})=k\}$ and $\mathbf{B}=\mathbf{L}\mathbf{L}^T$
 \end{prop}
 \begin{proof}
The proof is simply obtained by choosing a square-root $\mathbf{L}$ of $\mathbf{B}$ (i.e., $\mathbf{B}=\mathbf{LL}^T$ and noting that:
\begin{align*}
\Delta \mathbf{B}&=\mathbf{B}+\mathbf{\Pi} \mathbf{B} \mathbf{\Pi}^T - \mathbf{B} \mathbf{\Pi}^T - \mathbf{\Pi} \mathbf{B} \\
&=(\mathbf{L}-\mathbf{\Pi}\mathbf{L} )(\mathbf{L}-\mathbf{\Pi}\mathbf{L} )^T \\
 \end{align*}
Since $\Delta \mathbf{B}$ is minimum for the  L\"owner partial ordering, one has in particular:
\begin{align*}
\mathrm{Tr}\left [ (\mathbf{L}-\mathbf{\Pi}_{opt}\mathbf{L} )(\mathbf{L}-\mathbf{\Pi}_{opt}\mathbf{L} )^T \right ] &= \min_\Pi \mathrm{Tr}\left [ (\mathbf{L}-\mathbf{\Pi}\mathbf{L} )(\mathbf{L}-\mathbf{\Pi}\mathbf{L} )^T \right ] \\ 
\Longleftrightarrow \\ 
\| \mathbf{L}-\mathbf{\Pi}_{opt}\mathbf{L}  \|_F& = \min_\Pi \| \mathbf{L}-\mathbf{\Pi}\mathbf{L}  \|_F
\end{align*}

 \end{proof}
   Finally, a simple interpretation for the optimal couple $\left ( \mathbf{\Gamma}, \mathbf{\Gamma}_{opt}^\star \right )$ is possible based on the following results:

\begin{prop}[\textbf{Posterior Solution of the Reduced Bayesian Problem}]
\label{post_sol_reduc}
Let us define a Bayesian problem $\mathcal{B}=(E,F,\mathbf{H},\mathbf{B},\mathbf{R})$. Let us consider a reduction $\mathbf{\Gamma}$ and its optimal prolongation $\mathbf{\Gamma}^\star_{opt}$, and $\mathcal{B}_\omega=(E_\omega,F,\mathbf{H\Gamma}^\star_{opt},\mathbf{\Gamma}^T\mathbf{B}\mathbf{\Gamma},\mathbf{R}_{\Pi_{opt}})$ the associated reduced Bayesian problem. One has:
\begin{eqnarray}
\label{eq:omega_mean}
\mathbf{x}^{a}_\omega&=&\mathbf{\Gamma} {\mathbf{x}^a} \\
\label{eq:omega_post}
\mathbf{P}^{a}_\omega &=& \mathbf{\Gamma} {\mathbf{P}^{a}} {\mathbf{\Gamma}}^T\\
\label{eq:omega_model_res}
\mathbf{A}_\omega&=&\mathbf{\Gamma} {\mathbf{A}} {\mathbf{\Gamma}^\star_{opt}},
\end{eqnarray}
where $\mathbf{x}^{a}$, $\mathbf{P}^{a}$ and $\mathbf{A}$ are the posterior mean, posterior error covariance and model resolution matrix (respectively) of $\mathcal{B}$, and $\mathbf{x}^{a}_\omega$, $\mathbf{P}^{a}_\omega$ and $\mathbf{A}_\omega$ are the posterior mean, posterior error covariance and model resolution matrix (respectively) of $\mathcal{B}_\omega$.
\end{prop}
\begin{proof}
Formulas (\ref{eq:omega_mean})-(\ref{eq:omega_model_res}) are obtained by using the optimality properties $\mathbf{\Pi}_{opt}\mathbf{B}\mathbf{\Pi}_{opt}^T=\mathbf{\Pi}_{opt}\mathbf{B}=\mathbf{B}\mathbf{\Pi}_{opt}^T$ and the (resulting) invariance of the innovation statistics $(\mathbf{HBH}^T+\mathbf{R})$ in formulas (\ref{eq:xa_proj1a})-(\ref{eq:A_proj1a}).
\end{proof}
Similar formulas can be obtained for the solution of the projected Bayesian problem in the canonical basis of $E$:
\begin{cor}[\textbf{Posterior Solution of the Projected Bayesian Problem }]
\label{post_sol_projec}
Let us define a Bayesian problem $\mathcal{B}=(E,F,\mathbf{H},\mathbf{B},\mathbf{R})$. Let us consider a reduction $\mathbf{\Gamma}$ and its optimal prolongation $\mathbf{\Gamma}^\star_{opt}$, and the associated projector $\mathbf{\Pi}_{opt}=\mathbf{\Gamma}^\star_{opt}\mathbf{\Gamma}$. One has:
\begin{eqnarray}
\label{eq:opt_reduc_post1}
\mathbf{x}^{a}_{\Pi_{opt}}&=&\mathbf{\Pi}_{opt}\mathbf{x}^{a} \\
\label{eq:opt_reduc_post2}
\mathbf{P}^{a}_{\Pi_{opt}}&=&\mathbf{\Pi}_{opt}  \mathbf{P}^{a}\mathbf{\Pi}_{opt}^T\\
\label{eq:opt_reduc_post3}
\mathbf{A}_{\Pi_{opt}}&=&\mathbf{\Pi}_{opt}  {\mathbf{A}}\mathbf{\Pi}_{opt},
\end{eqnarray}
where $\mathbf{x}^{a}$, $\mathbf{P}^{a}$ and $\mathbf{A}$ are the posterior mean, posterior error covariance and model resolution matrix (respectively) of $\mathcal{B}$, and $\mathbf{x}^{a}_{\Pi_{opt}}$, $\mathbf{P}^{a}_{\Pi_{opt}}$ and $\mathbf{A}_{\Pi_{opt}}$ are the posterior mean, posterior error covariance and model resolution matrix (respectively) of the projected Bayesian problem $\mathcal{B}_{\Pi_{opt}}\equiv(E_{\Pi_{opt}},F,\mathbf{H}_{\Pi_{opt}},\mathbf{B}_{\Pi_{opt}},\mathbf{R}_{\Pi_{opt}})$ in the canonical basis of $E$.
\end{cor}
\begin{proof}
Formulas (\ref{eq:opt_reduc_post1})-(\ref{eq:opt_reduc_post3}) are obtained simply by applying (\ref{eq:xa_proj})-(\ref{eq:A_proj}) to (\ref{eq:omega_mean})-(\ref{eq:omega_model_res}).
\end{proof}
In other words, if the projector is of the form $\mathbf{\Pi}_{opt}=\mathbf{B}\mathbf{\Gamma}^T(\mathbf{\Gamma}^T\mathbf{B}\mathbf{\Gamma})^{-1}\mathbf{\Gamma} $, the solution of the projected Bayesian problem is simply the projection of the solution of the initial Bayesian problem. From (\ref{post_sol_reduc}), it is clear that the couple $(\mathbf{\Gamma},\mathbf{\Gamma}^\star_{opt})$ can be interpreted as a generalized change of variable for the solutions of linear Bayesian problems, where the transformation $\mathbf{\Gamma}$ can be non-invertible and $\mathbf{\Gamma}^\star_{opt}$ defines a right inverse for $\mathbf{\Gamma}$. It is straightforward to verify that in the case where $\mathbf{\Gamma}$ is invertible $\mathbf{\Gamma}^\star_{opt}=\mathbf{\Gamma}^{-1}$.

\begin{rem}
\label{rem:proj_pb}
It is interesting to note that the posterior solution of the projected Bayesian problem $\mathcal{B}_{\Pi_{opt}}\equiv(E_{\Pi_{opt}},F,\mathbf{H}_{\Pi_{opt}},\mathbf{B}_{\Pi_{opt}},\mathbf{R}_{\Pi_{opt}})$ is also the posterior solution of the Bayesian problem $\mathcal{B}_{\mathbf{H}_{\Pi_\textrm{opt}}}\equiv(E,F,\mathbf{H}\mathbf{\Pi}_\text{opt},\mathbf{B},\mathbf{R}_{\mathbf{\Pi}_\text{opt}})$, as one can see by considering an optimal prolongation in formulas (\ref{eq:xa_proj1a})-(\ref{eq:A_proj1a}), and using $\mathbf{H}\mathbf{\Pi}\mathbf{B}\mathbf{\Pi}^T\mathbf{H}^T+\mathbf{R}_\Pi=\mathbf{HBH}^T+\mathbf{R}$ and $\mathbf{\Pi B \Pi}^T=\mathbf{\Pi B}=\mathbf{B\Pi}^T$. The Bayesian problem $\mathcal{B}_{\mathbf{H}_{\Pi_\textrm{opt}}}$ has similar prior error covariance as the original problem $\mathcal{B}$, but corresponds to a forward model for which the modes are filtered by the projection $\mathbf{\Pi}$. Note that $\mathcal{B}_{\Pi_{opt}}$ and $\mathcal{B}_{\mathbf{H}_{\Pi_\textrm{opt}}}$ do not in general define the same Bayesian problem. The equality between the posterior solutions of those two problems only holds when an optimal prolongation is used.
\end{rem}

 \subsubsection{An Optimal Projection}
 \label{opt_proj}
The results established in Section \ref{gen_change_var}, and in particular Theorem \ref{thm:best_prolong}, can be used to simplify the problem of defining an optimal low-rank projection for the Bayesian problem $\mathcal{B}_{\mathbf{\Pi}}\equiv(E_\Pi,F,\mathbf{H}_\Pi,\mathbf{B}_\Pi,\mathbf{R}_{\mathbf{\Pi}})$. Indeed, one can restrict our search to the class of projections associated with optimal prolongations, using (\ref{best_prolong_eq}). We first note that this class of projections can be redefined in a simpler form:
\begin{prop}[\textbf{Canonical Form of Projections}]
 \label{canon_form}
Let us define the class of projections of rank $k$ associated with optimal prolongations $\mathcal{P}_{opt}\equiv\{\mathbf{B}\mathbf{\Gamma}^T(\mathbf{\Gamma}\mathbf{B}\mathbf{\Gamma}^T)^{-1}\mathbf{\Gamma},\,\mathbf{\Gamma}\in \mathcal{M}_{k,n} \}$. One has: 
\begin{align}
\label{canon_form2}
\mathcal{P}_{opt}=\{\mathbf{B}^{1/2}\mathbf{U}\mathbf{U}^T\mathbf{B}^{-1/2},\, \mathbf{U}^T\mathbf{U}=\mathbf{Id}_k \text{ and } \mathbf{U}\in \mathcal{M}_{k,n} \}
\end{align}
\end{prop}
\begin{proof}
To prove that $ \{\mathbf{B}\mathbf{\Gamma}^T(\mathbf{\Gamma}\mathbf{B}\mathbf{\Gamma}^T)^{-1}\mathbf{\Gamma},\,\mathbf{\Gamma}\in \mathcal{M}_{k,n} \}\subset\{\mathbf{B}^{1/2}\mathbf{U}\mathbf{U}^T\mathbf{B}^{-1/2},\, \mathbf{U}^T\mathbf{U}=\mathbf{Id}_k\} $, we define $\mathbf{U}=\mathbf{B}^{1/2}\mathbf{\Gamma}^T(\mathbf{\Gamma}\mathbf{B}\mathbf{\Gamma}^T)^{-1/2}$, and note that $\mathbf{B}^{1/2}\mathbf{U}\mathbf{U}^T\mathbf{B}^{-1/2}=\mathbf{B}\mathbf{\Gamma}^T(\mathbf{\Gamma}\mathbf{B}\mathbf{\Gamma}^T)^{-1}\mathbf{\Gamma}$ and $\mathbf{U}\mathbf{U}^T=\mathbf{Id}_k$. To prove that $\{ \mathbf{B}^{1/2}\mathbf{U}\mathbf{U}^T\mathbf{B}^{-1/2},\, \mathbf{U}^T\mathbf{U}=\mathbf{Id}_k\} \subset  \{\mathbf{B}\mathbf{\Gamma}^T(\mathbf{\Gamma}\mathbf{B}\mathbf{\Gamma}^T)^{-1}\mathbf{\Gamma},\,\mathbf{\Gamma}\in \mathcal{M}_{k,n} \} $, we define $\mathbf{\Gamma}=\mathbf{U}^T\mathbf{B}^{1/2}$, which verifies $\mathbf{B}\mathbf{\Gamma}^T(\mathbf{\Gamma}\mathbf{B}\mathbf{\Gamma}^T)^{-1}\mathbf{\Gamma}=\mathbf{B}^{1/2}\mathbf{U}\mathbf{U}^T\mathbf{B}^{-1/2}$.
\end{proof}
\begin{rem}
Prop. \ref{canon_form} allows a simple (statistical) interpretation for the optimal projections, that is, they correspond to a whitening transformation ($\mathbf{B}^{-1/2}$) followed by a orthogonal projection ($\mathbf{U}\mathbf{U}^T$) onto a rank-$k$ subspace of $E$, and a coloring transformation ($\mathbf{B}^{1/2}$) that recovers the prior error covariances.
\end{rem}
We can now state one of the main results of this paper, which provides an optimal low-rank projection for the linear Bayesian problem. The following Lemma provides a basis of eigenvectors for the model resolution matrix, which, together with the canonical form of Prop. \ref{canon_form}, allows for construction of a projected Bayesian problem $\mathcal{B}_{\mathbf{\Pi}_{opt}}$ with maximum DOFS. Note that in this paper eigenvectors shall always be presented in descending order of their corresponding eigenvalues. 

\begin{lem}[\textbf{Diagonalization of the Model Resolution Matrix}]
\label{diag_avk}
 Let us consider the following eigenvalue decomposition:
 \begin{align}
  \label{eq:SVD_max_dof}
\mathbf{Q}_\mathrm{dof}\equiv\mathbf{B}^{1/2}\mathbf{H}^T(\mathbf{HBH}^T+\mathbf{R})^{-1}\mathbf{H}\mathbf{B}^{1/2}= \mathbf{V}^T\mathbf{\Sigma}\mathbf{V},
 \end{align}
where $\mathbf{V}$ is the matrix whose columns are the eigenvectors of $\mathbf{Q}_\mathrm{dof}$ $\{ \mathbf{v}_i,\, i=1,...,n \}$, and $\mathbf{\Sigma}$ is a diagonal matrix whose elements are the eigenvalues of $\mathbf{Q}_\mathrm{dof}$  $\{ \sigma_i,\, i=1,...,n \}$. The vectors $\{\mathbf{B}^{1/2}\mathbf{v}_i,\, i=1,...,n \}$ form a basis of eigenvectors for the model resolution matrix $\mathbf{A}$.
\end{lem}
\begin{proof}
Using (\ref{eq8}), one can write $\mathbf{A}=\mathbf{B}^{1/2}\left ( \mathbf{B}^{1/2}\mathbf{H}^T(\mathbf{HBH}^T+\mathbf{R})^{-1}\mathbf{H}\mathbf{B}^{1/2}\right )\mathbf{B}^{-1/2}$. Therefore, $\mathbf{A}=\mathbf{B}^{1/2}\mathbf{V}^T\mathbf{\Sigma}\mathbf{V}\mathbf{B}^{-1/2}=\mathbf{B}^{1/2}\mathbf{V}^T\mathbf{\Sigma}(\mathbf{B}^{1/2}\mathbf{V}^T)^{-1}$, and the vectors $\{\mathbf{B}^{1/2}\mathbf{v}_i,\, i=1,...,n \}$ diagonalize $\mathbf{A}$.
\end{proof}
\begin{thm}[\textbf{Maximum-DOFS Projection}] 
\label{thm:opt_proj}
 Let us define $\mathbf{V}_k$ the matrix whose columns are the first $k$ eigenvectors of $\mathbf{Q}_\mathrm{dof}$ $\{ \mathbf{v}_i,\, i=1,...,k \}$, and 
 let us define the projector:
\begin{align}
\label{opt_proj_dof}
 \mathbf{\Pi}_\mathrm{dof}=\mathbf{B}^{1/2}\mathbf{V}_k\mathbf{V}_k^T\mathbf{B}^{-1/2}
 \end{align}
 The projector $ \mathbf{\Pi}_\mathrm{dof}$ maximizes the DOFS of the projected Bayesian problem $\mathcal{B}_{\Pi}=(E,F,\mathbf{H}\mathbf{\Pi},\mathbf{B},\mathbf{R}_{\Pi})$ among all projectors $\mathbf{\Pi}$ of maximum rank $k$, i.e.:

\begin{align}
\label{eq:max_dof}
&\forall \mathbf{\Pi} \in \mathcal{P}, \, \mathrm{Tr}(\mathbf{A}_{\Pi_\mathrm{dof}}) \ge   \mathrm{Tr}(\mathbf{A}_{\Pi}),
\end{align}
where $\mathbf{A}_{\mathbf{\Pi}}$ is the model resolution matrix associated with the problem $\mathcal{B}_{\Pi}$. 
\end{thm}
\begin{proof}
Let us consider a projection associated with an optimal prolongation, i.e., of the form $\mathbf{\Pi}=\mathbf{B}^{1/2}\mathbf{U}\mathbf{U}^T\mathbf{B}^{-1/2}$, with $\mathbf{U}$ an orthogonal matrix (see (\ref{canon_form2})). Replacing $\mathbf{\Pi}$ by this expression in (\ref{eq:A_proj1a}) yields $\mathbf{A}_\Pi=\mathbf{B}^{1/2}\mathbf{U}\mathbf{U}^T\mathbf{B}^{1/2}\mathbf{H}^T ( \mathbf{H}\mathbf{B}^T\mathbf{H}^T+\mathbf{R})^{-1}\mathbf{H}\mathbf{B}^{1/2}\mathbf{U}\mathbf{U}^T\mathbf{B}^{-1/2}$. Using the property of invariance of the trace under matrix permutation and the fact that $\mathbf{U}^T\mathbf{U}=\mathbf{Id}_k$, one obtains $\mathrm{Tr}(\mathbf{A}_\Pi)=\mathrm{Tr}(\mathbf{U}^T\mathbf{B}^{1/2}\mathbf{H}^T ( \mathbf{H}\mathbf{B}^T\mathbf{H}^T+\mathbf{R})^{-1}\mathbf{H}\mathbf{B}^{1/2}\mathbf{U})$. By Lemma \ref{max_tr}, the maximum of $\mathrm{Tr}(\mathbf{A}_\Pi)$ is obtained for $\mathbf{U}=\mathbf{V}$, where $\mathbf{V}$ is the matrix whose columns are the first $k$ eigenvectors of $\mathbf{Q}_\mathrm{dof}$, $\{ \mathbf{v}_i,\, i=1,...,k \}$, which proves (\ref{eq:max_dof}).
\end{proof}
 \begin{rem}
As suggested in \ref{info_content}, another criteria to optimize the projection is to minimize the total error variance (i.e., $\mathrm{Tr}(\mathbf{P}^a_{\Pi}$)). However, unlike the maximum-DOFS projection, this minimum-error projection does not have a simple analytical expression, which prevents its efficient computation. In Section \ref{link_lr_approx}, we present alternative optimal approximations of the posterior error covariance matrix whose Frobenius distance to the true posterior error covariance matrix is minimal and whose total error variance is closest to true total error variance.
 \end{rem}
 \begin{rem}
The optimal projection defined in Thm. \ref{thm:opt_proj} has been proposed by \citet{spantini2015optimal}. We note that in their study the projected problem is defined as $\mathcal{B}_{\mathbf{\Pi}_{opt}}=(E,F,\mathbf{H\Pi}_{opt},\mathbf{B},\mathbf{R})$. However, as discussed in the present study, it is necessary to include a representativeness error, i.e., to use $\mathbf{R}_{\mathbf{\Pi}_{opt}}$ instead of $\mathbf{R}$ when defining the projected Bayesian problem. \citet{spantini2015optimal} overlooked this issue in their analysis, which is taken into account in our proofs.
 \end{rem}
 Once the truncated eigendecomposition of $\mathbf{Q}_\mathrm{dof}$ is available, the posterior mean and posterior error covariance of the projected problem can be explicitly expressed as a function of the first $k$ eigenvectors and eigenvalues. Note that in its current form $\mathbf{Q}_\mathrm{dof}$ requires the inversion of a potentially high-dimensional $p\times p$ matrix. In fact, one can circumvent this difficulty by noting that the eigendecomposition of $\mathbf{Q}_\mathrm{dof}$ can be efficiently obtained from the eigendecomposition of an auxiliary matrix called the \textit{prior-preconditioned Hessian}. The following properties establish the formulas to compute the maximum-DOFS solution based on that improved implementation:
  \begin{prop}[\textbf{Posterior Solution of the maximum-DOFS Projection}] 
  \label{simp_dof}
  Let us define the \textit{prior-preconditioned Hessian} $\widehat{\mathbf{H}_p} \equiv \mathbf{B}^{1/2}\mathbf{H}^T\mathbf{R}^{-1}\mathbf{H}\mathbf{B}^{1/2}$ and its eigenvalue decomposition $\widehat{\mathbf{H}_p}=\mathbf{V'}^T\mathbf{\Lambda}\mathbf{V'}$. One has:
\begin{align}
\label{v_eq1}
&\mathbf{V}=\mathbf{V}'  \\
\label{v_eq2}
&\mathbf{\Sigma}=\mathbf{\Lambda}\left( \mathbf{Id}+\mathbf{\Lambda}\right)^{-1},
\end{align}
  where $\mathbf{Q}_\mathrm{dof}=\mathbf{V}^T\mathbf{\Sigma}\mathbf{V}$ is the eigendecomposition of $\mathbf{Q}_\mathrm{dof}$. Moreover, the solution of the maximum-DOFS projection can be expressed as:
  \begin{align}
  \label{eq:update_max_dof_mean_new}
&\mathbf{x}^a_{\mathbf{\Pi}_\mathrm{dof}}=\mathbf{B}^{1/2} \left(\sum_{i=1}^{k} \lambda_i^{1/2}(1+\lambda_i)^{-1}\mathbf{v}_i\mathbf{w}^T_i\right)\mathbf{R}^{-1/2}\mathbf{d} \\
\label{eq:update_max_dof_cov_new}
&\mathbf{P}^a_{\mathbf{\Pi}_\mathrm{dof}}=\mathbf{B}^{1/2}\sum_{i=1}^{k}\left( 1+\lambda_i \right)^{-1}\mathbf{v}_i\mathbf{v}_i^T\mathbf{B}^{1/2} \\
\label{eq:update_max_dof_mod_res_new}
&\mathbf{A}_{\mathbf{\Pi}_\mathrm{dof}}=\mathbf{B}^{1/2}\left(\sum_{i=1}^{k} \lambda_i(1+\lambda_i)^{-1}\mathbf{v}_i\mathbf{v}^T_i\right)\mathbf{B}^{-1/2} ,
\end{align}
where $\mathbf{w}_i=\mathbf{R}^{-1/2}\mathbf{H}\mathbf{B}^{1/2}\mathbf{v}_i$ and the $\{\lambda_i,\, i=1,...,n\}$ are the diagonal elements of $\mathbf{\Lambda}$.
\end{prop}
\begin{proof}
Let us first prove (\ref{v_eq1})-(\ref{v_eq2}). The matrix $\mathbf{Q}_\mathrm{dof}$ can be rewritten:
\begin{align}
\mathbf{Q}_\mathrm{dof}&=\mathbf{B}^{1/2}\mathbf{H}^T\mathbf{R}^{-1/2}(\mathbf{R}^{-1/2}\mathbf{HBH}^T\mathbf{R}^{-1/2}+\mathbf{Id})^{-1}\mathbf{R}^{-1/2}\mathbf{HB}^{1/2} \\ 
\label{eq_Q2}
                                        &=\mathbf{B}^{1/2}\mathbf{H}^T\mathbf{R}^{-1/2} \mathbf{W}^T ( \mathbf{Id}-\mathbf{\Lambda}(\mathbf{Id}+\mathbf{\Lambda})^{-1})\mathbf{W}\mathbf{R}^{-1/2}\mathbf{HB}^{1/2} ,
\end{align}
where $\mathbf{B}^{1/2}\mathbf{H}^T\mathbf{R}^{-1/2}=\mathbf{V}^T\mathbf{\Lambda}^{1/2}\mathbf{}\mathbf{W}$ is the SVD of the square-root of the prior-preconditioned matrix $\widehat{\mathbf{H}_p}$ and the Shermann-Morrison-Woodbury formula was applied to derive $(\mathbf{R}^{-1/2}\mathbf{HBH}^T\mathbf{R}^{-1/2}+\mathbf{Id})^{-1}=\mathbf{W}^T ( \mathbf{Id}-\mathbf{\Lambda}(\mathbf{Id}+\mathbf{\Lambda})^{-1})\mathbf{W}$. Replacing the square-root of $\mathbf{B}^{1/2}\mathbf{H}^T\mathbf{R}^{-1/2}$ by its SVD in (\ref{eq_Q2}) and using the fact that $\mathbf{W}\mathbf{W}^T=\mathbf{Id}$, one obtains: 
\begin{align*}
\mathbf{Q}_\mathrm{dof}&=\mathbf{V}^T\mathbf{\Lambda}\left( \mathbf{I+\Lambda}\right)^{-1}\mathbf{V}
\end{align*}

To prove (\ref{eq:update_max_dof_mean_new})-(\ref{eq:update_max_dof_mod_res_new}), we first use formulas (\ref{post_update2}), (\ref{eq5}), and (\ref{eq8}) for $\mathbf{x}^a$, $\mathbf{P}^a$ and $\mathbf{A}$, respectively, and substitute $\mathbf{Q}_\mathrm{dof}$ in them to obtain the following expressions:
\begin{align}
\mathbf{x}^{a}&=\mathbf{B}^{1/2}(\mathbf{Id}-\mathbf{Q}_\mathrm{dof})\mathbf{B}^{1/2}\mathbf{H}^T\mathbf{R}^{-1}\mathbf{d} \\
\mathbf{P}^{a}&=\mathbf{B}^{1/2}(\mathbf{Id}-\mathbf{Q}_\mathrm{dof})\mathbf{B}^{1/2} \\
\mathbf{A}&=\mathbf{B}^{1/2}\mathbf{Q}_\mathrm{dof}\mathbf{B}^{-1/2}
\end{align}
We then substitute those expressions in formulas (\ref{eq:opt_reduc_post1})-(\ref{eq:opt_reduc_post3}) and replace $\mathbf{\Pi}$ by its optimal solution $\mathbf{\Pi}_\mathrm{dof}=\mathbf{B}^{1/2}\mathbf{V}_k\mathbf{V}_k^T\mathbf{B}^{-1/2}$, which yields:
\begin{align}
\label{xa_pi_dof}
\mathbf{x}^{a}_{\Pi_\mathrm{dof}}&=\mathbf{B}^{1/2}\mathbf{V}_k\mathbf{V}_k^T(\mathbf{Id}-\mathbf{Q}_\mathrm{dof})\mathbf{B}^{1/2}\mathbf{H}^T\mathbf{R}^{-1}\mathbf{d} \\
\label{pa_pi_dof}
\mathbf{P}^{a}_{\Pi_\mathrm{dof}}&=\mathbf{B}^{1/2}\mathbf{V}_k\mathbf{V}_k^T(\mathbf{Id}-\mathbf{Q}_\mathrm{dof})\mathbf{V}_k\mathbf{V}_k^T\mathbf{B}^{1/2} \\
\label{A_pi_dof}
\mathbf{A}_{\Pi_\mathrm{dof}}&=\mathbf{B}^{1/2}\mathbf{V}_k\mathbf{V}_k^T\mathbf{Q}_\mathrm{dof}\mathbf{V}_k\mathbf{V}_k^T\mathbf{B}^{-1/2}
\end{align}
Noting $\mathbf{B}^{1/2}\mathbf{H}^T\mathbf{R}^{-1}=\mathbf{V}^T\mathbf{\Lambda}^{1/2}\mathbf{}\mathbf{W}\mathbf{R}^{-1/2}$ in (\ref{xa_pi_dof}), and replacing $\mathbf{Q}_\mathrm{dof}$ by its SVD in (\ref{xa_pi_dof})-(\ref{A_pi_dof}), one obtains the desired formulas (\ref{eq:update_max_dof_mean_new})-(\ref{eq:update_max_dof_mod_res_new}).
\end{proof}
Note that an alternative formula can be derived for Eq. (\ref{eq:update_max_dof_mean_new}), which has the advantage that it does not require computation of the singular vectors $\{\mathbf{w}_i\}$:

\begin{prop}[\textbf{Alternative Formulation for Posterior Mean of Maximum-DOFS Projection}]
\label{lem_CG_optimal_equiv}
Using the previous notations, let $\{(\mathbf{v}_i,\lambda_i),\, i=1,...,k\}$ be the first $k$ eigenpairs of the prior-preconditioned Hessian $\widehat{\mathbf{H}}_p$. The posterior mean of the rank-$k$ maximum-DOFS projection can be expressed as:
\begin{align}
\label{pg_equiv}
\mathbf{x}^a_{\mathbf{\Pi}_\mathrm{dof}}=\mathbf{B}^{1/2}\left [ \sum_{i=1}^{k}\left( 1+\lambda_i \right)^{-1}\mathbf{v}_i\mathbf{v}_i^T\right ]\mathbf{B}^{1/2}\mathbf{H}^T\mathbf{R}^{-1} \mathbf{d} 
\end{align}
\end{prop}
\begin{proof}
One has:
\begin{align}
\mathbf{B}^{1/2}\left [ \sum_{i=1}^{k}\left( 1+\lambda_i \right)^{-1}\mathbf{v}_i\mathbf{v}_i^T\right ]\mathbf{B}^{1/2}\mathbf{H}^T\mathbf{R}^{-1} \mathbf{d}& = \mathbf{B}^{1/2}\left [ \sum_{i=1}^{k}\left( 1+\lambda_i \right)^{-1}\mathbf{v}_i\mathbf{v}_i^T\right ]\mathbf{B}^{1/2}\mathbf{H}^T\mathbf{R}^{-1/2}\mathbf{R}^{-1/2} \mathbf{d} \\
&= \mathbf{B}^{1/2}\left [ \sum_{i=1}^{k}\left( 1+\lambda_i \right)^{-1}\mathbf{v}_i\mathbf{v}_i^T\right ]  \mathbf{V}\mathbf{\Lambda}^{1/2} \mathbf{W}^T \mathbf{R}^{-1/2}  \mathbf{d} \\
&= \mathbf{B}^{1/2} \left(\sum_{i=1}^{k}\lambda_i^{1/2}(1+\lambda_i)^{-1}\mathbf{v}_i\mathbf{w}_i^T\right)\mathbf{R}^{-1/2}\mathbf{d} ,
\end{align}
with: 
\begin{align*}
\mathbf{w}_i=\mathbf{R}^{-1/2}\mathbf{H}\mathbf{B}^{1/2}\mathbf{v}_i
\end{align*}
From (\ref{eq:update_max_dof_mean_new}), we obtain equality (\ref{pg_equiv}).

\end{proof}

Finally, from Prop. \ref{simp_dof}, one also obtained the following useful result:
\begin{cor} 
\label{cor:dofs}
The DOFS of the rank-$k$ projected Bayesian problem $\mathcal{B}_{\mathbf{\Pi}_{opt}}$ with maximum DOFS is the sum of the first $k$ eigenvalues of the model resolution matrix, i.e.:
\begin{align}
 \mathrm{Tr}(\mathbf{A}_{\Pi_\mathrm{dof}})=\sum_{i=1}^{k} \lambda_i(1+\lambda_i)^{-1}
\end{align}
 
\end{cor}
 \begin{proof}
 \begin{align*}
 \mathrm{Tr}(\mathbf{A}_{\Pi_\mathrm{dof}})&=\mathrm{Tr}(\mathbf{B}^{1/2}\left(\sum_{i=1}^{k} \lambda_i(1+\lambda_i)^{-1}\mathbf{v}_i\mathbf{v}^T_i\right)\mathbf{B}^{-1/2}) \\
 &=\mathrm{Tr}(\mathbf{B}^{-1/2}\mathbf{B}^{1/2}\left(\sum_{i=1}^{k} \lambda_i(1+\lambda_i)^{-1}\mathbf{v}_i\mathbf{v}^T_i\right) \\
 &=\sum_{i=1}^{k} \lambda_i(1+\lambda_i)^{-1}
 \end{align*}
 \end{proof}

\subsubsection{Interpretation}

\paragraph{Information Content of Subspaces}
Using our previous analysis, a natural generalization of the concept of information content to subspaces of a linear Bayesian problem can be derived. Given a subspace of dimension $k$ defined by the basis $\{\mathbf{r}_i, i = 1, ..., k\}$ and its associated matrix column $\mathbf{R}$, let us define the projection $\mathbf{\Pi}_R$ with range $\mathbf{R}$ and direction $\mathbf{D} = \mathbf{Id} - \mathbf{B}^{-1}\mathbf{R}$, that is, $\mathbf{\Pi}_R = \mathbf{R}(\mathbf{R}^T\mathbf{B}^{-1}\mathbf{R})^{-1}\mathbf{R}^T\mathbf{B}^{-1}$. One can verify that $\mathbf{\Pi}_R$ belongs to the class of optimal projections $\mathcal{P}_{opt}$ defined in Prop. \ref{canon_form} by defining $\mathbf{U}=\mathbf{B}^{-1/2}\mathbf{R}(\mathbf{R}^T\mathbf{B}^{-1}\mathbf{R})^{-1/2}$, and noting that $\mathbf{U}^T\mathbf{U}=\mathbf{Id}$ and $\mathbf{\Pi}_R=\mathbf{B}^{1/2}\mathbf{U}\mathbf{U}^T\mathbf{B}^{-1/2}$. With this particular choice for the direction of the projection $\mathbf{\Pi}_R$, the information content of the subspace $\{\mathbf{r}_i, i = 1, ..., k\}$ can be defined as the DOFS of the projected Bayesian problem $\mathcal{B}_{\Pi_R}$, that is, the DOFS of the Bayesian problem projected onto that subspace along the direction that maximizes the DOFS (see Thm \ref{thm:best_prolong}).

\paragraph{Most Informed Subspaces}
Thm. \ref{thm:opt_proj} shows that the maximum-DOFS projection is constructed incrementally using  $\mathbf{\Pi}^k=\mathbf{\Pi}^{k-1}+\mathbf{B}^{1/2}\mathbf{v}_i\mathbf{v}_i^T\mathbf{B}^{-1/2}$. Therefore, the subspace defined by the basis $\{\mathbf{B}^{1/2}\mathbf{v}_i,\, i=1,...,k\}$, which corresponds to the range of $\mathbf{\Pi}^k$, can be interpreted as the \textit{most informed subspace} of dimension $k$, while the vector $\mathbf{B}^{1/2}\mathbf{v}_j$ defines the $j$th most constrained direction.

\paragraph{Independently Constrained Modes}
The vectors $\mathbf{B}^{1/2}\mathbf{v}_j$ are the eigenvectors of the model resolution matrix $\mathbf{A}\equiv \frac{\partial \mathbf{x}^a}{\partial \mathbf{x}^t}$. They can therefore be interpreted as the modes that are independently constrained by the observations, since one has (Eq. (\ref{eq:update_max_dof_mod_res_new})) $\frac{\partial(\mathbf{B}^{1/2}\mathbf{v}_i)^a}{\partial(\mathbf{B}^{1/2}\mathbf{v}_j)^t}=\lambda_i(1+\lambda_i)^{-1} \delta_{ij}$ (where $\delta_{ij}$ represents the Kronecker delta).
 
\paragraph{Projected Forward Model}
Based on Rem. \ref{rem:proj_pb}, one can establish a link between the posterior solutions of the maximum-DOFS projection and the posterior solutions of the  Bayesian problem $\mathcal{B}_{\mathbf{H}_{\Pi_\mathrm{dof}}}\equiv(E,F,\mathbf{H}{\mathbf{\Pi}_\mathrm{dof}},\mathbf{B},\mathbf{R}_{\Pi_\mathrm{dof}})$, which corresponds to the initial Bayesian problem, $\mathcal{B}$, with a projected forward model. Indeed, one can verify that the posterior solutions of $\mathcal{B}_{\mathbf{H}_{\Pi_\mathrm{dof}}}$ can be expressed as:
\begin{align}
\label{H_pi_1}
 \mathbf{x}^a_{\mathbf{H}_{\Pi_\mathrm{dof}}} &= \mathbf{x}^a_{\Pi_\mathrm{dof}}\\
 \label{H_pi_2}
 \mathbf{P}^a_{\mathbf{H}_{\Pi_\mathrm{dof}}}&= \mathbf{P}^a_{\Pi_\mathrm{dof}} + \mathbf{B}^{1/2}\left(\sum_{i=k+1}^{n} \mathbf{v}_i\mathbf{v}^T_i\right)\mathbf{B}^{1/2} =\mathbf{B}-\mathbf{B}^{1/2}\left(\sum_{i=1}^{k} \lambda_i(1+\lambda_i)^{-1}\mathbf{v}_i\mathbf{v}^T_i\right)\mathbf{B}^{1/2}  \\
 \label{H_pi_3}
 \mathbf{A}_{\mathbf{H}_{\Pi_\mathrm{dof}}}&=\mathbf{A}_{\Pi_\mathrm{dof}}
\end{align}
The optimality properties of the low-rank update approximation (\ref{H_pi_2}) will be described and exploited in the following Section.
 
 \subsection{Link With Low-Rank Approximations}
\label{link_lr_approx}
 
The maximum-DOFS projection constructed in Section \ref{opt_proj} defines a rank-$k$ Bayesian inverse problem whose information content is maximal among all rank-$k$ projections of the initial Bayesian problem.  Furthermore, the posterior mean and posterior error covariance matrix of the maximum-DOFS projection are also low-rank approximations to the initial full-dimensional posterior mean and posterior error covariance matrix, respectively. In this Section, we provide important optimality results associated with those low-rank approximations to the posterior solution. Additionally, useful optimality results associated with low-rank approximations to the posterior solution that do not correspond to projections are also provided, which can be alternatively used when only best approximations to the solution of the original Bayesian problem are sought and consistency between the approximated posterior mean and posterior errors is not required (see Section \ref{diff_proj_approx}).

\subsubsection{Optimal Low-Rank Approximations}
Before establishing optimality results associated with low-rank approximations to the posterior solution, one needs to define an appropriate class of approximations. As discussed in Section \ref{info_content} , a natural class of approximations for the posterior error covariance matrix is the one that corresponds to negative updates to the prior error covariance matrix. This particular class is central to our analysis, since the negative update can be interpreted as the information content of the inversion (see Section \ref{info_content}). We note that previous studies have already demonstrated the importance of this approximation class for metrics useful in the Bayesian framework (e.g., \citet{spantini2015optimal,cui2014likelihood}).
\begin{defn}[\textbf{Classes of Approximations}]
\label{def_class}
  Let us define the following classes of matrices:
      \begin{align}
&\mathcal{A}_k\equiv \{  \mathbf{M}\in\mathcal{M}_n \,|\, \mathrm{rank}(\mathbf{M})\le k  \}   \nonumber \\
  &    \hat{\mathcal{A}}_k\equiv  \{  \mathbf{M}\in\mathcal{M}_n \,|\, \mathbf{M}=\mathbf{B}- \mathbf{QQ}^T\ge 0,\, \mathrm{rank}(\mathbf{Q})\le k  \} \nonumber \\
&   \hat{\mathcal{O}_k}\equiv \{  \mathbf{M}=\mathbf{P}\mathbf{H}^T\mathbf{R}^{-1}\,|\,\mathbf{P}\in  \hat{\mathcal{A}}_k  \},   \nonumber 
    \end{align}
     where $\hat{\mathcal{A}}_k$ defines the class of negative semidefinite updates to the prior error covariance matrix $\mathbf{B}$. 
          \end{defn}

   The approximations that belong to the class $\mathcal{A}_k$ are referred to as low-rank approximations, while the approximations that belong to the class $ \hat{\mathcal{A}}_k$ are low-rank update approximations. Therefore, the class $\hat{\mathcal{O}_k}$ is associated with full-rank approximations to the posterior mean update.

In the following Proposition, optimality results for several posterior error covariance matrix approximations are provided. All approximations are based on the classes defined in Def. \ref{def_class} and on the truncated eigendecomposition of either $\mathbf{Q}_\textrm{dof}$ (see Section \ref{opt_proj}) or $\mathbf{Q}_\textrm{var}\equiv\mathbf{B}-\mathbf{P}^a$, which are both related to the information content of the inversion. 

 \begin{prop}[\textbf{Optimal Approximations of the Posterior Error Covariance}]
 \label{thm:opt_approx_scov}
   \text{   }  
\newline   
 
     Using the previous notations, let us define: 
      \begin{align}
      \label{LR_dof_post_err}
&  \mathbf{P}^a_{\mathbf{H}_{\Pi_\mathrm{dof}}}=\mathbf{B}-\mathbf{B}^{1/2}\left(\sum_{i=1}^{k} \lambda_i(1+\lambda_i)^{-1}\mathbf{v}_i\mathbf{v}^T_i\right)\mathbf{B}^{1/2} \\
           \label{var_post_err}
     &\mathbf{P}^a_{\mathrm{var}}\equiv\mathbf{B}-\sum_{i=1}^{k} \delta_i\mathbf{u}_i\mathbf{u}^T_i ,
           \end{align}
    where $\mathbf{u}_i$ and $\delta_i$ are the $i$th eigenvector and eigenvalue, respectively, of $\mathbf{Q}_\mathrm{var}=\mathbf{B}\mathbf{H}^T(\mathbf{HBH}^T+\mathbf{R})^{-1}\mathbf{H}\mathbf{B}$, and where $  \mathbf{P}^a_{\mathbf{H}_{\Pi_\mathrm{dof}}}$ is the posterior error covariance of the projected Bayesian problem $\mathcal{B}_{\mathbf{H}_{\Pi_\mathrm{opt}}}$.     \newline
    
     One has the following optimality properties:

 \begin{align}
  \label{F_B_var}
 &\| \mathbf{P}^a_{\mathrm{var}}-\mathbf{P}^a\|_F =  \min_{\mathbf{\tilde{\mathbf{P}}}\in\hat{\mathcal{A}}_k} \|\tilde{\mathbf{P}} -\mathbf{P}^a\|_F=\sqrt{\sum_{i>k} \delta_i^2}  \\ 
 \label{F_B_dof}
 &\|\mathbf{P}^a_{\mathbf{H}_{\Pi_\mathrm{dof}}}-\mathbf{P}^a\|_{F,\mathbf{B}^{-1}} =  \min_{\mathbf{\tilde{\mathbf{P}}}\in\hat{\mathcal{A}}_k} \|\tilde{\mathbf{P}} -\mathbf{P}^a \|_{F,\mathbf{B}^{-1}}=\sqrt{\sum_{i>k} \left (\frac{\lambda_i}{1+\lambda_i}\right )^2} \\
  \label{F_Pa_dof}
 &\|\mathbf{P}^a_{\mathbf{H}_{\Pi_\mathrm{dof}}}-\mathbf{P}^a\|_{F,({\mathbf{P}^a})^{-1}} =  \min_{\mathbf{\tilde{\mathbf{P}}}\in\hat{\mathcal{A}}_k} \|\tilde{\mathbf{P}} -\mathbf{P}^a \|_{F,({\mathbf{P}^a})^{-1}} =\sqrt{\sum_{i>k} \lambda_i^2} \\
   \label{F_B_LRdof}
 &\|  \mathbf{P}^{a}_{\mathbf{\Pi}_\mathrm{dof}}-\mathbf{P}^a\|_{F,\mathbf{B}^{-1}} =  \min_{\mathbf{\tilde{\mathbf{P}}}\in\mathcal{A}_k} \|\tilde{\mathbf{P}} -\mathbf{P}^a \|_{F,\mathbf{B}^{-1}}=\sqrt{\sum_{i>k} \left (\frac{1}{1+\lambda_i}\right )^2} \\
    \label{F_Pa_LRdof}
 &\| \mathbf{P}^{a}_{\mathbf{\Pi}_\mathrm{dof}}-\mathbf{P}^a\|_{F,({\mathbf{P}^a})^{-1}} =  \min_{\mathbf{\tilde{\mathbf{P}}}\in\mathcal{A}_k} \|\tilde{\mathbf{P}} -\mathbf{P}^a \|_{F,({\mathbf{P}^a})^{-1}} =\sqrt{n-k}
 \end{align}
where:
\begin{itemize} 
\item $\|.\|_F$ represents the Frobenius norm.
 \item $\|.\|_{F,\mathbf{W}}$ is the weighted Frobenius norm defined by $\|\mathbf{M}\|_{F,\mathbf{W}}=\|\mathbf{B}^{1/2}\mathbf{M}\mathbf{B}^{1/2}\|_F$, where $\mathbf{B}^{1/2}$ is a square-root of $\mathbf{W}$ $(\mathbf{W}=\mathbf{LL}^T)$.
\end{itemize}

\end{prop}
\begin{proof}
The proof of (\ref{F_B_var}) follows immediately from the Eckart-Young theorem [Eckart and Young, 1936], since one has $\| \mathbf{P}^a_\mathrm{var} -\mathbf{P}^a\|_F=\| \sum_{i=1}^{k} \delta_i\mathbf{u}_i\mathbf{u}^T_i-\mathbf{B}\mathbf{H}^T(\mathbf{HBH}^T+\mathbf{R})^{-1}\mathbf{H}\mathbf{B}\|_F=\min_{\mathbf{\tilde{\mathbf{P}}}\in\hat{\mathcal{A}}_k}  \|\tilde{\mathbf{P}} -\mathbf{Q}_\mathrm{var}\|_F$.\newline
The proofs for formulas (\ref{F_B_dof})-(\ref{F_Pa_LRdof}) are obtained by writing:
\begin{align} 
\label{diff_pafr_pa}
\mathbf{P}^a_{\mathbf{H}_{\Pi_\mathrm{dof}}}-\mathbf{P}^a&=\mathbf{B}^{1/2}\left(\sum_{i=1}^{k} \lambda_i(1+\lambda_i)^{-1}\mathbf{v}_i\mathbf{v}^T_i - \sum_{i=1}^{n} \lambda_i(1+\lambda_i)^{-1}\mathbf{v}_i\mathbf{v}^T_i\right)\mathbf{B}^{1/2} \\
\label{diff_padof_pa}
\mathbf{P}^a_{\mathbf{\Pi}_\mathrm{dof}}-\mathbf{P}^a&=\mathbf{B}^{1/2}\left(\sum_{i=1}^{k}( 1+\lambda_i )^{-1}\mathbf{v}_i\mathbf{v}_i^T-\sum_{i=1}^{n}( 1+\lambda_i )^{-1}\mathbf{v}_i\mathbf{v}_i^T\right )\mathbf{B}^{1/2}
\end{align}
Multiplying on the left and on the right by the inverse of the square-root of $\mathbf{B}$ in (\ref{diff_pafr_pa}) and (\ref{diff_padof_pa}), we obtain:
\begin{align*}
 \|\mathbf{P}^a_{\mathbf{H}_{\Pi_\mathrm{dof}}}-\mathbf{P}^a\|_{F,\mathbf{B}^{-1}}& =\| \left(\sum_{i=1}^{k} \lambda_i(1+\lambda_i)^{-1}\mathbf{v}_i\mathbf{v}^T_i - \sum_{i=1}^{n} \lambda_i(1+\lambda_i)^{-1}\mathbf{v}_i\mathbf{v}^T_i\right) \|_F \\
 &=\min_{\mathbf{\tilde{\mathbf{M}}}\in\hat{\mathcal{A}}_k}  \|\tilde{\mathbf{M}} -\sum_{i=1}^{n} \lambda_i(1+\lambda_i)^{-1}\mathbf{v}_i\mathbf{v}^T_i\|_F \\
 &=\sqrt{\sum_{i>k} \left (\frac{\lambda_i}{1+\lambda_i}\right )^2}\\ 
 \|\mathbf{P}^a_{\mathbf{\Pi}_\mathrm{dof}}-\mathbf{P}^a\|_{F,\mathbf{B}^{-1}}& =\| \left(\sum_{i=1}^{k}(1+\lambda_i)^{-1}\mathbf{v}_i\mathbf{v}^T_i - \sum_{i=1}^{n} (1+\lambda_i)^{-1}\mathbf{v}_i\mathbf{v}^T_i\right) \|_F \\
 &=\min_{\mathbf{\tilde{\mathbf{M}}}\in\hat{\mathcal{A}}_k}  \|\tilde{\mathbf{M}} -\sum_{i=1}^{n} (1+\lambda_i)^{-1}\mathbf{v}_i\mathbf{v}^T_i\|_F \\
 &=\sqrt{\sum_{i>k} \left (\frac{1}{1+\lambda_i}\right )^2}
 \end{align*}
Using the previous notations, a square-root of $(\mathbf{P}^a)^{-1}=(\mathbf{B}^{-1}+\mathbf{H}^T\mathbf{R}^{-1}\mathbf{H})=\mathbf{B}^{-1/2}(\mathbf{Id}+\mathbf{B}^{1/2}\mathbf{H}^T\mathbf{R}^{-1}\mathbf{H}\mathbf{B}^{1/2})\mathbf{B}^{-1/2}$ is given by $\mathbf{L}_{{\mathbf{P}^a}^{-1}}=\mathbf{B}^{-1/2}\mathbf{V}(\mathbf{Id}+\mathbf{\Lambda})^{1/2}\mathbf{V}^T$. Multiplying on the left and on the right by ${\mathbf{L}_{\mathbf{P}^a}}^T$ and ${\mathbf{L}_{\mathbf{P}^a}}$, respectively, in (\ref{diff_pafr_pa}) and (\ref{diff_padof_pa}), we obtain:
\begin{align*}
 \|\mathbf{P}^a_{\mathbf{H}_{\Pi_\mathrm{dof}}}-\mathbf{P}^a\|_{F,(\mathbf{P}^a)^{-1}}& =\| \left(\sum_{i=1}^{k} \lambda_i\mathbf{v}_i\mathbf{v}^T_i - \sum_{i=1}^{n} \lambda_i\mathbf{v}_i\mathbf{v}^T_i\right) \|_F \\
 &=\min_{\mathbf{\tilde{\mathbf{M}}}\in\hat{\mathcal{A}}_k}  \|\tilde{\mathbf{M}} -\sum_{i=1}^{n} \lambda_i\mathbf{v}_i\mathbf{v}^T_i\|_F \\
 &=\sqrt{\sum_{i>k}  \lambda_i^2}\\ 
 \|\mathbf{P}^a_{\mathbf{\Pi}_\mathrm{dof}}-\mathbf{P}^a\|_{F(\mathbf{P}^a)^{-1}}& =\| \left(\sum_{i=1}^{k}\mathbf{v}_i\mathbf{v}^T_i - \sum_{i=1}^{n} \mathbf{v}_i\mathbf{v}^T_i\right) \|_F \\
 &=\min_{\mathbf{\tilde{\mathbf{M}}}\in\hat{\mathcal{A}}_k}  \|\tilde{\mathbf{M}} -\sum_{i=1}^{n}\mathbf{v}_i\mathbf{v}^T_i\|_F \\
 &=\sqrt{n-k}
 \end{align*}
\end{proof}
\begin{rem}
 It has been shown in \citet{spantini2015optimal} that $\mathbf{P}^a_{\mathbf{H}_{\Pi_\mathrm{dof}}}$ also verifies: $d_{\mathcal{F}}( \mathbf{P}^a_{\mathbf{H}_{\Pi_\mathrm{dof}}},\mathbf{P}^a) =  \min_{ \tilde{\mathbf{P}}\in\hat{\mathcal{A}}_k} d_{\mathcal{F}}( \tilde{\mathbf{P}} ,\mathbf{P}^a) $, where $d_{\mathcal{F}}$ is the \textit{F\"orstner distance}, defined by  $d_{\mathcal{F}}(\mathbf{P},\mathbf{N})=\sum_i(\ln{\sigma_i})^2$, where ($\sigma_i$) is the sequence of generalized eigenvalues of the pencil $(\mathbf{P},\mathbf{N})$.
\end{rem}

We now establish optimality results for several posterior mean approximations. In addition to the maximum-DOFS solution ($\mathbf{x}^a_{\mathbf{\Pi}_\mathrm{dof}}$), a full-rank posterior mean approximation ($\mathbf{x}^{a}_{\mathrm{FR}_\mathrm{dof}}$) is considered. It is obtained by replacing the posterior error covariance implicit in Eq. (\ref{post_update2}) by the low-rank update approximation $\mathbf{P}^a_{\mathbf{H}_{\Pi_\mathrm{dof}}}$. The truncated eigendecomposition of $\mathbf{Q}_\textrm{var}$ is also exploited to define another optimal approximation of the posterior mean ($\mathbf{x}^a_{\mathrm{var}}$).  The optimality results (\ref{eq:opt_dof_Pa_mean}) and (\ref{eq:opt_dof_B_mean_FR}) below can be found in \citet{spantini2015optimal}. We recall the proofs here since the same principles can be applied to prove the other optimality results.

 \begin{prop}[\textbf{Optimal Approximations of the Posterior Mean}]
 \label{thm:opt_approx_mean}
  \text{   }  
\newline   

 Using the previous notations, let us define the following posterior mean approximations:
   \begin{align} 
    \label{eq:FR_postupdate}
     &  \mathbf{x}^{a}_{\mathrm{FR}_\mathrm{dof}}\equiv\mathbf{P}^a_{\mathbf{H}_{\Pi_\mathrm{dof}}}\mathbf{H}^T\mathbf{R}^{-1} \mathbf{d} \\
     \label{eq:update_min_var_mean}
   &\mathbf{x}^a_{\mathrm{var}}=\sum_{i=1}^{k} \delta_i^{1/2}\mathbf{u}_i\mathbf{z}^T_i\left( \mathbf{HBH}^T+\mathbf{R}\right)^{-1/2}\mathbf{d} ,
\end{align}
   where $\mathbf{z}_i=(\mathbf{HBH}^T+\mathbf{R})^{-1/2}\mathbf{H}\mathbf{B} \mathbf{u}_i$.   \newline
   
  One has the following optimality properties:
  
\begin{align}
\label{eq:opt_var_mean}
 & \mathbb{E}  \| \mathbf{x}^a_{\mathrm{var}}-\mathbf{x}^a  \|^2 = \min_{\tilde{\mathbf{K}}\in\mathcal{A}_k}  \mathbb{E}  \| (\tilde{\mathbf{K}}-\mathbf{K}) \mathbf{d}\|^2 =\sum_{i>k} \delta_i \\
\label{eq:opt_dof_B_mean}
  & \mathbb{E}  \| \mathbf{x}^a_{\mathbf{\Pi}_\mathrm{dof}}-\mathbf{x}^a  \|_{\mathbf{B}^{-1}}^2 = \min_{\tilde{\mathbf{K}}\in\mathcal{A}_k}  \mathbb{E}  \| (\tilde{\mathbf{K}}-\mathbf{K}) \mathbf{d}\|^2_{\mathbf{B}^{-1}} =\sum_{i>k}\frac{\lambda_i}{1+\lambda_i} \\
  \label{eq:opt_dof_Pa_mean}
 & \mathbb{E} \| \mathbf{x}^a_{\mathbf{\Pi}_\mathrm{dof}}-\mathbf{x}^a  \|_{(\mathbf{P}^{a})^{-1}}^2 = \min_{\tilde{\mathbf{K}}\in\mathcal{A}_k}  \mathbb{E}  \| (\tilde{\mathbf{K}}-\mathbf{K}) \mathbf{d}\|^2_{(\mathbf{P}^{a})^{-1}} =\sum_{i>k} \lambda_i\\
   \label{eq:opt_FRdof_B_mean}
  & \mathbb{E}  \|  \mathbf{x}^{a}_{\mathrm{FR}_\mathrm{dof}}-\mathbf{x}^a  \|_{\mathbf{B}^{-1}}^2 = \min_{\tilde{\mathbf{P}}\in\hat{\mathcal{O}_k} }  \mathbb{E}  \|(\tilde{\mathbf{P}}-\mathbf{P}^a)\mathbf{H}^T\mathbf{R}^{-1} \mathbf{d}\|^2_{\mathbf{B}^{-1}} =\sum_{i>k}\frac{\lambda_i^3}{1+\lambda_i} \\
  \label{eq:opt_dof_B_mean_FR}
 & \mathbb{E} \| \mathbf{x}^{a}_{\mathrm{FR}_\mathrm{dof}}-\mathbf{x}^a  \|_{(\mathbf{P}^{a})^{-1}}^2 = \min_{\tilde{\mathbf{P}}\in\hat{\mathcal{O}_k} } \mathbb{E}  \| (\tilde{\mathbf{P}}-\mathbf{P}^a)\mathbf{H}^T\mathbf{R}^{-1} \mathbf{d}\|^2_{(\mathbf{P}^{a})^{-1}}=\sum_{i>k} \lambda_i^3,
\end{align}
where:
\begin{itemize} 
\item$\| .\|$ represent the Euclidian norm.
\item $\|\mathbf{x}\|_{\mathbf{W}}=\sqrt{\mathbf{x}^T\mathbf{W}\mathbf{x}}$ is the weighted Euclidian norm.
 \item$\mathbb{E}$ is the average operator.
 \item$\mathbf{K}=\mathbf{BH}^T(\mathbf{HBH}^T+\mathbf{R})^{-1}$ is the gain matrix of the initial Bayesian problem.
\end{itemize}
 \end{prop}
  
 \begin{proof}
 To prove (\ref{eq:opt_var_mean}), we use Lemma \ref{equiv_average_F} and the fact that $ \mathbb{E} (\mathbf{d}\mathbf{d}^T)=\mathbf{HBH}^T+\mathbf{R}$:  
 \begin{align}
 \mathbb{E}  \| (\tilde{\mathbf{K}}-\mathbf{K}) \mathbf{d}\|^2 &=\| (\tilde{\mathbf{K}}-\mathbf{K})(\mathbf{HBH}^T+\mathbf{R})^{1/2}\|^2_F \\
 \label{min_var_mean}
 &=\| (\tilde{\mathbf{K}}(\mathbf{HBH}^T+\mathbf{R})^{1/2}-\mathbf{BH}^T(\mathbf{HBH}^T+\mathbf{R})^{-1/2})\|^2_F
  \end{align} 
 Now, using Theorem 2.1 of \citet{friedland2007generalized}, a solution of $ \min_{\tilde{\mathbf{K}}\in\mathcal{A}_k}  \mathbb{E}  \| (\tilde{\mathbf{K}}-\mathbf{K}) \mathbf{d}\|^2$ is given by:
 \begin{align}
 \label{Kopt}
\tilde{\mathbf{K}}_\mathrm{opt}=\sum_{i=1}^{k} \delta_i^{1/2}\mathbf{u}_i\mathbf{z}^T_i\left( \mathbf{HBH}^T+\mathbf{R}\right)^{-1/2},
 \end{align}
 where $\sum_{i=1}^{k} \delta_i^{1/2}\mathbf{u}_i\mathbf{z}^T_i$ is the truncated SVD of rank $k$ of $\mathbf{BH}^T(\mathbf{HBH}^T+\mathbf{R})^{-1/2}$ and $(\mathbf{HBH}^T+\mathbf{R})^{1/2}$ is a non-singular square-root of $\mathbf{HBH}^T+\mathbf{R}$. Replacing $\tilde{\mathbf{K}}$ by (\ref{Kopt}) and $\mathbf{BH}^T(\mathbf{HBH}^T+\mathbf{R})^{-1/2}$ by $\sum_{i=1}^{n} \delta_i^{1/2}\mathbf{u}_i\mathbf{z}^T_i$ in (\ref{min_var_mean}) yields (\ref{eq:opt_var_mean}).\newline
 
 Below we prove (\ref{eq:opt_dof_Pa_mean}). The proof for (\ref{eq:opt_dof_B_mean}) is obtained similarly. To prove (\ref{eq:opt_dof_Pa_mean}), we also use Lemma \ref{equiv_average_F} and the square-roots ${\mathbf{L}_{({\mathbf{P}^a})^{-1}}}=\mathbf{B}^{-1/2}\mathbf{V}(\mathbf{I+\Lambda})^{1/2}\mathbf{V}^T$ and ${\mathbf{L}_{\mathbf{Y}}}=\mathbf{R}^{1/2}\mathbf{W}(\mathbf{Id}+\mathbf{\Lambda})^{1/2}\mathbf{W}^T$ of $(\mathbf{P}^a)^{-1}$ and  $\mathbf{Y}=\mathbf{HBH}^T+\mathbf{R}$, respectively. One has:
 \begin{align}
 \label{eq_K_Pa}
  \mathbb{E}  \| (\tilde{\mathbf{K}}-\mathbf{K}) \mathbf{d}\|^2_{(\mathbf{P}^{a})^{-1}}  &=\|{\mathbf{L}_{({\mathbf{P}^a})^{-1}}}^T(\tilde{\mathbf{K}}-\mathbf{K}){\mathbf{L}_{\mathbf{Y}}}\|^2_F
  \end{align} 
Using Theorem 2.1 of \citet{friedland2007generalized}, a solution of $\min_{\tilde{\mathbf{K}}\in\mathcal{A}_k}  \mathbb{E}  \| (\tilde{\mathbf{K}}-\mathbf{K}) \mathbf{d}\|^2_{(\mathbf{P}^a)^{-1}}$ is therefore given by:
 \begin{align}
 \label{Kopt}
\tilde{\mathbf{K}}_\mathrm{opt}={\mathbf{L}_{(\mathbf{P}^a)^{-1}}}^{-T} \left[ {\mathbf{L}_{(\mathbf{P}^a)^{-1}}}^T\mathbf{K}  {\mathbf{L}_{\mathbf{Y}}}\right ]_k {\mathbf{L}_{\mathbf{Y}}}^{-1 },
 \end{align}
  where $\left[ \mathbf{M}\right ]_k$ is the rank-$k$ truncated SVD of the matrix $\mathbf{M}$. Further developing (\ref{Kopt}), one obtains:
  \begin{align*}
  \tilde{\mathbf{K}}_\mathrm{opt}=&\mathbf{B}^{1/2} \mathbf{V}(\mathbf{I+\Lambda})^{-1/2}\mathbf{V}^T \\
  & \left[ \mathbf{V}(\mathbf{Id}+\mathbf{\Lambda})^{1/2}\mathbf{V}^{T}\mathbf{B}^{1/2}\mathbf{H}^T\mathbf{R}^{-1/2}(\mathbf{Id}+\mathbf{R}^{-1/2}\mathbf{HBH}^T\mathbf{R}^{-1/2})^{-1}\mathbf{W}(\mathbf{Id}+\mathbf{\Lambda})^{1/2}\mathbf{W}^T \right ]_k \\
 &  \mathbf{W}(\mathbf{Id}+\mathbf{\Lambda})^{-1/2}\mathbf{W}^T \mathbf{R}^{-1/2}
  \end{align*}
  Using $\mathbf{B}^{1/2}\mathbf{H}^T\mathbf{R}^{-1/2}=\mathbf{V}\mathbf{\Lambda}^{1/2}\mathbf{W}^T$ in the expression above one obtains:
   \begin{align*}
  \tilde{\mathbf{K}}_\mathrm{opt}=&\mathbf{B}^{1/2} \mathbf{V}(\mathbf{I+\Lambda})^{-1/2}\mathbf{V}^T \\
  & \left[ \mathbf{V}(\mathbf{Id}+\mathbf{\Lambda})^{1/2}\mathbf{V}^{T}\mathbf{V}\mathbf{\Lambda}^{1/2}\mathbf{W}^T\mathbf{W}(\mathbf{Id}-\mathbf{\Lambda}(\mathbf{Id}+\mathbf{\Lambda})^{-1})\mathbf{W}^T\mathbf{W}(\mathbf{Id}+\mathbf{\Lambda})^{1/2}\mathbf{W}^T \right ]_k \\
 &  \mathbf{W}(\mathbf{Id}+\mathbf{\Lambda})^{-1/2}\mathbf{W}^T \mathbf{R}^{-1/2} \\
 =&\mathbf{B}^{1/2} \mathbf{V}(\mathbf{I+\Lambda})^{-1/2}\mathbf{V}^T \\
  & \left[ \sum_{i=1}^{k} \mathbf{v}_i \mathbf{w}^T_i(1+\lambda_i)^{1/2}\lambda_i^{1/2}(1-\lambda_i(1+\lambda_i)^{-1})(1+\lambda_i)^{1/2} \right ] \\
 &  \mathbf{W}(\mathbf{Id}+\mathbf{\Lambda})^{-1/2}\mathbf{W}^T \mathbf{R}^{-1/2} \\
 =&\mathbf{B}^{1/2} \left(\sum_{i=1}^{k} \lambda_i^{1/2}(1+\lambda_i)^{-1}\mathbf{v}_i\mathbf{w}^T_i\right)\mathbf{R}^{-1/2},
  \end{align*} 
  where we used the orthogonality properties of $\mathbf{W}$ and $\mathbf{V}$. Finally, using \\
  $\tilde{\mathbf{K}}_\mathrm{opt}=\mathbf{B}^{1/2} \left(\sum_{i=1}^{k} \lambda_i^{1/2}(1+\lambda_i)^{-1}\mathbf{v}_i\mathbf{w}^T_i\right)\mathbf{R}^{-1/2}$ and $\mathbf{K}=\mathbf{B}^{1/2} \left(\sum_{i=1}^{n} \lambda_i^{1/2}(1+\lambda_i)^{-1}\mathbf{v}_i\mathbf{w}^T_i\right)\mathbf{R}^{-1/2}$ in the right-hand side of (\ref{eq_K_Pa}) leads to $ \mathbb{E}  \| (\tilde{\mathbf{K}}_\mathrm{opt}-\mathbf{K}) \mathbf{d}\|^2_{(\mathbf{P}^{a})^{-1}}  =\sum_{i>k} \lambda_i$, which proves (\ref{eq:opt_dof_Pa_mean}).\newline
  Below we prove (\ref{eq:opt_dof_B_mean_FR}), (\ref{eq:opt_FRdof_B_mean}) being obtained similarly. Using expression (\ref{post_update2}) for the posterior mean and Lemma \ref{equiv_average_F}, we obtain:
  \begin{align}
 \label{eq_P_Pa1}
 \mathbb{E} \| \mathbf{x}^{a}_{\mathrm{FR}_\mathrm{dof}}-\mathbf{x}^a  \|_{(\mathbf{P}^{a})^{-1}}^2 &= \mathbb{E}  \| (\tilde{\mathbf{P}}-\mathbf{P}^a)\mathbf{H}^T\mathbf{R}^{-1} \mathbf{d}\|^2_{(\mathbf{P}^{a})^{-1}}  \\
  \label{eq_P_Pa2}
 &= \| {\mathbf{L}_{(\mathbf{P}^a)^{-1}}}^T(\tilde{\mathbf{P}}-\mathbf{P}^a)\mathbf{H}^T\mathbf{R}^{-1}{\mathbf{L}_{\mathbf{Y}}}\|^2_F
 \end{align}
 We first note that:
 \begin{align*}
  \min_{\tilde{\mathbf{P}}\in\hat{\mathcal{O}_k} } \| {\mathbf{L}_{(\mathbf{P}^a)^{-1}}}^T(\tilde{\mathbf{P}}-\mathbf{P}^a)\mathbf{H}^T\mathbf{R}^{-1}{\mathbf{L}_{\mathbf{Y}}}\|^2_F =\min_{\tilde{\mathbf{F}}\in\mathcal{A}_k} \| {\mathbf{L}_{(\mathbf{P}^a)^{-1}}}^T(\tilde{\mathbf{F}}-(\mathbf{P}^a-\mathbf{B}))\mathbf{H}^T\mathbf{R}^{-1}{\mathbf{L}_{\mathbf{Y}}}\|^2_F
  \end{align*}
  Using Theorem 2.1 of \citet{friedland2007generalized}, a solution of \\
 $\min_{\tilde{\mathbf{F}}\in\mathcal{A}_k} \| {\mathbf{L}_{(\mathbf{P}^a)^{-1}}}^T(\tilde{\mathbf{F}}-(\mathbf{P}^a-\mathbf{B}))\mathbf{H}^T\mathbf{R}^{-1}{\mathbf{L}_{\mathbf{Y}}}\|^2_F$ is given by:
 \begin{align*}
\tilde{\mathbf{F}}_\mathrm{opt}= {\mathbf{L}_{(\mathbf{P}^a)^{-1}}}^{-T} \left[ {\mathbf{L}_{(\mathbf{P}^a)^{-1}}}^T(\mathbf{P}^a-\mathbf{B})  \mathbf{H}^T\mathbf{R}^{-1}{\mathbf{L}_{\mathbf{Y}}}\right ]_k( { \mathbf{H}^T\mathbf{R}^{-1}\mathbf{L}_{\mathbf{Y}}})^{+},
 \end{align*}
 where $^{+}$ denotes the Moore-Penrose pseudoinverse.
 One can verify that another minimizer of (\ref{eq_P_Pa2}) is:
  \begin{align}
\tilde{\mathbf{F}}'_\mathrm{opt}= {\mathbf{L}_{(\mathbf{P}^a)^{-1}}}^{-T} \left[ {\mathbf{L}_{(\mathbf{P}^a)^{-1}}}^T(\mathbf{P}^a-\mathbf{B})\mathbf{H}^T\mathbf{R}^{-1}{\mathbf{L}_{\mathbf{Y}}}\right ]_k( {\mathbf{B}^{1/2} \mathbf{H}^T\mathbf{R}^{-1}\mathbf{L}_{\mathbf{Y}}})^{+} \mathbf{B}^{1/2}
 \end{align}
Factorizing the expression above using the square-root $\mathbf{B}^{1/2} \mathbf{H}^T\mathbf{R}^{-1/2}$ and its SVD like before, we obtain:
 \begin{align}
\tilde{\mathbf{F}}'_\mathrm{opt}=& \mathbf{B}^{1/2} \mathbf{V}(\mathbf{I+\Lambda})^{-1/2}\mathbf{V}^T \\
& \left[ \sum_{i=1}^{k} \mathbf{v}_i \mathbf{w}^T_i(1+\lambda_i)^{1/2}(\lambda_i(1+\lambda_i)^{-1})\lambda_i^{1/2} (1+\lambda_i)^{1/2}\right ] \\ 
&\mathbf{W}(\mathbf{Id}+\mathbf{\Lambda})^{-1/2}\mathbf{\Lambda}^{-1/2} \mathbf{V}^T \mathbf{B}^{1/2}  \\
=&\mathbf{B}^{1/2}\left(\sum_{i=1}^{k} \lambda_i(1+\lambda_i)^{-1}\mathbf{v}_i\mathbf{v}^T_i\right)\mathbf{B}^{1/2}
 \end{align}
Therefore:
\begin{align}
\tilde{\mathbf{P}}'_\mathrm{opt}=\mathbf{B}-\tilde{\mathbf{F}}'_\mathrm{opt}
&=\mathbf{B}-\mathbf{B}^{1/2}\left(\sum_{i=1}^{k} \lambda_i(1+\lambda_i)^{-1}\mathbf{v}_i\mathbf{v}^T_i\right)\mathbf{B}^{1/2},
\end{align}
which proves the first equality of (\ref{eq:opt_dof_B_mean_FR}). Finally, using $\tilde{\mathbf{P}}'_\mathrm{opt}=\mathbf{B}-\mathbf{B}^{1/2}\left(\sum_{i=1}^{k} \lambda_i(1+\lambda_i)^{-1}\mathbf{v}_i\mathbf{v}^T_i\right)\mathbf{B}^{1/2}$ and $\mathbf{P}^a=\mathbf{B}-\mathbf{B}^{1/2}\left(\sum_{i=1}^{n} \lambda_i(1+\lambda_i)^{-1}\mathbf{v}_i\mathbf{v}^T_i\right)\mathbf{B}^{1/2}$ in (\ref{eq_P_Pa2}) we obtain the equality with the right-hand side of (\ref{eq:opt_dof_B_mean_FR}).
\end{proof}

\begin{rem}
 The approximations associated with the eigendecomposition of $\mathbf{Q}_\textrm{var}$, i.e., $\mathbf{x}^a_{\mathrm{var}}$ and $\mathbf{P}^a_{\mathrm{var}}$, both correspond to optimal total (non-normalized) error variance approximations. Indeed, $\mathbf{P}^a_{\mathrm{var}}$ is also the negative low-rank update to the prior error that best approximates the total error variance, i.e., $ |\mathrm{Tr}(\mathbf{P}^a_{\mathrm{var}}-\mathbf{P}^a)|= \min_{\mathbf{P} \in\hat{\mathcal{A}}_k} | \mathrm{Tr}(\mathbf{P}-\mathbf{P}^a)|$. 
   \end{rem}
\subsubsection{Interpretation and Application}
\label{interpretation_optimal}
\paragraph{Interpretation of the Norms}
The norms considered for the posterior error covariance approximations in Prop. \ref{thm:opt_approx_scov} are all based on the Frobenius norm, which is defined as $\|\mathbf{A}\|_F=\sqrt{\sum_i | a_{ij} |^2}$ (where $a_{ij}$ is the element of $\mathbf{A}$ associated with the $i$th row and $j$th column), or alternatively, as $\|\mathbf{A}\|_F=\sqrt{\sum_i \sigma_i^2}$ (where $\sigma_i$ represents the $i$th singular value of $\mathbf{A}$). Therefore, this norm accounts for all elements of the matrix in the approximation, or equivalently in the context of covariances matrices, accounts for variances in all directions (this is not the case of, e.g., the spectral norm $\| \mathbf{A} \|_S = \max_i \sigma_i$). Several norms in Prop. \ref{thm:opt_approx_scov} are weighted Frobenius norms. In the case of covariance matrix approximations such as Eq. (\ref{F_B_dof}) to (\ref{F_Pa_LRdof}), those norms can be interpreted as total approximation errors normalized by the variances of the principal modes associated with the weight matrices. Indeed, one has $\|\mathbf{A}\|_{F,\mathbf{Q}^{-1}}=\| \mathbf{V}^T\mathbf{D}^{-1/2} \mathbf{V}\mathbf{A}\mathbf{V}^T\mathbf{D}^{-1/2} \mathbf{V}\|_F$, where $\mathbf{Q}= \mathbf{V}^T\mathbf{D} \mathbf{V}$ is the eigendecomposition of the Hermitian matrix $\mathbf{Q}$. The matrix $\mathbf{D}^{-1/2} \mathbf{V}^T\mathbf{A}\mathbf{V}\mathbf{D}^{-1/2}$ can be interpreted as the covariance matrix $\mathbf{A}$ expressed in the basis of the principal components of $\mathbf{Q}$ (i.e., the eigenvectors $\mathbf{V}$), and whose variances are normalized by the variance of the principal modes (defined by the diagonal elements of $\mathbf{D}$). The left and right products by $\mathbf{V}^T$ and $\mathbf{V}$, respectively, transform the resulting matrix back into the original canonical basis. Therefore, the $\mathbf{B}^{-1}$-weighted Frobenius norm in Prop. \ref{thm:opt_approx_scov} measures the relative approximation error in the posterior error covariance matrix with respect to the prior errors, while the $(\mathbf{P}^a)^{-1}$-weighted Frobenius norm measures the relative approximation error in the posterior error covariance matrix with respect to the posterior errors.

A similar analysis can be performed to interpret the statistical approximation error in the posterior mean in Prop. \ref{thm:opt_approx_mean}. Indeed, one has:
\begin{align*}
 \mathbb{E} \| \mathbf{x}^a_\textrm{approx} -\mathbf{x}^a  \|_{\mathbf{Q}^{-1}}^2&=\mathbb{E}  \left [ (\mathbf{x}^a_\textrm{approx} -\mathbf{x}^a)^T\mathbf{Q}^{-1} (\mathbf{x}^a_\textrm{approx} -\mathbf{x}^a)\right ] \\ 
 &=\mathbb{E}  \left [ (\mathbf{V}(\mathbf{x}^a_\textrm{approx} -\mathbf{x}^a))^T\mathbf{D}^{-1} (\mathbf{V}(\mathbf{x}^a_\textrm{approx} -\mathbf{x}^a))\right ]\\  
  &= \mathbb{E}\left ( \mathrm{Tr}  \left [ (\mathbf{V}(\mathbf{x}^a_\textrm{approx} -\mathbf{x}^a))^T\mathbf{D}^{-1} (\mathbf{V}(\mathbf{x}^a_\textrm{approx} -\mathbf{x}^a)) \right ]  \right ) \\
 &= \mathbb{E}\left ( \mathrm{Tr}  \left [  (\mathbf{D}^{-1/2} \mathbf{V}(\mathbf{x}^a_\textrm{approx} -\mathbf{x}^a))(\mathbf{D}^{-1/2} \mathbf{V}(\mathbf{x}^a_\textrm{approx} -\mathbf{x}^a))^T \right ]  \right ) \\
 &= \mathbb{E}\left ( \mathrm{Tr}  \left [  ( \mathbf{V}^T\mathbf{D}^{-1/2} \mathbf{V}(\mathbf{x}^a_\textrm{approx} -\mathbf{x}^a))( \mathbf{V}^T\mathbf{D}^{-1/2} \mathbf{V}(\mathbf{x}^a_\textrm{approx} -\mathbf{x}^a))^T \right ]  \right ) \\
 &= \mathbb{E} \|  \mathbf{V}^T\mathbf{D}^{-1/2}\mathbf{V}(\mathbf{x}^a_\textrm{approx} -\mathbf{x}^a) \|_F
\end{align*}
Therefore, $\mathbb{E} \| \mathbf{x}^a_\textrm{approx} -\mathbf{x}^a  \|_{\mathbf{Q}^{-1}}^2$ measures the average total error in the posterior mean approximation normalized by the standard deviation of the error covariance $\mathbf{Q}$ in the principal mode directions. 

\paragraph{Adaptive Approximations}
An important consequence of Prop. \ref{thm:opt_approx_scov} and Prop. \ref{thm:opt_approx_mean} is that for a given rank $k$ of the approximations, an optimal strategy can be devised to minimize the normalized error in the posterior error covariance and posterior mean. Indeed, based on Eq. (\ref{F_B_dof})-(\ref{F_Pa_LRdof}) and Eq. (\ref{eq:opt_dof_B_mean})-(\ref{eq:opt_dof_B_mean_FR}), two regimes can be distinguished: if $\lambda_k<1$, then $\mathbf{P}^a_{\mathbf{H}_{\Pi_\mathrm{dof}}}$ and $\mathbf{x}^{a}_{\mathrm{FR}_\mathrm{dof}}$ should be chosen to minimize either the $\mathbf{B}$-normalized or the $\mathbf{P}^a$-normalized errors in $\mathbf{P}^a$ and $\mathbf{x}^a$, respectively; if $\lambda_k$ is significantly greater than $1$, then a sensible strategy would be to use the updates $ \mathbf{P}^{a}_{{\mathrm{\Pi}_\mathrm{dof}}}$ and $\mathbf{x}^{a}_{\mathrm{\Pi}_\mathrm{dof}}$ to approximate $\mathbf{P}^a$ and $\mathbf{x}^a$, respectively. Note that this adaptive update procedure could also be used in the context of non-linear Gauss-Newton methods to improve the convergence rate of the minimization by using an optimal update for the quadratic solution at each linearization step (see Section \ref{improv_4Dvar}).

\section{Practical Implementation}
\label{implementation}
\subsection{Remarks on Eigendecompositions}
\label{rem_svd}
The optimal approximations of the posterior error covariance matrix and the posterior mean described in Section \ref{link_lr_approx} rely on the eigendecompositions of the large $n\times n$ matrices $\mathbf{Q}_\mathrm{var}$ and $\mathbf{Q}_\mathrm{dof}$. In the high-dimensional framework considered in our study, those matrices cannot be formed explicitly, and therefore only matrix-free SVD algorithms can be employed (i.e., algorithms that use only matrix-vector products). In addition to its many theoretical benefits, including its interpretation as the solution of a projected Bayesian problem, the maximum-DOFS approximation associated with $\mathbf{Q}_\mathrm{dof}$ has important computational advantages through the simplification presented in Prop. \ref{simp_dof}.  Indeed, the SVD of $\widehat{\mathbf{H}_p}$ does not involve direct inversions of large matrices \footnote{Although the observation error covariance matrix $\mathbf{R}$ can be high-dimensional and non-diagonal,  in practice covariance matrices and their inverses are constructed implicitly (e.g., \citet{Singh11})}. Assuming the tangent-linear and adjoint model are available, the SVD of $\widehat{\textbf{H}_p}$ can be efficiently performed using matrix-free algorithms such as Lanczos or randomized SVD methods \citep{lanczos50,halko2011finding}. The singular vectors $\mathbf{v}_i$ computed from a truncated SVD of $\widehat{\textbf{H}_p}$ can be used to obtain the singular vectors $\mathbf{w}_i$ used to construct the approximated posterior mean $\mathbf{x}^a_{\mathbf{\Pi}_\mathrm{dof}}$ in Eq. (\ref{eq:update_max_dof_mean_new}), using the relation $\mathbf{w}_i=\mathbf{R}^{-1/2}\mathbf{H}\mathbf{B}^{1/2}\mathbf{v}_i$. Alternatively, a non-symmetric SVD algorithm such as that of Arnoldi \citep{golub2012matrix} or \citet{halko2011finding} (e.g., Alg. 5.1) can be used for direct computation of the SVD of the square-root of $\widehat{\textbf{H}_p}$, $\widehat{\textbf{H}_p}^{1/2}=\mathbf{R}^{-1/2}\mathbf{H}\mathbf{B}^{1/2}$.  As shown in Prop. \ref{lem_CG_optimal_equiv}, computation of the singular vectors $\{\mathbf{w}_i\}$ can be avoided using Eq. (\ref{pg_equiv}). However, in the context of approximated SVD, one has to keep in mind that the equality between Eq. (\ref{eq:update_max_dof_mean_new}) and Eq. (\ref{pg_equiv}) does not strictly hold. Moreover, the singular vectors $\{\mathbf{w}_i\}$ can be useful for information content analysis, as discussed in Section \ref{large_scale}.
The possibility to efficiently compute the optimal approximations associated with the eigendecomposition of $\mathbf{Q}_\mathrm{dof}$ when both the control and the observation spaces are high-dimensional is in contrast with the optimal approximations associated with $\mathbf{Q}_\mathrm{var}$ (i.e., $\mathbf{x}^a_{\mathrm{var}}$ and $\mathbf{P}^a_{\mathrm{var}}$), since algebraic simplifications similar to Prop.  \ref{simp_dof} do not exist for $\mathbf{Q}_\mathrm{var}$. In this case the $p\times p$ matrix of innovation statistics $\mathbf{HBH}^T+\mathbf{R}$ needs to be formed and inverted, and a matrix-free algorithm can then be used to compute the SVD of $\mathbf{Q}_\mathrm{var}$ (see Rem. \ref{rem_Qvar}).

\begin{rem}
\label{rem_Qvar}
 In the context of atmospheric source inversion, a typical case where the number of observations $p$ is usually small enough to allow direct inversion of $\mathbf{HBH}^T+\mathbf{R}$ and compute the SVD of $\mathbf{Q}_\mathrm{var}$ is the inversion of (possibly high-dimensional) sources from a sparse network of \textit{in situ} observations. In contrast, satellite-based inversions, for which $p$ can be very large, may not allow $\mathbf{HBH}^T$ to be explicitly formed, unless the dataset is reduced prior to the inversion (e.g., using an aggregation scheme).
 \end{rem}
\begin{rem}
 
In the case where the matrix of innovation statistics $\mathbf{HBH}^T+\mathbf{R}$ can be inverted explicitly, the full-dimensional analysis $\mathbf{x}_a$ in Eq. (\ref{eq2}) can be computed analytically even for control vectors with very large dimensions $n$ (as long as the tangent-linear and adjoint models are available). However, even in that case, computing optimal approximations (based on either $\mathbf{Q}_\mathrm{var}$ or $\mathbf{Q}_\mathrm{dof}$) is still useful in order to quantify the information content of the inversion, since the posterior error covariance and the model resolution matrices are both of dimension $n\times n$. To this aim, Eq. (\ref{eq:update_max_dof_cov_new}), (\ref{LR_dof_post_err}), (\ref{var_post_err}) and (\ref{eq:update_max_dof_mod_res_new}) can be used to efficiently extract subsets of elements (e.g., the entire diagonal) from the approximated posterior error covariance or model data resolution matrices.
\end{rem}

\subsection{Randomized Singular Value Decomposition}
\label{rand_svd_sec}
The most widely used matrix-free SVD algorithms are based on the Lanczos method, which computes the dominant eigenvectors and eigenvalues of an Hermitian matrix using Krylov subspace iterations \citep{golub2012matrix}. Recently, randomized SVD methods have attracted interest due to their proven accuracy and  high scalability for a large variety of problems. In this Section, we describe a randomized SVD algorithm, some of its theoretical properties, as well as a practical probabilistic error estimate for the approximation. The use of this randomized SVD method allows critical improvement in computational performance that we shall exploit in a numerical experiment in the context of large-scale atmospheric source inversions (see Section \ref{num_exp}).  
\subsubsection{Principle}
\label{rand_svd}
Randomization algorithms are powerful and modern tools to perform matrix decomposition. Some of their key advantages compared to standard Krylov subspace methods are their inherent stability and the possibility to massively parallelize the computations.  Recently, \citet{halko2011finding} presented an extensive analysis of the theoretical and computational properties of randomized methods to compute approximate matrix decomposition, including low-rank SVDs. The approach relies on the ability to efficiently approximate the range of a matrix $\mathbf{A}$ using a relatively small sample of image vectors $\{\mathbf{y}^{(i)}=\mathbf{A}\boldsymbol{\omega}^{(i)},\,i=1,..,k\}$, where the input vectors $\boldsymbol{\omega}^{(i)}$ are independent vectors with i.i.d.\ Gaussian entries. The quality of the approximation for the range of $\mathbf{A}$ can be objectively determined by evaluating the spectral norm of the difference between the original matrix and its projection onto the subspace defined by the random images, i.e., one wants: 
\begin{align}
\label{range_approx}
\| (\mathbf{Id}-\mathbf{QQ}^T)\mathbf{A}\|\le\epsilon,
\end{align}
where $\epsilon$ is some tolerance level, $\|.\|$ is the spectral norm, and $\mathbf{Q}$ is the matrix whose columns form an orthonormal basis of the subspace spanned by $\{\mathbf{y}^{(i)},\, i=1,...,k\}$. Once a satisfactory level of precision for the range has been reached, the SVD can be performed in the reduced space (defined by $\mathbf{Q}$) using dense matrix algebra, and the resulting singular vectors projected back onto the original space. Here we describe an algorithm especially adapted to the treatment of large Hermitian matrices involving expensive PDE solvers (in our case, the transport model $\mathbf{H}$ and its adjoint $\mathbf{H}^T$). The reader is referred to \citet{halko2011finding} for a complete review and explanation of those techniques in other contexts. The following algorithm allows one to compute an approximate truncated SVD of an Hermitian matrix using the randomized approach \citep{halko2011finding}. It uses only matrix-vector products, and is highly parallelizable. 
\begin{algorithm}
\caption{One-Pass Eigenvalue Decomposition}
 \label{alg_one_pass}
 Given $\mathbf{A}$ a $n\times n$ Hermitian matrix and $\boldsymbol{\Omega}$ a $n\times k$ random Gaussian test matrix ($k<<n$):
\begin{algorithmic}[1]
 \STATE Form the $n\times k$ matrix $\mathbf{Y}=\mathbf{A}\boldsymbol{\Omega}$
 \STATE Construct a $n\times k$ matrix $\mathbf{Q}$ whose columns form an orthonormal basis for the range of $\mathbf{Y}$, using, e.g., a QR or SVD factorization 
 \STATE Form the $k\times k$ matrix $\mathbf{B}\equiv(\mathbf{Q}^T\mathbf{Y})(\mathbf{Q}^T\mathbf{\Omega})^{-1}(\approx\mathbf{Q}^T\mathbf{A}\mathbf{Q})$
 \STATE Perform the SVD $\mathbf{B}=\mathbf{Z}\boldsymbol{\Lambda}\mathbf{Z}^T$
 \STATE Form the $n\times k$ matrix column of singular vectors $\mathbf{V}=\mathbf{QZ}$
 \STATE One has the approximation $\mathbf{A}\approx\mathbf{V}\boldsymbol{\Lambda}\mathbf{V}^T$
 \end{algorithmic}
\end{algorithm}

In Algorithm \ref{alg_one_pass}, the products $\mathbf{A}\boldsymbol{\Omega}[:,i]$ in step 1 generating the columns $\mathbf{Y}[:,i]$ can be all performed in parallel, which renders this method highly scalable. In our case, the matrix-vector product $\mathbf{A}\boldsymbol{\Omega}[:,i]$ amounts to integrating a PDE solver, which requires intensive computations (e.g., $\mathbf{A}=\mathbf{B}^{1/2}\mathbf{H}^T\mathbf{R}^{-1}\mathbf{H}\mathbf{B}^{1/2}$ for the maximum DOF approximation). Therefore, assuming $k<<n$, the cost of the one-pass algorithm is largely dominated by step 1, the remaining steps involving dense linear algebra in small dimension.

In practice, the number of input samples $k$ is increased until relation (\ref{range_approx}) is verified. The spectral norm of the estimation error is not directly evaluated (doing so would entail computing the SVD of a large $n\times n$ matrix, which is precisely what we want to approximate), but can be estimated a posteriori using an inexpensive probabilistic approach (see Section \ref{error_analysis}).

\subsubsection{Error Analysis}
\label{error_analysis}
The precision of the approximate SVD generated by Algorithm \ref{alg_one_pass} depends on the error in the estimation of the range in  (\ref{range_approx}), which itself depends (for a given number of samples) on the shape of the singular value spectra (i.e., fast or slow decay) and the dimension $n$. The following result demonstrates this dependence by providing a bound for the average spectral error \citep{halko2011finding}:
\begin{prop}[\textbf{Average Spectral Error}] 
\label{average_error}
Let $\mathbf{A}$ be a $n\times n$ matrix with eigenvalues $\sigma_1\le...\le\sigma_n$. Let $\mathbf{Q}$ be a $n\times(k+q)$ orthonormal basis that approximates the range of $\mathbf{A}$ ($k+ q << n$), generated from steps 1-2 of Algorithm \ref{alg_one_pass}. One has the following error bound:
\begin{align}
\mathbb{E}\left ( \|\mathbf{A}-\mathbf{QQ}^T\mathbf{A} \| \right) \le \left[ 1+\frac{4\sqrt{k+q}}{q-1}\sqrt{n} \right] \sigma_{k+1}
\end{align}
\end{prop}
In Prop. \ref{average_error}, $k$ is the targeted rank for the approximation, and $l$ is an oversampling parameter. Note that the smallest possible spectral error for a rank-$k$ approximation is $\sigma_{k+1}$ [ref]. This formula shows that the average spectral error lies within a small polynomial factor of the theoretical minimum. In particular, one also sees that for very large $n$ the bound is not significantly modified when $n$ is increased (factor $\sqrt{n}$). Moreover, increasing the oversampling parameter $q$ rapidly decreases the amplification factor $\frac{4\sqrt{k+q}}{q-1}$. As a result, in practice using an oversampling parameter $q\approx5$ yields very good results.

In order to efficiently compute an estimate of the spectral error (\ref{range_approx}), one can use the following result, which provides a probabilistic bound based on sample error estimates available for free \citep{halko2011finding}:
\begin{prop}[\textbf{(Cost-Free) Posterior Probabilistic Error Bound For Range Approximation}] 
\label{prob_error}
Let $\mathbf{A}$ be a $n\times n$ matrix. Let $\mathbf{Q}$ be an orthonormal basis that approximates the range of $\mathbf{A}$, generated from steps 1-2 of Algorithm \ref{alg_one_pass}, and a set of $l$ independent vectors with i.i.d. Gaussian entries $\{\boldsymbol{\omega}^{(i)},\, i=1,...,l\}$. One has the following error bound:
\begin{align}
\label{bound_prob}
 &\|\mathbf{A}-\mathbf{QQ}^T\mathbf{A} \| \le 10\sqrt{\frac{2}{\pi}} \max_{i=1,..,l}\|(\mathbf{A}-\mathbf{QQ}^T\mathbf{A})\boldsymbol{\omega}^{(i)} \|, \\ \nonumber
 &\text{with a probability }1-10^{-l} 
 \end{align}
\end{prop}
Based on this result, one sees that using only two samples ($l=2$) to estimate the right-hand side of (\ref{bound_prob}) provides a bound on the spectral error with a probability 0.99. Note that samples of the form $\mathbf{A}\boldsymbol{\omega}^{(i)}$ are computed at step 1 of Alg. \ref{alg_one_pass} to construct $\mathbf{Q}$. Therefore, a practical implementation of the a posteriori error bound estimate would be to use two samples out of the total set of samples generated at step 1 of Alg. \ref{alg_one_pass} to derive the probabilistic bound in (\ref{bound_prob}), while the remaining samples are used to compute $\mathbf{Q}$. An adaptive range finder method using similar principles to build a matrix $\mathbf{Q}$ associated with a desired spectral error tolerance $\epsilon$ can also be found in \citep{halko2011finding} (see Alg. 4.2).   

\section{Numerical Illustrations}
\label{num_exp}
In this Section we first illustrate the theoretical results obtained in Section \ref{theory} using a small inverse problem ($n=300$). This enables exact computation of the SVD involved in the optimal approximations of the Bayesian problem, and direct evaluation of the performance of the associated posterior mean and posterior error covariance estimates against the true solutions. We then test the performance of the algorithm for a large-scale experiment using the randomized SVD method described in Section \ref{rand_svd}.

\subsection{Atmospheric Source Inversion Problem}
\label{inv_pb_test}
Our numerical experiments are carried out in the context of an atmospheric transport source inversion problem. The setup consists of an Observation Simulation System Experiment (OSSE) where pseudo-observations of methane columns (XCH$_4$) from a Short Wave Infrared (SWIR) instrument in low-earth orbit are generated and used to optimize randomly perturbed prior methane fluxes over North America. A nested domain with spatial resolution ($0.5^{\circ} \times  0.7^{\circ}$) is used and one scaling factor is optimized for each grid-cell for the month of July 2008, which corresponds to an initial control space of dimension $n=151\times121=18,271$. A uniform prior error standard deviation of 40\% is assumed for the CH$_4$ fluxes throughout the domain, with no spatial error correlations (diagonal $\mathbf{B}$ matrix), and the observational error standard deviations are uniformly set to 8 ppb, with no spatial or temporal correlations.  More information about the general configuration of this type of OSSE experiment can be found in \citet{bousserez2016constraints}. The atmospheric transport (forward model, $H$) is simulated using GEOS-Chem, which is an offline atmospheric transport model widely used in the atmospheric chemistry community. The model configuration we used is described in \citet{wecht2014spatially}. The adjoint of GEOS-Chem, also employed in our experiment, is described in \citet{Henze07} and has been extensively used in previous sensitivity and inverse modeling studies \citep{kopacz2009comparison,jiang2011quantifying,xu2013constraints,wells2015simulation}. Figure \ref{fig:map_prior} shows a map of the prior CH$_4$ emissions over North America. Emissions are geographically contrasting and highly variable. Since in our setup the prior error standard deviation is proportional to the prior emission magnitude, a similar high spatial variability is obtained for the prior error variances.

Note that the uniform diagonal $\mathbf{B}$ matrix used in our setup implies that the matrices $\mathbf{Q}_\text{dof}$ and $\mathbf{Q}_\text{var}$ as well as their associated approximations will coincide modulo a scalar multiplication (this is obviously not the case for a general $\mathbf{B}$ matrix). This simplified configuration allows a clear interpretation of the results in term of observational constraints (e.g., all posterior error correlations are due to the observations only), while demonstrating the theoretical properties and numerical efficiency of the optimal approximations. As discussed in Section \ref{rem_svd}, the approximations based on the solution of the maximum-DOFS projection have clear theoretical and practical advantages compared to the approximations derived from the eigendecomposition of  $\mathbf{Q}_\text{var}$. Therefore, contrasting the characteristics of these two types of approximations in the general case (i.e., non-diagonal $\mathbf{B}$) was not viewed as a priority for our analysis. 

\subsection{Convergence Analysis for a Small Problem}
\label{toy_prob}
The convergence properties of the optimal posterior mean and posterior error covariance approximations are investigated using a reduced version of the source inversion problem described in Section \ref{inv_pb_test}. In this experiment, the control vector dimension is reduced to $n=300$ by selecting the model grid-cells which correspond to the first 300 highest gradients of the 4D-Var cost function (\ref{eq3}) with respect to the emission scaling factors. With this setup, the Hessian of the cost function is explicitly calculated using the finite-difference method combined with adjoint model integrations. More specifically, the following formula is used to estimate each element of the Hessian matrix:
\begin{eqnarray}
\label{hessian_fd}
\widehat{\textbf{H}}_{i,j}\approx\frac{(\nabla J(\textbf{x}+\epsilon_i)-\nabla J(\textbf{x}))_j}{\epsilon} ,
\end{eqnarray}
where $\epsilon_i = (\epsilon \delta_{i,k})_k$, $\delta$ is the Kronecker delta, and $\epsilon$ is a small real number. For this experiment $\epsilon=0.01$, which corresponds to a 1\% perturbation of the CH$_4$ flux for a particular grid cell. Since $\widehat{\textbf{H}}$ is symmetric, this calculation requires $n+1$ gradient calculations, which corresponds to $301$ forward model integrations and $301$ adjoint model integrations for the reduced problem. The fact that these gradient calculations can be performed in parallel makes the computation efficient. The inverse of the Hessian matrix $\widehat{\textbf{H}}$ provides the posterior error covariance matrix $\mathbf{P}^a$, which is used to compute the exact posterior mean using formulas (\ref{post_update2}).
The prior-preconditioned Hessian $\widehat{\mathbf{H}}_p \equiv \mathbf{B}^{1/2}\mathbf{H}^T\mathbf{R}^{-1}\mathbf{H}\mathbf{B}^{1/2}$ is also computed explicitly from the Hessian finite-difference estimate (\ref{hessian_fd}) using: $\widehat{\textbf{H}}_p=\mathbf{B}^{-1/2}(\widehat{\textbf{H}}-\mathbf{B})\mathbf{B}^{-1/2}$.

Figure \ref{fig:toy_spectra} shows the singular value spectra of the prior-preconditioned Hessian of the reduced inverse problem. The spectra shows a fast decrease of the first 20 singular values (by an order of magnitude), followed by a slow decrease. The basis for the rank-$k$ maximum-DOFS projection is made of the first $k$ singular vectors of $\widehat{\textbf{H}}_p$, which are computed exactly here. Moreover, according to Corollary \ref{cor:dofs}, the DOFS of the inversion for a rank-$k$ optimal projection is equal to $\sum_{i=1}^k \lambda_i/(1+\lambda_i)$, where the $\{\lambda_i,\, i=1,...,k\}$ are the first $k$ singular values of $\widehat{\textbf{H}}_p$. The DOFS for our reduced inverse problem is $\sim$43. Figure \ref{fig:perc_dof} shows the $\mathbf{P}^a$-weighted error in the posterior error covariance matrix approximations $\mathbf{P}^a_{\mathbf{H}_{\Pi_\text{dof}}}$ and $\mathbf{P}^a_{\Pi_\text{dof}}$ defined by Eq. (\ref{F_Pa_dof}) and (\ref{F_Pa_LRdof}), respectively. For the sake of simplicity only the $\mathbf{P}^a$-weighted errors are considered here. As expected, for small ranks $k$ (here, for $k<3$) the low-rank approximation $\mathbf{P}^a_{\Pi_\text{dof}}$ is associated with a smaller error than the low-rank update approximation $\mathbf{P}^a_{\mathbf{H}_{\Pi_\text{dof}}}$. However, the approximation error associated with $\mathbf{P}^a_{\mathbf{H}_{\Pi_\text{dof}}}$ shows an exponential decrease and is about an order of magnitude smaller than the approximation error of $\mathbf{P}^a_{\Pi_\text{dof}}$ for $k>50$.  

Also shown on Fig. \ref{fig:perc_dof} is the relative error in the DOFS approximation for solution of the maximum-DOFS projection ($\mathbf{A}_{\Pi_\text{dof}}$), as well as the relative error in the total variance approximations $\mathrm{Tr}(\mathbf{P}^a_{\mathbf{H}_{\Pi_\text{dof}}})$ and $\mathrm{Tr}(\mathbf{P}^a_{\Pi_\text{dof}})$, as a function of the rank $k$ of the approximation. The results for the DOFS show that more than 80\% of the information content of the inversion is captured by the first 120 modes, with a fast decrease of the error for the first third of the spectra followed by a slower decrease. Interestingly, the low-rank update posterior error covariance approximation $\mathbf{P}^a_{\mathbf{H}_{\Pi_\text{dof}}}$ shows much better performances than the low-rank approximation $\mathbf{P}^a_{\Pi_\text{dof}}$, even for small values of the rank $k$. These results show that, for our experiment, the low-rank update approximation $\mathbf{P}^a_{\mathbf{H}_{\Pi_\text{dof}}}$ should be chosen when estimating posterior error variances. Overall, our numerical tests demonstrate the fast convergence of the low-rank update approximation $\mathbf{P}^a_{\mathbf{H}_{\Pi_\text{dof}}}$ for different information content metrics.

 In addition to analyzing the convergence globally, it may be of interest to analyze the local behavior of this approximation. Figures \ref{fig:map_var} and \ref{fig:scatterplot_var} illustrate the convergence of the approximated error variances for $\mathbf{P}^a_{\mathbf{H}_{\Pi_\text{dof}}}$ for all control vector elements. Figure \ref{fig:map_var} represents the spatial distribution of the true posterior error variances as well as the approximated posterior error variances for the ranks $k=40$, $k=100$, and $k=200$, while Fig. \ref{fig:scatterplot_var} shows the corresponding scatterplots and linear regression fits for each of those ranks. A very good accuracy of the approximated posterior error variances is observed for all ranks. This is further confirmed by the linear regression analysis, with a Pearson correlation coefficient of about 1 and an almost perfect regression line (1:1) for all ranks. From Fig. \ref{fig:scatterplot_var} it is evident that increasing the rank of the approximation above 40 does not significantly improve the results, which is consistent with the results obtained for the total error variance in Fig. \ref{fig:perc_dof}. 

Similarly to the posterior error covariance approximations, we now investigate the convergence properties of the posterior mean approximations described in Prop. \ref{thm:opt_approx_mean}. Again, for the sake of simplicity only the $\mathbf{P}^a$-weighted errors are considered here. Figure \ref{fig:bayes_risk} shows the expectancy (or average) of the total $\mathbf{P}^a$-weighted error in the approximated posterior mean (or Bayes risk) for the solution of the maximum-DOFS projection $\mathbf{x}^a_{\mathbf{\Pi}_\mathrm{dof}}$ and for the full-rank posterior mean approximation $\mathbf{x}^{a}_{\mathrm{FR}_\mathrm{dof}}$, as a function of the rank. The results for one single realization of the prior and the observations are also shown. The Bayes risk for each rank $k$ is calculated using Eq. (\ref{eq:opt_dof_Pa_mean}) and (\ref{eq:opt_dof_B_mean_FR}), while the $\mathbf{P}^a$-weighted posterior mean error for one single realization is calculated by explicitly computing the error for one particular instance of the prior and observation probability distributions. As expected from the theory, for small values of the rank ($k<10$) the error associated with the low-rank projection $\mathbf{x}^a_{\mathbf{\Pi}_\mathrm{dof}}$ is much smaller (by several order of magnitude) than the error associated with the full-rank approximation $\mathbf{x}^{a}_{\mathrm{FR}_\mathrm{dof}}$. However the full-rank approximation $\mathbf{x}^{a}_{\mathrm{FR}_\mathrm{dof}}$ becomes rapidly the most accurate (for $k>10$), with an exponential decay of the error. The results for one realization of the prior and observation statistics also suggest very small deviations of the $\mathbf{P}^a$-weighted error from its mean behavior, which is consistent with previous findings in \citet{spantini2015optimal}.

Similar to the posterior error analysis, it may be of interest to analyze locally the convergence of the approximated posterior mean, i.e., in our case, the spatial distribution of the posterior flux increments. Figure \ref{fig:lr_incr_toy} shows the spatial distribution of the true and approximated posterior flux increments for the solution of the maximum-DOFS projection $\mathbf{x}^a_{\mathbf{\Pi}_\mathrm{dof}}$ and the full-rank posterior increment approximation $\mathbf{x}^{a}_{\mathrm{FR}_\mathrm{dof}}$, for the ranks $k=5$ and $k=100$. The contrast between the low-rank nature of the maximum-DOFS solution $\mathbf{x}^a_{\mathbf{\Pi}_\mathrm{dof}}$ and the full-rank nature of $\mathbf{x}^{a}_{\mathrm{FR}_\mathrm{dof}}$ is evident for $k=5$. As shown in Fig. \ref{fig:bayes_risk}, for a rank $k=5$ the maximum-DOFS solution $\mathbf{x}^a_{\mathbf{\Pi}_\mathrm{dof}}$ is associated with a smaller $\mathbf{P}^a$-normalized error (Bayes risk) than $\mathbf{x}^{a}_{\mathrm{FR}_\mathrm{dof}}$. Although the full-rank approximation $\mathbf{x}^{a}_{\mathrm{FR}_\mathrm{dof}}$ better captures the true posterior increment distribution over large areas (e.g., Canada) compared to $\mathbf{x}^a_{\mathbf{\Pi}_\mathrm{dof}}$, those areas are associated with large posterior errors (see Figure \ref{fig:map_var}), and thus are attributed less weight in the $\mathbf{P}^a$-normalized posterior mean score than regions over the east of the US domain associated with small posterior errors. The better performance of the low-rank maximum-DOFS solution $\mathbf{x}^a_{\mathbf{\Pi}_\mathrm{dof}}$ over those regions with small posterior errors explains its overall better score. Consistent with our previous analysis (see Fig. \ref{fig:bayes_risk}), for $k=100$, the full-rank posterior increment approximation $\mathbf{x}^{a}_{\mathrm{FR}_\mathrm{dof}}$ better reproduces the spatial distribution of the true posterior increment across the whole domain, with now similar performances as the maximum-DOFS posterior increment $\mathbf{x}^a_{\mathbf{\Pi}_\mathrm{dof}}$ over regions associated with small posterior errors (see, e.g., the eastern US).

\subsection{Performance for a Large-Scale Experiment}
\label{large_scale}
In this Section we illustrate the efficiency of combining the optimal approximation methods with the randomized SVD algorithm by applying this approach to the full-dimensional source inversion problem defined in Section \ref{inv_pb_test}. Since the dimension of the control vector is now $n=151\times121=18,271$, explicitly forming the prior-preconditioned Hessian matrix and computing its SVD using dense linear algebra is not practical. Therefore, we use the One-Pass SVD algorithm described in Alg. \ref{alg_one_pass} to compute the eigenvectors and eigenvalues of $\widehat{\textbf{H}}_p$. The computation can be performed in parallel, since the $k$ matrix-vector products of $\widehat{\textbf{H}}_p$ with the columns of the $n\times k$ random Gaussian test matrix $\boldsymbol{\Omega}$ can be performed independently. In our context, the computational cost of the method is largely dominated by the Hessian matrix-vector products used to build the basis for the reduced space $\mathbf{Q}$ (see step 1 of Alg. \ref{alg_one_pass}), since each of them amounts to integrating one tangent linear and adjoint model. The shape of the approximated eigenvalue spectra of $\widehat{\textbf{H}}_p$, not shown here, is comparable to the one obtained with the small problem of Section \ref{toy_prob}, with an exponential decay followed by a long flat tail, typical of severely underconstrained inverse problems.

Figure \ref{fig:prob_err} shows the probabilistic spectral error bound calculated using Prop. \ref{prob_error}, as a function of the number of samples (i.e., the number of columns $k$ of $\boldsymbol{\Omega}$) used in the randomized SVD estimate. A fast decrease of the spectral error bound is observed for the first 100 samples, followed by a smaller decrease between 100 and 400 samples. It appears clearly that increasing the number of samples beyond 200 does not significantly reduce the spectral error bound ($\sim$3), which suggests a reduced basis of 200 vectors would provide a reasonably good approximation of the range of  $\widehat{\textbf{H}}_p$ for our problem. Here we used those 400 samples to form an orthonormal basis $\mathbf{Q}$ of the subspace in which the low-rank SVD was performed. We note that no significant differences were found in the approximated posterior error covariances or the approximated posterior mean updates between computations using 200 or 400 samples.

 The analytical expressions for the posterior error covariance approximations, $\mathbf{P}^a_{\mathbf{H}_{\Pi_\mathrm{dof}}}$ (Eq. (\ref{LR_dof_post_err})) and $\mathbf{P}^a_{\Pi_\mathrm{dof}}$ (Eq. (\ref{pa_pi_dof})), and for the maximum-DOFS model resolution matrix approximation, $\mathbf{A}_{\Pi_\mathrm{dof}}$ (Eq. (\ref{A_pi_dof})), can be used to efficiently compute any matrix-vector product involving those matrices, and in particular to extract subsets of their elements. As an example, Fig. \ref{fig:diag_avk} shows the diagonal of the model resolution matrix $\mathbf{A}_{\Pi_\text{dof}}$ of the rank-$400$ maximum-DOFs projection, which represents the observational constraints for each flux estimate. More specifically, each value on the diagonal quantifies the relative contribution of the observations to the total information content, with respect to the prior information. As discussed previously, it also corresponds to the sensitivity of the posterior mean flux to its true value (see Eq.(\ref{eq6})). The sum of the diagonal elements of $\mathbf{A}_{\Pi_\mathrm{dof}}$ is the DOFS for the projected problem, which is equal to 46. Therefore, 46 independent pieces of information (i.e., flux modes) can be obtained from the observational constraints. As explained in Section \ref{opt_proj}, the eigenvectors of the model resolution matrix represent precisely the modes that are independently constrained by the observations (with respect to the prior information), and are given by $\{\mathbf{B}^{1/2}\mathbf{v}_i,\, i=1,...,k\}$, where the $\{\mathbf{v}_i,\, i=1,...,k\}$ are the first $k$ eigenvectors of the prior-preconditioned Hessian $\widehat{\mathbf{H}}_p$. Figure \ref{fig:princ_modes} shows the 1st, 2nd and 5th eigenvectors of $\widehat{\mathbf{H}}_p$, which in the case of a uniform diagonal $\mathbf{B}$ matrix are therefore scalar multiples of the eigenvectors of the model resolution matrix $\mathbf{A}_{\Pi_\mathrm{dof}}$. Additionally, the corresponding modes in observation space, that is, the left singular vectors $\{\mathbf{w}_i,\, i=1,...,k\}$ of the square-root of $\widehat{\textbf{H}}_p$, are also shown and were computed by propagating each singular vectors $\mathbf{v}_i$ using the relation $\mathbf{w}_i=\mathbf{R}^{-1/2}\mathbf{H}^T\mathbf{B}^{1/2}\mathbf{v}_i$ (see Section \ref{rem_svd}). As shown, the relative contribution of the observations to the information content of each posterior mode is greater than 80\% for all three modes. The three modes correspond to clearly distinct geographical patterns.  The results show that the satellite observations allow one to constrain CH$_4$ emissions at high (almost grid-scale) spatial resolution over the Toronto (1st mode) and (to a lesser extent) the Los Angeles areas, but that constraints between the New York and Appalachian regions tend to be correlated. These figures illustrate the potential of the maximum-DOFS low-rank projection problem as a robust framework to objectively characterize the information content of an inversion.

 In addition to providing useful tools to analyze fundamental properties of the inverse problem, the joint use of the optimal approximations with the randomized SVD approach provides a fast method to compute the posterior solution. The efficiency of the approach can be illustrated by comparing the number of iterations required for convergence of a standard iterative minimization routine with the number of samples needed for the randomized SVD to provide similar results. Figure \ref{fig:comp_eff} shows the posterior scaling factor increments obtained from a BFGS minimization with 40 iterations and from our adaptive posterior mean approximation method using randomized SVD computations with 200 samples. In this case, since $\lambda_{200}<1$, the approximated posterior mean chosen is $\mathbf{x}^{a}_{\mathrm{FR}_\mathrm{dof}}$ (see Section \ref{interpretation_optimal}). The results show that a parallel implementation of the randomized SVD approach using only 200 samples provides posterior flux increments comparable to 40 iterations of the BFGS algorithm over most areas (with differences rarely exceeding $0.10$). The moderate number of samples for this test allowed us to run all independent simulations in the randomized SVD algorithm simultaneously, resulting in a wall time computation of $\sim$72 mins. Comparatively, the BFGS minimization algorithm required a wall time 40 times longer than the randomized SVD, with a total of 48 hours. These results clearly illustrate the benefit of using the randomized SVD approach to drastically reduce the computational cost of large-scale inverse problems.
 
  As discussed in \citet{halko2011finding}, the efficiency of randomized range approximation methods is driven by the rate of decrease of the singular value spectra, and is therefore problem-dependent. As suggested by Thm. \ref{average_error}, a sharp decrease of the singular values is associated with a more accurate estimate of the range of the matrix. We note that many inverse problems in geophysics, including atmospheric source inversion and weather data assimilation problems, generally exhibit this desirable behavior. However, for inverse problems whose singular values decay slowly, randomized power-iteration methods can still be used, which provides much more accurate results while mitigating the cost of the SVD compared to standard Krylov subspace methods. Their efficiency stems from exploiting parallel implementation of randomized estimates combined with a small number of power-iterations \citep{halko2011finding}.
    
\section{Link With Methods in Data Assimilation}
\label{link_da}
In this Section we discuss some interesting links between the theory developed in this paper and methods used in operational data assimilation centers. In Section \ref{sec:incr_4dvar}, we recall the incremental 4D-Var algorithm and its different square-root formulations, and explain how it is closely related to the maximum-DOFS low-rank projection. An overlooked issue related to the representativeness error when singular square-roots are used is also discussed. In Section \ref{improv_4Dvar}, we propose an improved implementation of the incremental 4D-Var method, combining the optimality results of Section \ref{link_lr_approx} with the randomized SVD technique described in Section \ref{rand_svd_sec}.

\subsection{Incremental 4D-Var}
\label{sec:incr_4dvar}
\subsubsection{Principle}
Most large-scale inverse problems, such as those encountered in atmospheric data assimilation, are solved by minimizing the least-squares cost function (\ref{eq3}). Operational DA centers often use the incremental 4D-Var technique (e.g., \citet{courtier94}), which consists of a sequence of quadratic CG minimizations based on linearizations of the forward model. Each quadratic minimization is called an inner-loop, and is often performed using  simplified tangent-linear and adjoint models. After each quadratic minimization, the updated posterior mean is propagated through the non-linear model $H$ and compared to the observations. This step, called the outer-loop, is repeated until a convergence criterion is reached. 

In its original form, the linear (or quadratic) least-squares problem can be severely ill-conditioned, i.e., the gap between the eigenvalues of the Hessian of the cost function can be very large. This can prevent fast convergence of the CG algorithm. One widely used remedy is preconditioning, which consists of solving the least-squares problem in a basis where the eigenvalues of the Hessian of the cost function are as tightly clustered as possible. In the ideal case, where the Hessian is the identity matrix, the convergence is obtained in one iteration, which is why the matrix of change of basis (or preconditioner) should be constructed so as to best approximate a square-root of the Hessian. In practice, only limited prior knowledge of the Hessian is available (note that knowing this matrix perfectly amounts to solving the least-squares problem). Since the Hessian can be written as a positive low-rank update to the inverse prior error covariance matrix ($\nabla ^2 J(\mathbf{x}^a)=\mathbf{B}^{-1}+\mathbf{H}^T\mathbf{R}^{-1}\mathbf{H}$), a common approach is to use the square-root of $\mathbf{B}$ as a preconditioner, which one shall note $\mathbf{L}$, such as $\mathbf{B}=\mathbf{L}\mathbf{L}^{T}$. This change of variable leads to the following least-squares problem for each inner-loop of the incremental 4D-Var algorithm:
\begin{eqnarray}
 \label{incr_4dvar}
&\min_\mathbf{z} \tilde{J}(\mathbf{z}) =\frac{1}{2}(\mathbf{d}-\mathbf{HLz})^T\mathbf{R}^{-1}(\mathbf{d}-\mathbf{HLz})+\frac{1}{2}\mathbf{z}^T\mathbf{z} \\
 &\mathbf{B}=\mathbf{LL}^T, \nonumber \\
 & \mathbf{x}=\mathbf{L}\mathbf{z},\nonumber
\end{eqnarray}
where $\mathbf{x}$ is the increment and $\mathbf{d}$ the innovation.
\subsubsection{Square-Root Formulations}
\label{sq_formulations}
There are two main types of square-root formulations used in current 4D-Var DA systems.
The Control Variable Transform (CVT) consists of modeling the prior error covariance matrix $\mathbf{B}$ implicitly, by projecting the control vector onto a basis where its elements are uncorrelated (i.e., their covariance matrix is the identity matrix) \citep{Bannister08}. In this approach, the change of basis is constructed using operators that impose, e.g., dynamical balances, or filtering out modes of variability that one does not wish to include in the posterior update. In addition to allowing an implicit representation of a matrix that would otherwise be too large to store in computer memory ($n\sim10^8$ in operational NWP), this technique provides an efficient preconditioning of the variational minimization\citep{courtier94}.
With the CVT technique, the product of all the operators that decorrelate the control variables forms the square-root $\mathbf{L}$. 

The other approach to the square-root formulation is the ensemble-based method. An ensemble-based square-root consists of modeling the matrix $\mathbf{B}$ using an ensemble of forward model perturbation trajectories centered around their mean, i.e.:
\begin{eqnarray}
\mathbf{B}_{\text{ens}}=\textbf{W} \textbf{W}^T, \qquad \textbf{W}=\frac{1}{\sqrt{K-1}}(\textbf{w}_1-\overline{\textbf{w}},...,\textbf{w}_K-\overline{\textbf{w}}), 
\end{eqnarray}
where $(\textbf{w}_1,\textbf{w}_2,...,\textbf{w}_K)$ represents an ensemble of $K$ perturbed states and $\overline{\textbf{w}}$ the mean of the ensemble. This approach is used, e.g., in ensemble-based DA methods such as the Ensemble Adjustment Kalman filter (EAKF) \citep{Anderson01}, or the Ensemble-Variational (EnVar) algorithm \citep{lorenc2003potential}. Localization techniques are usually required to mitigate undersampling noise in the approximated error variances and error correlations  (e.g., \citet{menetrier2015linear}).
 In practice, it may be desirable to implement a hybrid formulation, in which $\textbf{B}$ is a weighted average of a static full-rank matrix $\textbf{B}_{\text{stat}}$ that represents, e.g., a climatology, and an ensemble-based flow-dependent matrix $\textbf{B}_{\text{ens}}$, so that $\mathbf{B}$ is eventually modeled as:
\begin{eqnarray}
\label{Schur_product}
\mathbf{B}_\mathrm{hyb} &=&(1-\beta) \textbf{B}_{\text{stat}}+ \beta \textbf{B}_{\text{ens}} \circ \mathbf{C},
\end{eqnarray}
where $\mathbf{C}$ is made of compactly supported and space-limited covariance functions \citep{gaspari1999construction}, $\circ$ represents the Schur product between two matrices, and $\beta$ is a scalar verifying $0<\beta<1$. 

The matrix $\mathbf{B}_\mathrm{hyb}$ can be defined only at the initial time, $t_0$, of the variational window, which is the En4DVar method, or throughout the assimilation window (4D matrix), which is the 4DEnVar (or EnVar) approach \citep{desroziers20144denvar,lorenc2003potential}. Based on this very general square-root formulation of incremental 4D-Var, we will now establish some theoretical links between optimizations performed in current operational DA systems and the optimal low-rank approximation approaches presented in this study.

\subsubsection{Preconditioned Conjugate-Gradient as an Optimal Low-Rank Projection}
\label{cg_optimal}
The following proposition establishes the theoretical equivalence (i.e., modulo approximation errors) between the posterior mean of the maximum-DOFS rank-$k$ projection and the solution obtained after $k$ iteration of the preconditioned CG algorithm. It stems from the close relationship between the CG algorithm and the Lanczos eigenvalue decomposition \citep{meurant2006lanczos}. 
\begin{prop}[\textbf{Optimality of Preconditioned Conjugate-Gradient Minimizations}]
\label{prop:eq_cg}
Let us consider the $\mathbf{L}$-preconditioned least-squares minimization problem (\ref{incr_4dvar}), where $\mathbf{B}=\mathbf{LL}^T$ and $\mathrm{rank}(\mathbf{L})=\mathrm{rank}(\mathbf{B})=n$. One notes $\mathbf{x}^a_{\mathbf{\Pi}_\mathrm{dof}}$ the posterior mean of the rank-$k$ maximum-DOFS projection (\ref{eq:update_max_dof_mean_new}) and $\mathbf{z}_k^a$ the solution obtained after $k$ (non-converged) iterations of a conjugate-gradient minimization starting from  $\mathbf{z}_0^a=0$.
One has:
\begin{align}
\label{eq_gc_max_dof}
\mathbf{L}\mathbf{z}_k^a=\mathbf{x}^a_{\mathbf{\Pi}_\mathrm{dof}}
\end{align}
\begin{proof}
We first rewrite the cost function in (\ref{incr_4dvar}):
\begin{align*}
 \tilde{J}(\mathbf{z}) &=\frac{1}{2}(\mathbf{d}-\mathbf{HLz})^T\mathbf{R}^{-1}(\mathbf{d}-\mathbf{HLz})+\frac{1}{2}\mathbf{z}^T\mathbf{z} \\
 &=  \frac{1}{2} \left [ \mathbf{z}^T \left (  \mathbf{L}^T\mathbf{H}^T\mathbf{R}^{-1}\mathbf{HL}+\mathbf{Id} \right ) \mathbf{z} -\mathbf{z}^T  \mathbf{L}^T\mathbf{H}^T\mathbf{R}^{-1} \mathbf{d}-\mathbf{d}^T\mathbf{R}^{-1}\mathbf{HL} \mathbf{z} +\mathbf{d}^T\mathbf{R}^{-1}\mathbf{d}  \right ] \\
 &=  \frac{1}{2} \left [   \mathbf{z}^T\mathbf{A} \mathbf{z} -\mathbf{k}^T \mathbf{z} -\mathbf{z}^T\mathbf{k} \right ] + \textrm{constant},
\end{align*}
with $\mathbf{A}= \mathbf{L}^T\mathbf{H}^T\mathbf{R}^{-1}\mathbf{HL}+\mathbf{Id} $ and $\mathbf{b}= \mathbf{L}^T\mathbf{H}^T\mathbf{R}^{-1} \mathbf{d}$.\newline
Minimizing $ \tilde{J}(\mathbf{z}) $ is therefore equivalent to solving:
\begin{align*}
\mathbf{A}\mathbf{z} = \mathbf{b}
\end{align*}
Starting with $\mathbf{z}_1=\mathbf{0}$, the CG algorithm at iteration $k$ produces (e.g., \citet{golub2012matrix}):
\begin{align}
\label{cg_krylov}
\mathbf{z}_k^a=\textrm{argmin}_{\mathbf{z} \in \mathcal{K}_k(\mathbf{A},\mathbf{z}_1)}  \| \mathbf{A}\mathbf{z}-\mathbf{b}   \|^2,
\end{align}
 where $\mathcal{K}_k(\mathbf{A},\mathbf{z}_1)$ is the $k$-dimensional Krylov subspace associated with $\mathbf{A}$ and initialized with the vector $\mathbf{z}_1$, i.e., $\mathcal{K}_k(\mathbf{A},\mathbf{z}_1)=\{ \mathbf{z}_1, \mathbf{A} \mathbf{z}_1,...,  \mathbf{A}^k \mathbf{z}_1\} $.  One can rewrite (\ref{cg_krylov}):
 \begin{align}
 \mathbf{z}_k^a=\textrm{argmin}_{\mathbf{z}}  \| \mathbf{A}\mathbf{Q}\mathbf{Q}^T\mathbf{z}-\mathbf{b}   \|^2,
\end{align}
where $\mathbf{Q}$ is a matrix whose columns form an orthonormal basis of $\mathcal{K}_k(\mathbf{A},\mathbf{z}_1)$. Noting that $\mathcal{K}_k(\mathbf{A},\mathbf{p}_1) \approx \mathrm{Im}(\{\mathbf{v}_i,\, i=1,...,k \})$, where $\mathbf{v}_i$ represents the $i$-th eigenvector of $\widehat{\mathbf{H}_p}$, which is a basic property of the CG algorithm and is related to its link with the Lanczos method  \citep{golub2012matrix}. This is equivalent to assuming $\mathbf{Q} \mathbf{Q}^T \approx \mathbf{V}_k\mathbf{V}_k^T$. We therefore obtain the least-squares solution:
\begin{align*}
 \mathbf{z}_k^a&=\left [ \mathbf{A}\mathbf{Q} \mathbf{Q}^T \right ]^{+}\mathbf{b}  \\
&=\left [ \mathbf{A}\mathbf{V}_k \mathbf{V}_k^T \right ]^{+}\mathbf{b}  \\
 &=  \mathbf{L} \left [ \sum_{i=0}^{k} \mathbf{v}_i\mathbf{v}^T_i (1-\lambda_i(1+\lambda_i)^{-1}) \right ] \mathbf{V}\mathbf{\Lambda}^{1/2}\mathbf{W}\mathbf{R}^{-1/2} \mathbf{d} \\
&=  \mathbf{L}\left [ \sum_{i=1}^{k} \lambda_i^{1/2}(1+\lambda_i)^{-1}\mathbf{v}_i\mathbf{w}^T_i \right ]\mathbf{R}^{-1/2}\mathbf{d}   \\
&= \mathbf{x}^a_{\mathbf{\Pi}_\mathrm{dof}},
\end{align*}
where we used $\mathbf{A}=\mathbf{V}(\mathbf{Id}-\mathbf{\Lambda}^{1/2}(\mathbf{Id}+\mathbf{\Lambda})^{-1})\mathbf{V}^T$ and $\mathbf{b}=\mathbf{V}\mathbf{\Lambda}^{1/2}\mathbf{W}\mathbf{R}^{-1/2} \mathbf{d}$. 

\end{proof}
\end{prop}
From Proposition \ref{prop:eq_cg}, we see that the solution of the preconditioned CG minimization after $k$ non-converging iterations is also the posterior mean of an optimal projection, in the sense of Prop. \ref{thm:opt_approx_mean}. Moreover, it is also the posterior mean solution of the projected Bayesian problem with maximum DOF. This latter fact is useful in term of interpretation, since it means that the non-converged solution (for the initial problem) is still the exact solution of a low-rank Bayesian problem, with an information content explicitly defined by Eq. (\ref{eq:update_max_dof_mod_res_new}). Preconditioning the CG algorithm with a square-root of the prior covariance $\mathbf{B}$ was adopted for practical reasons in operational DA systems, that is, to improve the convergence of the minimization and to efficiently represent (implicitly) a high-dimensional $\mathbf{B}$ matrix. Interestingly, our result provides also a theoretical justification for preconditioning the CG algorithm with a square-root of $\mathbf{B}$ in quadratic 4D-Var minimizations. We also note that potentially better approximations of the posterior mean (i.e., $\mathbf{x}^{a}_{\mathrm{FR}_\mathrm{dof}}$, see Fig. \ref{fig:bayes_risk}) are available.

\subsubsection{Singular Square-Root Formulations and Representativeness Error}
In the context of Gaussian pdf assumptions and maximum likelihood estimation, the use of a non-singular $\mathbf{B}$ matrix is required, as evident from the presence of $\mathbf{B}^{-1}$ in the cost function (\ref{eq3}). The positive-definiteness of $\mathbf{B}$ indicates that none of the control space directions are associated with a null probability density in the case of a Gaussian pdf. Square-root formulations of $\mathbf{B}$ currently used in operational DA systems (CVT and/or ensemble-based) can violate this assumption. As emphasized recently by \citet{menetrier2015overlooked}, in some DA systems the balance operators used to construct the CVT lead to a non-square square-root $\mathbf{L}$, and therefore implicitly model a singular $\mathbf{B}$ matrix. Moreover, as previously mentioned, the number of trajectory perturbations, $K$, used to represent $\mathbf{B}$ in ensemble-based approaches is very small compared to the dimension of the control vector, i.e., $K<<n$, which also leads to singular square-roots for $\mathbf{B}$.

We note that preconditioned 4D-Var techniques that use a singular square-root for $\mathbf{L}$ can be readily interpreted within the theoretical framework developed in Section \ref{model_reduc}. Indeed, solving the preconditioned quadratic 4D-Var minimization problem (\ref{incr_4dvar}) with a singular square-root $\mathbf{L}$ amounts to applying a two-step low-rank projection method, as described in Section \ref{form_prel}, with $\mathbf{\Gamma}=(\mathbf{L}^T\mathbf{L})^{-1}\mathbf{L}^T$ and $\mathbf{\Gamma}^\star=\mathbf{L}$. In this case, it is also easy to check that the prolongation operator verifies $\mathbf{\Gamma}^\star=\mathbf{B}\mathbf{\Gamma}^T\left (\mathbf{\Gamma}^T\mathbf{B}\mathbf{\Gamma} \right )^{-1}$, which makes it an optimal prolongation (see Thm. \ref{thm:best_prolong}). Therefore, by Corollary \ref{post_sol_projec}, the exact solution of this singular preconditioned minimization problem (\ref{incr_4dvar}) is the projection of the solution of the initial (full-rank) Bayesian problem, i.e., one has:
\begin{align}
\label{eq:sing_sq_sol}
\mathbf{L}\mathbf{z}^a=\mathbf{\Pi}_\mathbf{L}\mathbf{x}^a,
\end{align}
with $\mathbf{\Pi}_\mathbf{L}=\mathbf{\Gamma}^\star\mathbf{\Gamma}$.
 Note that (\ref{eq:sing_sq_sol}) is verified if the matrix $\mathbf{R}_{\mathbf{\Pi}_\mathbf{L}}$, which accounts for the representativeness error due to the restriction of the problem to the subspace spanned by the column of $\mathbf{L}$, is used in (\ref{incr_4dvar}) instead of the initial observational error covariance matrix $\mathbf{R}$ (i.e., (\ref{eq:sing_sq_sol}) is the solution of the Bayesian problem $\mathcal{B}_{\mathbf{\Pi}_\mathbf{L}}=(E,F,\mathbf{H}\mathbf{\Pi}_\mathbf{L},\mathbf{B},\mathbf{R}_{\mathbf{\Pi}_\mathbf{L}})$. It is worthwhile to note that when singular square-roots formulations are used in operational DA systems, the initial observational error covariance ($\mathbf{R}$) is not modified, which can theoretically lead to errors in the optimization. 
Since the representativeness error \ref{agg_error} cannot be computed in practice, it is useful to understand under which conditions it vanishes, which is the object of the following Proposition:

 \begin{prop}
\label{prop:null_repres_error}
Let us consider  $\mathbf{L}$ a singular square-root of $\mathbf{B}$ $(\mathrm{rank}(\mathbf{L})<n)$, and $\mathbf{\Pi}_\mathbf{L}=\mathbf{L}(\mathbf{L}^T\mathbf{L})^{-1}\mathbf{L}^T$ and $\mathbf{R}_{\mathbf{\Pi}_\mathbf{L}}$ its associated projection and representativeness error, respectively. One has:
\begin{align}
\label{equiv_repres_error}
\mathbf{R}_{\mathbf{\Pi}_\mathbf{L}}=\mathbf{R} \Leftrightarrow \mathrm{Im}(\mathbf{L})^\perp \subset \mathrm{ker}(\mathbf{H})  \Leftrightarrow  \mathrm{Im}(\mathbf{L})^\perp \subset  {\mathrm{Im}(\mathbf{H}^T)}^\perp,
\end{align}
where $^\perp$ denotes the orthogonal complement.
\end{prop}
 
\begin{proof}
The second equivalence on the right is immediate since $\mathrm{ker}(\mathbf{H})= \mathrm{Im}(\mathbf{H}^T)$.
To show the first equivalence, let us write explicitly the representativeness error:
\begin{align}
\label{B_decomp1}
\mathbf{R}_{\mathbf{\Pi}_\mathbf{L}}=\mathbf{R} + \mathbf{H}(\mathbf{B}+\mathbf{\Pi}_\mathbf{L}\mathbf{B} \mathbf{\Pi}_\mathbf{L}- \mathbf{\Pi}_\mathbf{L}\mathbf{B}-\mathbf{B} \mathbf{\Pi}_\mathbf{L} )\mathbf{H}^T
\end{align}
Let us decompose the (full-rank) $\mathbf{B}$ matrix as:
 \begin{align}
 \label{B_decomp2}
 \mathbf{B}&=\mathbf{B}_\mathbf{L}+\mathbf{B}_{\mathbf{L}_\perp} = \mathbf{Z}\mathbf{\Theta}\mathbf{Z}^T+\mathbf{Z}_\perp\mathbf{\Omega}\mathbf{Z}^T_\perp,
 \end{align}
 where $\mathbf{B}_\mathbf{L}=\mathbf{LL}^T=\mathbf{Z}\mathbf{\Theta}\mathbf{Z}^T$ is an eigendecomposition of $\mathbf{B}$ in the subspace spanned by the columns of $\mathbf{L}$, and $\mathbf{B}_\perp=\mathbf{Z}_\perp \mathbf{\Omega} \mathbf{Z}^T_\perp$ is the expression of $\mathbf{B}$ in a basis $\{ \mathbf{z}^\perp_i,\, i=1,...,k\}$ of the orthogonal complement of $\mathrm{Im}(\{ \mathbf{z}_i,\, i=1,...,k\})$ (i.e., $\mathbf{z}_i^T\mathbf{z}^\perp_j=0,\, \forall i,j$). Using (\ref{B_decomp2}) and the orthogonality properties, it is clear that $\mathbf{\Pi}_\mathbf{L}\mathbf{B}=\mathbf{L}\mathbf{L}^T$. Therefore, one has: $\mathbf{B}+\mathbf{\Pi}_\mathbf{L}\mathbf{B} \mathbf{\Pi}_\mathbf{L}- \mathbf{\Pi}_\mathbf{L}\mathbf{B}-\mathbf{B} \mathbf{\Pi}_\mathbf{L}=\mathbf{B}-\mathbf{B} \mathbf{\Pi}_\mathbf{L}=\mathbf{Z}_\perp\mathbf{\Omega}\mathbf{Z}^T_\perp$. Replacing this expression in (\ref{B_decomp1}), one obtains:
\begin{align}
\mathbf{R}_{\mathbf{\Pi}_\mathbf{L}}=\mathbf{R}+\mathbf{H}\mathbf{Z}_\perp\mathbf{\Omega}\mathbf{Z}^T_\perp\mathbf{H}^T
\end{align}
Using a square-root of $\mathbf{\Omega}$\footnote{ Note that $\mathbf{\Omega}= \mathbf{\Omega}^T$, from $\mathbf{Z}_\perp\mathbf{\Omega}\mathbf{Z}^T_\perp =\mathbf{Z}_\perp\mathbf{\Omega}^T\mathbf{Z}^T_\perp$ and multiplication by $(\mathbf{Z}_\perp^T\mathbf{Z}_\perp)^{-1}\mathbf{Z_\perp}^T$ and $\mathbf{Z_\perp}(\mathbf{Z}_\perp^T\mathbf{Z}_\perp)^{-1}$ on the left and right, respectively, which guarantees the existence of a square-root for $\mathbf{\Omega}$.}, one can also write: 
\begin{align}
\mathbf{R}_{\mathbf{\Pi}_\mathbf{L}}=\mathbf{R}+\mathbf{M}\mathbf{M}^T,
\end{align}
where $\mathbf{M}=\mathbf{H}\mathbf{Z_\perp}\mathbf{\Omega}^{1/2}\mathbf{Z_\perp}^T$. Therefore, one obtains the equivalence:
\begin{align}
 \mathbf{R}_{\mathbf{\Pi}_\mathbf{L}}=\mathbf{R} &\Leftrightarrow  \mathbf{M}\mathbf{M}^T = \mathbf{0} \\
 &\Leftrightarrow   \mathbf{M}=\mathbf{0} \\
 &\Leftrightarrow  \mathrm{Im}(\mathbf{Z_\perp}\mathbf{\Omega}^{1/2}\mathbf{Z_\perp}^T) \subset  \mathrm{ker}(\mathbf{H}) \\
 &\Leftrightarrow   \mathrm{Im}(\mathbf{Z_\perp}) \subset  \mathrm{ker}(\mathbf{H}) \\
  &\Leftrightarrow   {\mathrm{Im}(\mathbf{L})}^\perp \subset  \mathrm{ker}(\mathbf{H}) ,
\end{align}
where we used the fact that $ \mathrm{Im}(\mathbf{Z_\perp}\mathbf{\Omega}^{1/2}\mathbf{Z_\perp}^T) = \mathrm{Im}(\mathbf{Z_\perp})$ (since $\mathbf{\Omega}$ is full-rank) and $\mathrm{Im}(\mathbf{Z_\perp})={\mathrm{Im}(\mathbf{L})}^\perp$.
\end{proof}
 \paragraph{Diagnosing the Representativeness Error}
 In practice, the right equivalence in (\ref{equiv_repres_error}) can be useful to diagnose the representativeness error.
 Indeed, one has the following necessary condition:
 \begin{align}
 \label{necess_cond_repres_error}
\mathbf{R}_{\mathbf{\Pi}_\mathbf{L}}=\mathbf{R} \Leftrightarrow  \mathrm{Im}(\mathbf{L})^\perp \subset  {\mathrm{Im}(\mathbf{H}^T)}^\perp  \Rightarrow \mathrm{Im}(\mathbf{H}^T) \subset \mathrm{Im}(\mathbf{L}),
 \end{align}
 where we used the fact that $ {\mathrm{Im}(\mathbf{A})}\subset \mathrm{Im}(\mathbf{B})\Rightarrow{\mathrm{Im}(\mathbf{B})}^\perp\subset {\mathrm{Im}(\mathbf{A})}^\perp$ for any two matrices $\mathbf{A}$ and $\mathbf{B}$. The necessary condition on the right of (\ref{necess_cond_repres_error}) means that the subspace corresponding to non-zero sensitivities of the observations to the control space should be restricted to the range of the singular square-root of $\mathbf{L}$. We note that an approximation of both $\mathrm{Im}(\mathbf{L})$ and $ \mathrm{Im}(\mathbf{H}^T)$ can be obtained using the randomized range finder described in Alg. \ref{alg_one_pass} (step 1 and 2) for $\mathbf{L}$ and $\mathbf{H}^T$, respectively. If one notes $\mathbf{Q}_{\mathbf{H}^T}$ an orthonormal basis for the range of $\mathbf{H}^T$ and $\mathbf{Q}_{\mathbf{L}}$ an orthonormal basis for the range of $\mathbf{L}$, one has: $\mathrm{Im}(\mathbf{H}^T) \subset \mathrm{Im}(\mathbf{L}) \Leftrightarrow \mathbf{Q}_{\mathbf{L}} \mathbf{Q}_{\mathbf{L}}^T \mathbf{Q}_{\mathbf{H}^T}  = \mathbf{Q}_{\mathbf{H}^T} $. Therefore, in the case where the orthonormal basis $\mathbf{Q}_{\mathbf{L}}$ and $\mathbf{Q}_{\mathbf{H}^T}$ are only approximations, a diagnostic for the representativeness error can consist of checking that: 
 \begin{align}
 \label{diag_repres_error}
\mathbf{Q}_{\mathbf{L}} \mathbf{Q}_{\mathbf{L}}^T \mathbf{Q}_{\mathbf{H}^T}  \approx \mathbf{Q}_{\mathbf{H}^T}  \Rightarrow  \| (\mathbf{Id}-\mathbf{Q}_{\mathbf{L}} \mathbf{Q}_{\mathbf{L}}^T)\mathbf{Q}_{\mathbf{H}^T} \|^2 \approx 0
 \end{align}
 In a cycling DA context, $\mathrm{Im}(\mathbf{H}^T)$ would need to be estimated at each DA cycle, while $\mathrm{Im}(\mathbf{L})$ would be estimated only once. However, in Section \ref{improv_4Dvar} we discuss a strategy that would allow estimation of $\mathrm{Im}(\mathbf{H}^T)$ as a by-product of the minimization of (\ref{incr_4dvar}), where the CG algorithm is replaced by a randomized SVD approach. 
 \begin{rem}
A recent study by \citet{menetrier2015overlooked} discusses the impact of singular square-root preconditionings in variational minimization. The discussion focuses on the idea that the change of variable $\mathbf{x}=\mathbf{L}\mathbf{v}$ should rigorously lead to a preconditioned background term of the cost function of the form $J_b(\mathbf{v})=\mathbf{v}^T \mathbf{L}^T\mathbf{B}^{-1}\mathbf{L}\mathbf{v}$, which can be different from the formulation $J_b(\mathbf{v})=\mathbf{v}^T \mathbf{v}$ always used in practice. The authors demonstrate that in the subspace where the CG minimization operates, the two background term formulations are equivalent. However, as described in the present paper (see Prop. \ref{proper:corresp}), and since the singular preconditioning can be interpreted as an effective dimension reduction, the correct approach here is to consider $\mathbf{B}_\omega=\mathbf{L}^T\mathbf{B}\mathbf{L}=\mathbf{Id}_m$ ($m<n$) as the prior error covariance matrix of the preconditioned cost function, since it is precisely the prior error covariance of the reduced Bayesian problem $\mathcal{B}_\omega=(E_\omega,F,\mathbf{H\Gamma}^\star,\mathbf{\Gamma}^T\mathbf{B}\mathbf{\Gamma},\mathbf{R}_{\Pi})$ (with $\mathbf{\Gamma}=(\mathbf{L}^T\mathbf{L})^{-1}\mathbf{L}^T$ and $\mathbf{\Gamma}^\star=\mathbf{\Gamma}$.). Therefore, one sees that the claim that there would be any potential inconsistency in using $J_b(\mathbf{v})=\mathbf{v}^T \mathbf{v}$ in the preconditioned 4D-Var minimization has no theoretical basis. The inconsistency can only arise from neglecting the representativeness error, as explained in this Section.
\end{rem}
\subsection{Improving Incremental 4D-Var}
\label{improv_4Dvar}
In this Section we propose some improvements to the standard 4D-Var algorithm, leveraging the results established in this paper. Description of the modifications and their justifications are provided below for each part of the algorithm. 
\subsubsection{Parallelization of the Inner-Loop Step and Optimal Increments}
\label{parallel_incr4D}
As shown in Prop. \ref{prop:eq_cg}, the increment obtained after $k$ iterations of the CG algorithm in a given inner-loop approximates the rank-$k$ optimal posterior mean update $\mathbf{x}^a_{\mathbf{\Pi}_\mathrm{dof}}$ associated with the linear Bayesian problem defined by the quadratic cost function (\ref{incr_4dvar}).  This stems from the fact that the CG procedure minimizes the cost function in a Krylov subspace that approximately spans the leading $k$ eigenvectors of the prior-preconditioned Hessian, $\widehat{\mathbf{H}}_p$. The eigendecomposition of  $\widehat{\mathbf{H}}_p$ $\{(\mathbf{v}_i,\lambda_i),\, i=1,...,k\}$ is the basis of the two optimal posterior mean updates $\mathbf{x}^a_{\mathbf{\Pi}_\mathrm{dof}}$ (\ref{eq:update_max_dof_mean_new}) and  $\mathbf{x}^{a}_{\mathrm{FR}_\mathrm{dof}}$ (\ref{eq:FR_postupdate}). Therefore, for each inner-loop step, either the Ritz pairs obtained from a CG minimization or the approximation of $\{(\mathbf{v}_i,\lambda_i),\, i=1,...,k\}$ computed from a randomized SVD method can be used in combination with formula (\ref{eq:update_max_dof_mean_new}) or (\ref{eq:FR_postupdate}) to define optimal truncated posterior solutions. Assuming a large number of processors are available, our numerical experiments in Section \ref{num_exp} clearly suggest that replacing an iterative CG minimization algorithm by a parallel implementation of a randomized SVD computation would dramatically enhance the computational efficiency of the 4D-Var algorithm. In particular, in an operational context where only $\sim$10 inner-loop iterations can be afforded for each CG minimizations, a randomized SVD algorithm is expected to provide many more approximated eigenpairs of $\widehat{\mathbf{H}}_p$ than the Ritz pairs. Note that in practice it may be necessary to oversample $\widehat{\mathbf{H}}_p$ in order to obtain accurate SVD estimates (see \ref{average_error}). As discussed in\citet{halko2011finding}, for most problems an oversampling parameter of $\sim$10 is sufficient to obtain satisfying results, so that the number of required parallel integrations of the tangent-linear and adjoint models is roughly equal to number of eigenpairs of $\widehat{\mathbf{H}}_p$ one wants to approximate. For the sake of simplicity, we shall use the term inner-loop to describe any quadratic minimization step of the incremental 4D-Var (i.e., any linearization step of the Gauss-Newton algorithm), even though in the context of randomized SVD methods the minimization is not performed through an iterative algorithm.

As discussed in \ref{interpretation_optimal}, for a given inner-loop and a given approximation of $k$ eigenpairs of $\widehat{\mathbf{H}}_p$ $\{(\widetilde{\mathbf{v}_i},\widetilde{\lambda_i}),\, i=1,...,k\}$, an optimal approach to define the truncated solution of the quadratic minimization problem can be designed using the optimality results established in Proposition \ref{thm:opt_approx_mean}. More precisely, here we present a strategy, described in Alg. \ref{opt_incr}, that can be employed to statistically minimize the error in the approximated increments.

\begin{algorithm}
\caption{Adaptive Posterior Increment}
\label{opt_incr}
For a given inner-loop, let $\{(\widetilde{\mathbf{v}_i},\widetilde{\lambda_i}),\, i=1,...,k\}$ be an approximation of the first $k$ eigenpairs of $\widehat{\mathbf{H}}_p$:
 \begin{algorithmic}[1]
  \STATE if $\widetilde{\lambda_k}\le1$, then the approximated solution to the quadratic minimization problem (\ref{incr_4dvar}) is:  
  \begin{align}
  \label{4dvar_upd_fr}
  \mathbf{x}^{a}_{\mathrm{FR}_\mathrm{dof}}=\mathbf{L}\left [\mathbf{Id}-\left(\sum_{i=1}^{k}\widetilde{\lambda_i}(1+\widetilde{\lambda_i})^{-1}\widetilde{\mathbf{v}_i}\widetilde{\mathbf{v}_i}^T\right) \right ]\mathbf{L}^T\mathbf{H}^T\mathbf{R}^{-1} \mathbf{d},
  \end{align}
   \STATE if $\widetilde{\lambda_k}>1$, then the approximated solution to the quadratic minimization problem (\ref{incr_4dvar}) is: 
    \begin{align}
      \label{4dvar_upd_dof}
 &\mathbf{x}^a_{\mathbf{\Pi}_\mathrm{dof}}=\mathbf{L}\left [ \sum_{i=1}^{k}\left( 1+\widetilde{\lambda_i} \right)^{-1}\widetilde{\mathbf{v}_i}\widetilde{\mathbf{v}_i}^T\right ]\mathbf{L}^T\mathbf{H}^T\mathbf{R}^{-1} \mathbf{d}. 
      \end{align}
 \end{algorithmic}
\end{algorithm}

Based on the optimality results of Prop. \ref{thm:opt_approx_mean}, it is clear that when $\lambda_k\le1$ the choice $\mathbf{x}^{a}_{\mathrm{FR}_\mathrm{dof}}$ for the increment will always yield a smaller average approximation error than $\mathbf{x}^a_{\mathbf{\Pi}_\mathrm{dof}}$. However, the choice $\mathbf{x}^a_{\mathbf{\Pi}_\mathrm{dof}}$ when $\lambda_k>1$ does not  ensure a smaller average approximation error than $\mathbf{x}^{a}_{\mathrm{FR}_\mathrm{dof}}$. The rationale here is that the terms $\lambda_i^3$ in (\ref{eq:opt_dof_B_mean}) will greatly dominate the terms $\lambda_i$ in (\ref{eq:opt_dof_B_mean_FR}) as soon as $\lambda_i$ is significantly greater than 1. This threshold should therefore be adjusted when appropriate, depending on the typical spectra of $\widehat{\mathbf{H}}_p$ for a given inverse problem. Indeed, as shown in our numerical experiments in Section \ref{toy_prob} (Fig. \ref{fig:perc_dof} and Fig. \ref{fig:bayes_risk}), the approximations $\mathbf{x}^{a}_{\mathrm{FR}_\mathrm{dof}}$ and $\mathbf{P}^a_{\mathbf{H}_{\Pi_\mathrm{dof}}}$ can be associated with significantly lower approximation errors than $\mathbf{x}^a_{\mathbf{\Pi}_\mathrm{dof}}$ and $\mathbf{P}^{a}_{\mathbf{\Pi}_\mathrm{dof}}$, respectively, for a number of ranks $k$ such that $k<\max\{j,\,\lambda_j>1\}$. In the remaining we shall use a threshold of 1 for the sake of simplicity, but in principle improved objective criteria to define this threshold based on both heuristic and theory can be derived. This is let for future work.

\begin{rem}
Statistically, Algorithm \ref{opt_incr} should improve the convergence rate of incremental 4D-Var using standard CG minimizations, since by Proposition \ref{prop:eq_cg} the latter always produces the increment $\mathbf{x}^a_{\mathbf{\Pi}_\mathrm{dof}}$, which yields a greater approximation error than $ \mathbf{x}^{a}_{\mathrm{FR}_\mathrm{dof}}$ whenever the smallest approximated eigenvalue of $\widehat{\mathbf{H}}_p$ verifies $\widehat{\lambda_k}\le1$. Note that this method can be readily integrated into existing incremental 4D-Var systems based on CG minimizations by extracting the Ritz pairs $\{(\widetilde{\mathbf{v}_i},\widetilde{\lambda_i}),\, i=1,...,k\}$ and applying Algorithm \ref{opt_incr} instead of using the CG solution. 
\end{rem}

\begin{rem}
In incremental 4D-Var, a stopping criterion needs to be defined for each inner-loop of the algorithm. In many operational DA systems, the stopping criterion corresponds to a fixed number of iterations defined by the available computing resources. Whenever possible, the choice can be improved, by, e.g., considering the absolute norm of the gradient of the cost function (\ref{incr_4dvar}) (i.e., one wants $ \| \nabla J^{k} \|_2 < \epsilon $, $k$ being the iteration number), the relative change of the cost function (i.e., $| J^{k}-J^{k-1}|<\epsilon(1+J^{k})$), or a more sophisticated criterion based on sufficient conditions of convergence for Gauss-Newton minimizations, such as the one described in Lawless and Nichols [2006] (i.e., $\frac{\| \nabla J^{k} \|_2}{\| \nabla J^{0}\|_2}<\epsilon$). In the context where an approximation of $k$ eigenpairs  of $\widehat{\mathbf{H}}_p$ $\{(\widetilde{\mathbf{v}_i},\widetilde{\lambda_i}),\, i=1,...,k\}$ is available and Algorithm \ref{opt_incr} is used instead of the CG solution, the cost function $J^{k}$ and its gradient $\nabla J^{k}$ at the optimal rank-$k$ posterior mean can be evaluated around the values of the posterior means founds from Eqs. (\ref{4dvar_upd_fr}) and (\ref{4dvar_upd_dof}). When using the stopping criteria above in a randomized SVD context for instance, a minimum number of $k$ iterations translates into a minimum number of $k+l$ samples ($l$ being the oversampling parameter) for the random test matrix (see Algorithm \ref{alg_one_pass}). Alternatively, a stopping criterion that combines the deterministic convergence criteria of Lawless and Nichols [2006] and the statistical approach of Algorithm \ref{opt_incr} can be defined as $k=\min\{i,\, \widehat{\lambda_i}<1 \mathrm{\,\, and\,\, } \frac{\| \nabla J^{i} \|_2}{\| \nabla J^{0}\|_2}<\epsilon\}$, since this guarantees both the convergence of the Gauss-Newton algorithm and the statistical optimality of the update $\mathbf{x}^{a}_{\mathrm{FR}_\mathrm{dof}}$.  

\end{rem}

\subsubsection{Optimal Posterior Perturbations in Ensemble-Based DA systems}
\label{opt_post_samp}
In ensemble-based variational DA systems (e.g., EnVar), the propagation of errors from one assimilation window to the next is carried out using an ensemble of perturbed forecast trajectories that sample the posterior probability distribution. In practice, this posterior distribution needs to be evaluated at the maximum-likelihood \citep{Tarantola05}. Therefore, a square-root of the posterior error covariance matrix evaluated at the last (converged) outer-loop iteration can be used to sample the posterior distribution. Two sampling strategies, one deterministic, and one stochastic, can then be adopted:
\paragraph{Stochastic Method}
The stochastic approach to sampling, similar to an EnKF, consists of generating an ensemble of random perturbations that sample the posterior distribution, using the formula:
\begin{align}
\label{stochastic_samp}
\delta \mathbf{w}_i=\mathbf{S} \xi_i,\, i=1,...,k,
\end{align}
where the $\xi_i$ are independent random vectors drawn from a standard normal distribution $\mathcal{N}(0,1)$, and $\mathbf{S}$ is a square-root of the approximated posterior error covariance matrix.
\paragraph{Deterministic Method}
In the deterministic approach, similar to an EAKF, the ensemble of perturbations $\{\delta \mathbf{w}_i,\,i=1,...,k\}$ is constructed so that it exactly verifies:
 \begin{align}
 \label{det_sqrt}
&\mathbf{P}^a_\mathrm{approx}=\delta \mathbf{W}\delta \mathbf{W}^T,
\end{align}
where $\textbf{W}=\frac{1}{\sqrt{k-1}}(\delta \mathbf{w}_1,...,\delta \mathbf{w}_k)$ and  $\mathbf{P}^a_\mathrm{approx}$ is a given approximation of $\mathbf{P}^a$.

Based on Proposition \ref{thm:opt_approx_scov}, optimal rank-$k$ posterior error covariance approximations can be constructed from the eigenpairs $\{(\mathbf{v}_i,\lambda_i)$, $i=1,..,k\}$ of $\widehat{\mathbf{H}_p}$. Given an estimate of the first $k$ eigenpairs $\{(\widetilde{\mathbf{v}_i},\widetilde{\lambda_i})$, $i=1,..,k\}$ obtained, e.g., from iterative CG minimizations or randomized SVD techniques, a adaptive strategy to approximate the square-root of the posterior error covariance matrix is provided in Algorithm \ref{opt_sqrt}.

\begin{algorithm}
\caption{Adaptive Posterior Sampling}
\label{opt_sqrt}
Let $\{(\widetilde{\mathbf{v}_i},\widetilde{\lambda_i}),\, i=1,...,k\}$ be an approximation of the first $k$ eigenpairs of $\widehat{\mathbf{H}}_p$, computed at the last outer-loop iteration of an incremental 4D-Var optimization:
 \begin{algorithmic}[1]
  \STATE if $\widetilde{\lambda_k}\le1$, then the approximated square-root for the posterior error covariance matrix is defined as:  
  \begin{align}
  \label{fr_sqrt}
 \mathbf{S}_{\mathbf{H}_{\Pi_\mathrm{dof}}}=\mathbf{L}\left ( \sum_{i=1}^{k}\left[ (1+\widetilde{\lambda_i})^{-1/2}-1\right ]\widetilde{\mathbf{v}_i}\widetilde{\mathbf{v}_i}^T+\mathbf{Id}\right ),
  \end{align}
   \STATE if $\widetilde{\lambda_k}>1$,  then the approximated square-root for the posterior error covariance matrix is defined as:  
    \begin{align}
      \label{lr_sqrt}
  \mathbf{S}_{\Pi_\text{dof}}=\mathbf{L}  \sum_{i=1}^{k} (1+\widetilde{\lambda_i})^{-1/2} \widetilde{\mathbf{v}_i}.
       \end{align}
 \end{algorithmic}
\end{algorithm}

Both Eq. (\ref{lr_sqrt}) and (\ref{fr_sqrt}) can be used with a stochastic method to generate optimal posterior perturbations. The approximated square-root (\ref{lr_sqrt}) can also be used in a deterministic approach by setting $\delta \mathbf{w}_i= \sqrt{k-1}\,\,\mathbf{L} (1+\lambda_i)^{-1/2} \mathbf{v}_i$ in (\ref{det_sqrt}). Note that the full-rank nature of (\ref{fr_sqrt}) prevents it from being modeled as a deterministic square-root $\delta \mathbf{W}$ using a (necessarily) small number of perturbations $\delta \mathbf{w}_i$.

\begin{rem}
The posterior square-roots (\ref{fr_sqrt}) and (\ref{lr_sqrt}) have also recently been proposed in \citet{auligne2016ensemble} in the context of ensemble-variational DA for NWP. Their method (EVIL) uses the Ritz pairs obtained from the CG minimization at each inner-loop to construct the posterior square-root. The practical advantages of this approach and its better properties compared to other ensemble-based DA methods are discussed in \citet{auligne2016ensemble}. Our optimality results bring further theoretical justifications for using these posterior error square-root formulations, while providing new adaptive methods to construct the approximations for the posterior square-root and the posterior mean update. Additionally, \citet{auligne2016ensemble} note that many Ritz pairs may be necessary to obtain an accurate posterior ensemble, which implies many iterations for the inner-loops and may not be practical for high-dimensional systems. It is further complicated by the fact that a rigorous approach would require one to use only the Ritz pairs obtained from the last inner-loop minimization, i.e., at the optimal state $\mathbf{x}^a$. This can have a significant impact when the system is highly non-linear, since in this case the Hessian matrix can be very different across inner-loops. Provided enough resources are available to perform many tangent-linear and adjoint integrations in parallel, the randomized SVD approach has the potential to approximate many more eigenpairs of the prior-preconditioned Hessian than the CG technique. In practice, it would therefore provide a more accurate approximation of the square-root of the posterior error covariance matrix at the optimal state $\mathbf{x}^a$ than obtained through the framework proposed by \citet{auligne2016ensemble}.
\end{rem}
\begin{rem}
Another possible way to define an approximated square-root of the posterior error covariance is to consider the preconditioner defined in \citet{Tshimanga08}, since their preconditioner effectively approximates the inverse Hessian of the 4D-Var cost function, which itself is an approximation of $\mathbf{P}^a$. In order to sample the posterior distribution, the preconditioner would be computed using only the approximated eigenpairs $\{(\widetilde{\mathbf{v}_i},\widetilde{\lambda_i}),\, i=1,...,k\}$ obtained at the last inner-loop, and the square-root constructed using the recursive formula provided in Thm. 2 of \citet{Tshimanga08}.
\end{rem}

\subsubsection{A New Randomized Incremental Optimal Technique (RIOT) for Variational Data Assimilation }

Combining all the results from Section \ref{improv_4Dvar}, a statistically optimal and computationally efficient ensemble-variational algorithm can be defined as follows:

\begin{algorithm}

\caption{One Cycle of the Randomized Incremental Optimal Technique for Variational Data Assimilation (RIOT VarDA)}
\label{RIOT_4DVAR}
 \begin{algorithmic}[1]
 \REQUIRE An initial time: $t_0$; a final time: $t_f$; an initial prior vector: $\mathbf{x}^b$; an initial innovation vector: $\mathbf{d}=\mathbf{y}-H(\mathbf{x}^b(t_0))$; a square-root of the prior error covariance  matrix: $\mathbf{L}$; a maximum number of outer-loop iterations (resource-dependent): $m$; a maximum number of samples for the randomized SVD approximation (resource-dependent): $r$ ; an oversampling parameter for the randomized SVD approximation: $q$; a targeted rank for the randomized SVD approximation: $k=r-2-q$; a maximum number of samples for the non-linear trajectories (resource-dependent): $p$  \footnotemark
 \FOR {$i =1,m$} 
        \STATE Compute $\mathbf{d}=\mathbf{y}-H(\mathbf{x}^b(t_0))$
         \STATE Using ($r-2$) samples for $\mathbf{\Omega}$ in Algorithm \ref{alg_one_pass}, generate $\mathbf{Q}$ that approximates the range of $\widehat{\mathbf{H}_p} = \mathbf{L}^T\mathbf{H}^T\mathbf{R}^{-1}\mathbf{H}\mathbf{L}$ at $\mathbf{x}^b(t_0)$, and evaluate the probabilistic error bound using two error samples ($l=2$) in Eq. (\ref{bound_prob}) \\
         \STATE Using $\mathbf{Q}$ from step 3, and Algorithm \ref{alg_one_pass}, compute the SVD of $\widehat{\mathbf{H}_p} = \mathbf{L}^T\mathbf{H}^T\mathbf{R}^{-1}\mathbf{H}\mathbf{L}$ \\
        	\STATE Update $\mathbf{x}^b(t_0)$ using Algorithm \ref{opt_incr} to compute the increment $\delta \mathbf{x}^a(t_0)$: \\
	 $\mathbf{x}^b(t_0) \gets \mathbf{x}^b(t_0)+\delta \mathbf{x}^a(t_0)$
 \ENDFOR
 \STATE Propagate $\mathbf{x}^b(t_0)$ to $t_f$ using:\\
  $\mathbf{x}^b(t^f) \gets H(\mathbf{x}^b(t^0))$
 \STATE Use the eigenpair approximations $\{(\widetilde{\mathbf{v}_j},\widetilde{\lambda_j})$, $j=1,..,k\}$ from outer-loop $m$ to define an optimal posterior error square-root $\mathbf{S}$ using Algorithm \ref{opt_sqrt}
 \STATE Construct $p$ stochastic perturbations $\{\delta \mathbf{w}_i(t^0),\, i=1,...,p\}$ at initial time $t_0$ using: \\
 $\delta \mathbf{w}_i(t_0)\gets\mathbf{S} \xi_i$, $\xi_i \sim \mathcal{N}(0,1)$, $i=1,...,p$ 
 \STATE Propagate the perturbations $\delta\mathbf{w}_i(t_0)$ to $t_f$ using: \\
 $\delta \mathbf{w}_i(t^f) \gets H[\mathbf{x}^b(t^0)+\delta  \mathbf{w}_i(t^0)]-H[\mathbf{x}^b(t^0)]$
 \STATE Update the prior square-root $\mathbf{L}$ using, e.g., the hybrid formulation of Eq. (\ref{Schur_product}) with:\\
 $\mathbf{B}_\mathrm{ens}= \mathbf{W}\mathbf{W}^T,\, \textbf{W}=\frac{1}{\sqrt{p-1}}(\delta \mathbf{w}_1(t^f),..., \delta\mathbf{w}_p(t^f))$
 \end{algorithmic}
\end{algorithm}
 \footnotetext{ All the resource-dependent parameters (i.e., $m$, $r$ and $p$) should be defined from previous tests based on the particular application and computational resources available.}

\begin{rem}
An alternative strategy to using Alg. \ref{alg_one_pass} is to use Algorithm 5.1 of \citet{halko2011finding} to approximate the SVD of the square-root of the prior-preconditioned Hessian $\mathbf{U}_{\widehat{\mathbf{H}_p}}=\mathbf{L}^T\mathbf{H}^T\mathbf{R}^{-1/2}\approx \sum_{i=1}^{k}\sqrt{\lambda_i}\mathbf{v}_i\mathbf{w}_i^T$ (or ${\mathbf{U}_{\widehat{\mathbf{H}_p}}}^T=\mathbf{R}^{-1/2}\mathbf{H}\mathbf{L}\approx \sum_{i=1}^{k}\sqrt{\lambda_i}\mathbf{w}_i\mathbf{v}_i^T$). In Alg. 5.1, after application of the range finder (step 1 of Alg. \ref{alg_one_pass}),  the operator $\mathbf{U}_{\widehat{\mathbf{H}_p}}$ (or $\mathbf{U}_{\widehat{\mathbf{H}_p}}^T$) needs to be projected onto the range $\mathbf{Q}$, i.e., one needs to perform the product $\mathbf{Q}^T\mathbf{U}_{\widehat{\mathbf{H}_p}}=({\mathbf{U}_{\widehat{\mathbf{H}_p}}^T\mathbf{Q}})^T$, which requires $k$ tangent linear integrations (or $\mathbf{Q}^T\mathbf{U}_{\widehat{\mathbf{H}_p}}^T=({\mathbf{U}_{\widehat{\mathbf{H}_p}}\mathbf{Q}})^T$, which requires $k$ adjoint integrations). Since the computational cost for both algorithms is largely dominated by the tangent linear and adjoint integrations, we see that Alg. 5.1 and Alg. 5.6 present similar complexities. Interestingly, if Alg. 5.1 is applied to $\mathbf{U}_{\widehat{\mathbf{H}_p}}=\mathbf{L}^T\mathbf{H}^T\mathbf{R}^{-1/2}$, one can used the random samples $\mathbf{H}^T\mathbf{R}^{-1/2}\boldsymbol{\Omega}$ generated as by-product of step 1 to estimate $\mathrm{Im}(\mathbf{H}^T)$, since one has $\mathrm{Im}(\mathbf{H}^T)=\mathrm{Im}(\mathbf{H}^T\mathbf{R}^{-1/2})$. Therefore, Alg. 5.1 would enable a cost-free implementation of the representativeness error diagnostic (\ref{diag_repres_error}). 
\end{rem}

\begin{rem}
Although Algorithm \ref{RIOT_4DVAR} has been presented in the context of a strong-constrained 4D-Var system, it is applicable to EnVar systems by formally removing the explicit time dimension, which is now implicitly included in all vectors and operators. In particular, the increments $\overline{\delta\mathbf{x}}$, the 4D observational operator $\overline{H}$ and the prior error covariance $\overline{\mathbf{B}}$ (represented by an ensemble) are now defined throughout the assimilation window (i.e., between $t_0$ and $t_f$). As a result, the tangent-linear and adjoint models include only spatial interpolation and measurement sensitivity operators (e.g., satellite averaging kernels), which are much easier to develop and faster to integrate than the tangent-linear and adjoint of the NWP model. In the context of EnVar the posterior error covariances are also defined by the updated (posterior) ensemble throughout the assimilation window, so that steps 7-11 of Alg. \ref{RIOT_4DVAR} are no longer needed. Instead, steps 2-5 can be applied to each prior perturbation member to generate the posterior ensemble. 
 The sampling noise in the posterior ensemble can then be mitigated by using an objective 4D-Var localization method, such as the one proposed in \citet{bocquet2016localization}.
\end{rem}
\newpage

\section{Conclusion}
This paper has presented a robust theoretical and numerical study of methods to approximate the solution of large-scale Bayesian problems. Our analysis focused principally on the construction of optimal low-rank projections for potentially high-dimensional Bayesian problems, taking into account the computational constraints of this framework, and in particular the need for matrix-free algorithms. A low-rank projection that maximizes the information content (i.e., the DOFS) of the inversion, the maximum-DOFS solution, was proposed and its optimality properties with respect to the posterior error covariance matrix and posterior mean approximations were analyzed in details. These results are also compared to other useful optimal low-rank approximations that are not associated with projections of the initial Bayesian problem. An interesting aspect of the low-rank approximations derived from the maximum-DOFS solution is that they allow one to adaptively optimize the posterior mean and the posterior error covariances based on the properties of the spectra of the so-called prior-preconditioned Hessian matrix. 

The performance of the optimal projection and alternate low-rank approximations are tested in the context of an atmospheric source inversion problem whose dimension is small enough to allow exact computation of the posterior solution. Our results indicate good convergence properties for both the posterior error covariances and posterior mean approximations and are consistent with the theory. A large-scale version of this experiment is also presented, where the maximum-DOFS solutions are computed using an efficient randomized SVD algorithm, whose parallel implementation dramatically improves the scalability of the SVD computation upon standard iterative matrix-free methods (e.g., the BFGS or Lanczos algorithms). 

Finally, we discussed the link between the maximum-DOFS low-rank approximation and the square-root formulation of incremental 4D-Var methods used in current operational DA systems. In particular, we showed the theoretical equivalence between $k$ iterations of the preconditioned conjugate-gradient algorithm in inner-loop (quadratic) minimizations and the rank-$k$ maximum-DOFS solutions of the associated projected Bayesian problem. Leveraging both our optimality results (e.g., the use of adaptive approximations in quadratic inner-loop minimizations) and the computational efficiency of the randomized SVD algorithms, we then proposed an improved implementation of incremental 4D-Var (RIOT). This approach is very generic and can be used with any square-root formulation of incremental 4D-Var, including hybrid ensemble-4D-Var \citep{clayton2013operational}. 

Randomized SVD methods can be exploited to massively parallelize adjoint-based 4D-Var minimization algorithms, and as such  represent an alternative approach that could rival the computational efficiency of ensemble-based DA, while preserving the full-rank nature of the variational formulation. However, adjoint-free ensemble-based approach such as EnKF or EnVar \citep{buehner2010intercomparison} still present the advantage that they do not require the development and maintenance of an adjoint model. Currently, randomization methods fundamentally rely on the availability of an adjoint model. One challenge ahead is to design matrix-free randomization methods that could perform SVD based on the forward model only. Additionally, preconditioning techniques such as the ones proposed in \citet{Tshimanga08} could be applied to the randomized SVD method to improve the accuracy of the inner-loop minimizations. Therefore, future studies should focus on leveraging both parallelization and preconditioning methods to maximize the efficiency of randomized SVD techniques in incremental 4D-Var. 

\paragraph{Acknowledgments} This work was supported by the NASA GEO-CAPE Science Team grant
  NNX14AH02G and the NOAA grant NA14OAR4310136. This work utilized the Janus supercomputer, which is supported by the National Science Foundation (award number CNS-0821794), the University of Colorado Boulder, the University of Colorado Denver, and the National Center for Atmospheric Research. The Janus supercomputer is operated by the University of Colorado Boulder.

\appendix
\section*{Appendix: Useful Lemmas and Their Proofs}
 
\begin{lem}
If $\mathbf{A}$, $\mathbf{B}$ and $\mathbf{C} > 0$ are two $n \times n$ Hermitian non-negative definite matrices, and $\mathbf{S}$ is a $n \times m$ matrix, one has the following properties:
\begin{align}
\label{eq:ineq_inv}
&\mathbf{A} > \mathbf{B}  \Leftrightarrow  \mathbf{B}^{-1} > \mathbf{A}^{-1} \\
\label{eq:ineq_proj}
&\mathbf{A} \ge \mathbf{B} \Rightarrow \mathbf{S}^T\mathbf{A} \mathbf{S}\ge \mathbf{S}^T\mathbf{B} \mathbf{S}
\end{align}
where $\ge$ (respectively $>$) denotes the L\"owner partial ordering (respectively strict) within the set of Hermitian non-negative definite matrices.
\end{lem}

\begin{proof}
A proof of (\ref{eq:ineq_inv}) is given in \citet{horn2012matrix} (Corollary 7.7.4). We recall it here for the reader's convenience. 
Let $\rho(.)$ be the spectral radius of a matrix. One has $\mathbf{A}>\mathbf{B}$ if and only if $\rho(\mathbf{A}^{-1}\mathbf{B})<1$. Since for all $\mathbf{M}$, $\mathbf{N}$ Hermitian matrices $\rho(\mathbf{MN})=\rho(\mathbf{NM})$, one has $\rho(\mathbf{A}^{-1}\mathbf{B})=\rho(\mathbf{B}\mathbf{A}^{-1})$. (\ref{eq:ineq_inv}) follows immediately. \newline
By linearity, demonstrating (\ref{eq:ineq_proj}) is equivalent to demonstrating $\mathbf{A} \ge 0 \Rightarrow \mathbf{S}^T\mathbf{A} \mathbf{S}\ge 0$. It is easily proven by considering $\mathbf{x}\ne 0$ and $\mathbf{y}=\mathbf{Sx}$. Since $\mathbf{A}\ge0$, in particular $\mathbf{y}^T\mathbf{A}\mathbf{y}= \mathbf{x}^T\mathbf{S}^T\mathbf{A}\mathbf{S} \mathbf{x}\ge0$, which shows that $\mathbf{S}^T\mathbf{A} \mathbf{S}\ge 0$.
\end{proof}

The demonstration of Theorem (\ref{thm:opt_proj}) requires the following lemma (see \citet{horn2012matrix}, Corollary 4.3.39):
\begin{lem}
\label{max_tr}
Let $\mathbf{A}$ be a $(n \times n)$ Hermitian matrix, suppose that $1\le p\le n$, and let $\mathbf{U}$ be a $(n \times p)$ matrix whose $p$ columns form an orthonormal basis . Let $\lambda_1\le...\le\lambda_n$ be the eigenvalues of $\mathbf{A}$. Then one has:
\begin{eqnarray}
\label{eq:ineq_eigenv}
\sum_{i=1}^{p} \lambda_i \le  \mathrm{Tr}(  \mathbf{U}^T\mathbf{AU} )\\ \notag
\sum_{i=n-p+1}^{n} \lambda_i \ge  \mathrm{Tr}( \mathbf{U}^T\mathbf{AU})
\end{eqnarray}
Moreover, for each $p=0,...,n$, one has:
 \begin{eqnarray}
 \label{eq:eq_eigenv}
 \mathrm{Tr}( \mathbf{V}_\mathrm{min}^T\mathbf{A} \mathbf{V}_\mathrm{min}  )&=& \sum_{i=1}^{p} \lambda_i  \\ \notag
 \mathrm{Tr}(  \mathbf{V}_\mathrm{max}^T\mathbf{A}\mathbf{V}_\mathrm{max} ) &=& \sum_{i=n-p+1}^{n} \lambda_i ,
\end{eqnarray}
where $\mathbf{V}_\mathrm{min}$  (resp., $\mathbf{V}_\mathrm{max}$) represents the matrix whose columns are the singular vectors associated with the $p$ smallest (resp., highest) singular values of $\mathbf{A}$. 
\end{lem}
\begin{proof}
This is a corollary of the so-called \textit{Poincar\'e separation theorem}, which states that:
\begin{align}
\lambda_i(\mathbf{A})\le\lambda_i(\mathbf{U}^T\mathbf{A}\mathbf{U})\le\lambda_{i+n-p}(\mathbf{A}),\, \forall i=0,...,p,
\end{align} 
where $\lambda_i(\mathbf{M})$ denotes the $i$th singular value of the matrix $\mathbf{M}$, the singular values being arranged in increasing order. 
To prove (\ref{eq:ineq_eigenv}), one just needs to remark that the sum of the eigenvalues and the sum of the main diagonal elements of an Hermitian matrix are equal. (\ref{eq:eq_eigenv}) is trivial.
\end{proof}

\begin{lem}
\label{equiv_average_F}
Let us define a matrix $\mathbf{M}$, a random vector $\mathbf{q}$ with covariance matrix $ \mathbb{E}(\mathbf{qq}^T)=\mathbf{Q}$, and a Hermitian positive-definite matrix $\mathbf{A}$. Let us also define the square-roots $\mathbf{L}_\mathbf{Q}$ and  $\mathbf{L}_\mathbf{A}$  of $\mathbf{Q}$ and $\mathbf{A}$, respectively, i.e., $\mathbf{Q}=\mathbf{L}_\mathbf{Q}{\mathbf{L}_\mathbf{Q}}^T$ and $\mathbf{A}=\mathbf{L}_\mathbf{A}{\mathbf{L}_\mathbf{A}}^T$. One has:
\begin{align}
 \mathbb{E}  \| \mathbf{M}\mathbf{q} \|^2_A =\| \mathbf{L}_\mathbf{A}^T  \mathbf{M}\mathbf{L}_\mathbf{Q}  \|^2_F
\end{align}
\end{lem}
\begin{proof}
\begin{align*}
  \mathbb{E}  \| \mathbf{M}\mathbf{q} \|^2_A &=  \mathbb{E}  \left [ (\mathbf{M}\mathbf{q})^T\mathbf{A}\mathbf{M}\mathbf{q} \right ] \\
&=  \mathbb{E}  \left [\mathbf{q}^T\mathbf{M}^T\mathbf{L}_\mathbf{A}{\mathbf{L}_\mathbf{A}}^T\mathbf{M}\mathbf{q}\right ]   \\
&= \mathbb{E}  \left [ \mathrm{Tr} \left ( \mathbf{q}^T\mathbf{M}^T\mathbf{L}_\mathbf{A}{\mathbf{L}_\mathbf{A}}^T\mathbf{M}\mathbf{q} \right )\right ] \\
&=  \mathbb{E}  \left [ \mathrm{Tr} \left ({\mathbf{L}_\mathbf{A}}^T\mathbf{M}\mathbf{q}  \mathbf{q}^T\mathbf{M}^T\mathbf{L}_\mathbf{A}\right )\right ] \\
&= \mathrm{Tr} \left [ {\mathbf{L}_\mathbf{A}}^T\mathbf{M}\mathbb{E} (\mathbf{q}  \mathbf{q}^T)\mathbf{M}^T\mathbf{L}_\mathbf{A}\right ]  \\
&= \mathrm{Tr} \left [ {\mathbf{L}_\mathbf{A}}^T\mathbf{M}\mathbf{L}_\mathbf{Q}{\mathbf{L}_\mathbf{Q}}^T\mathbf{M}^T\mathbf{L}_\mathbf{A}\right ]   \\
&= \mathrm{Tr} \left [ {\mathbf{L}_\mathbf{A}}^T\mathbf{M}\mathbf{L}_\mathbf{Q}( {\mathbf{L}_\mathbf{A}}^T\mathbf{M}\mathbf{L}_\mathbf{Q})^T \right ] \\
&= \| \mathbf{L}_\mathbf{A}^T  \mathbf{M}\mathbf{L}_\mathbf{Q}  \|^2_F
\end{align*}
\end{proof}


\begin{figure}[]
\centering
\includegraphics[width=25pc]{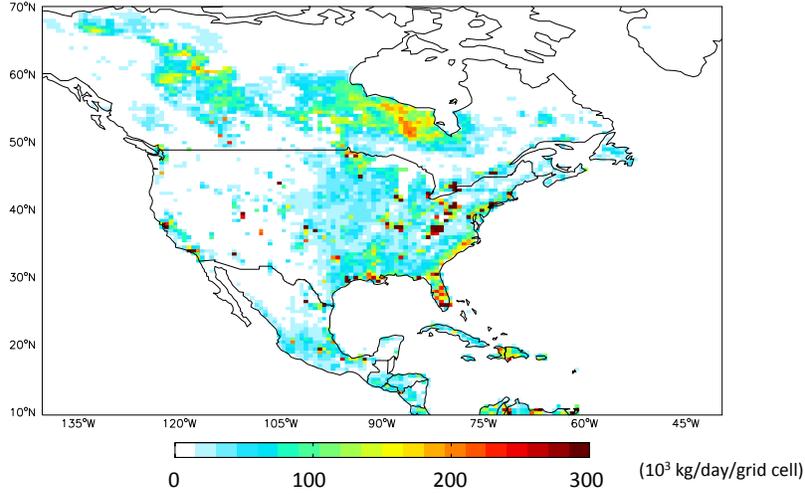}
\caption{Monthly averaged total daily prior methane emissions for the nested North America domain ($0.5^{\circ} \times  0.7^{\circ}$). The size of the control vector for the original inverse problem is $n=151\times121=18,271$. The reduced inverse problem considers only a subset of the grid-cells, resulting in a control vector of size $n=300$.}
\label{fig:map_prior}
\end{figure}

\begin{figure}[]
\centering
\includegraphics[width=25pc]{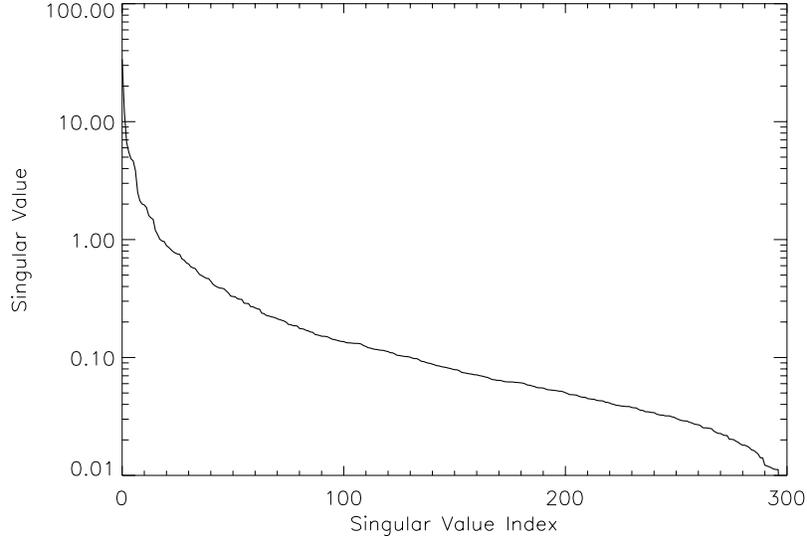}
\caption{ Singular value spectra of the prior-preconditioned Hessian matrix of the reduced source inversion problem. }
\label{fig:toy_spectra}
\end{figure}

\begin{figure}[]
\centering
\includegraphics[width=35pc]{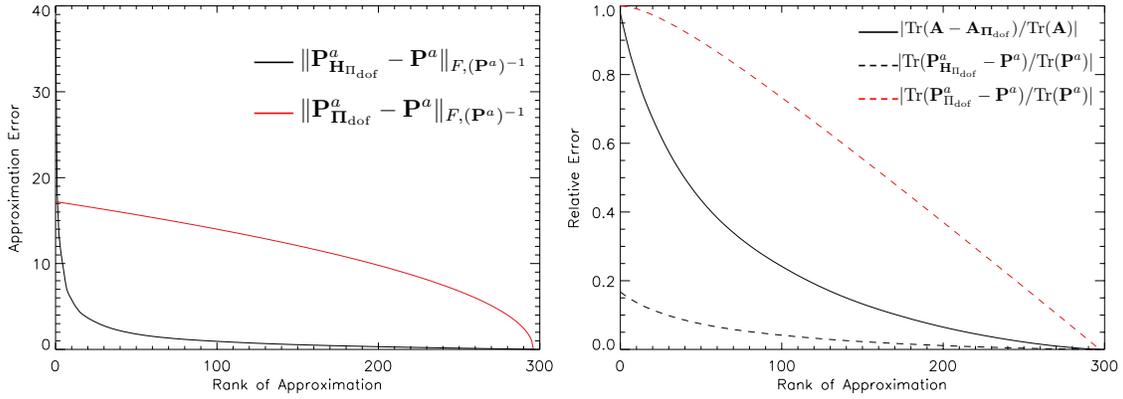}
\caption{Left: $\mathbf{P}^a$-weighted error in the posterior error covariance matrix approximations $\mathbf{P}^a_{\mathbf{H}_{\Pi_\text{dof}}}$ (black line) and $\mathbf{P}^a_{\Pi_\text{dof}}$ (red line), as a function of the rank $k$ of the approximation; Right: relative error in the DOFS approximation for solution of the maximum-DOFS projection ($\mathbf{A}_{\Pi_\text{dof}}$) (black solid line), and relative error in the total variance approximations $\mathrm{Tr}(\mathbf{P}^a_{\mathbf{H}_{\Pi_\text{dof}}})$ (black dashed line) and $\mathrm{Tr}(\mathbf{P}^a_{\Pi_\text{dof}})$ (red dashed line), as a function of the rank $k$ of the approximation. }
\label{fig:perc_dof}
\end{figure}

\begin{figure}[]
\centering
\includegraphics[width=35pc]{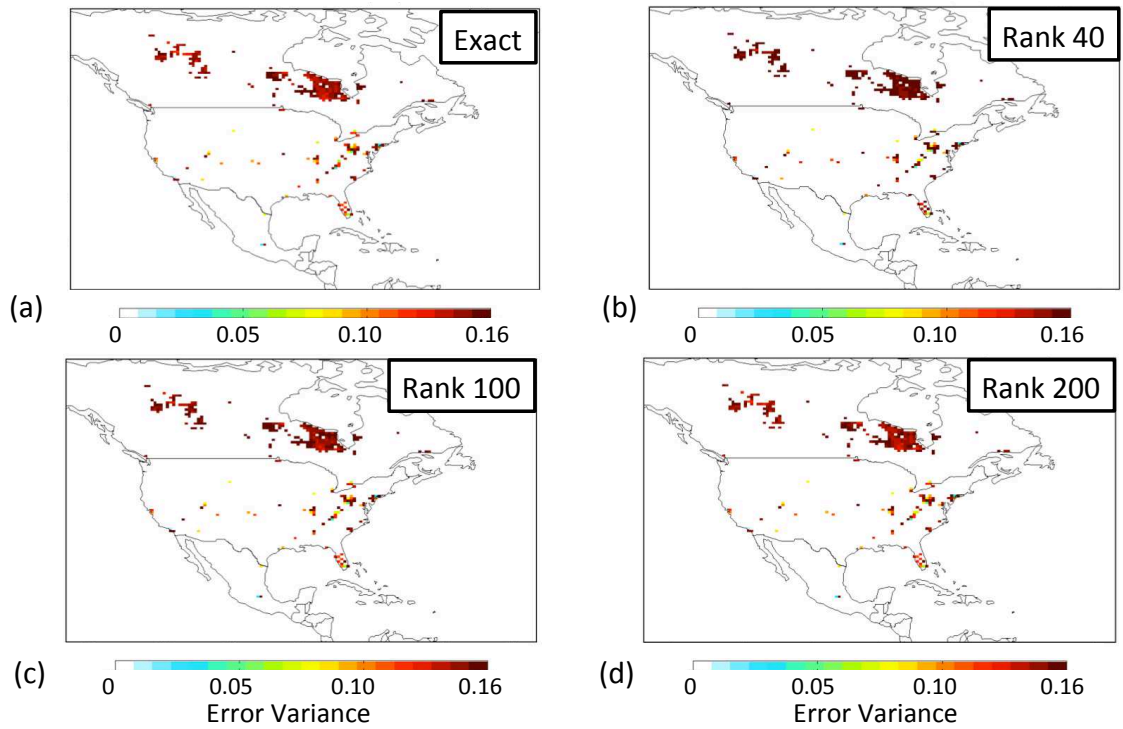}
\caption{Exact and approximated posterior error variances using the low-rank update estimation $\mathbf{P}^a_{\mathbf{H}_{\Pi_\text{dof}}}$: (a) exact variances; (b) variances for a rank-40 update; (c) variances for a rank-100 update; (d) variances for a rank-200 update. }
\label{fig:map_var}
\end{figure}

\begin{figure}[]
\centering
\includegraphics[width=30pc]{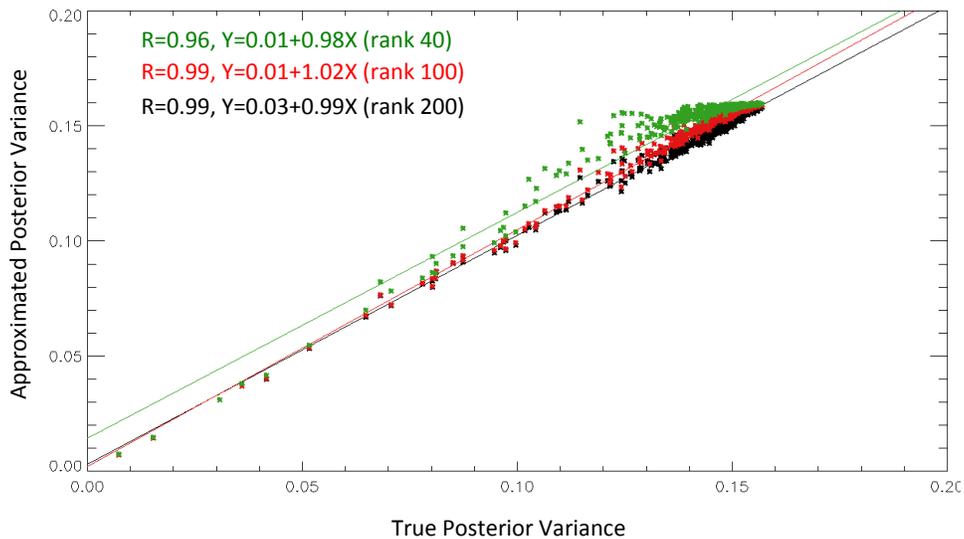}
\caption{Scatterplot of true and approximated posterior error variances using the low-rank update estimation $\mathbf{P}^a_{\mathbf{H}_{\Pi_\text{dof}}}$, for different value of the rank $k$. Least-squares fit lines are shown along with the corresponding Pearson correlation coefficients (R) and linear fit equations (Y=b+aX). Green: variances for a rank-40 update; red: variances for a rank-100 update; black: variances for a rank-200 update. }
\label{fig:scatterplot_var}
\end{figure}

\begin{figure}[]
\centering
\includegraphics[width=25pc]{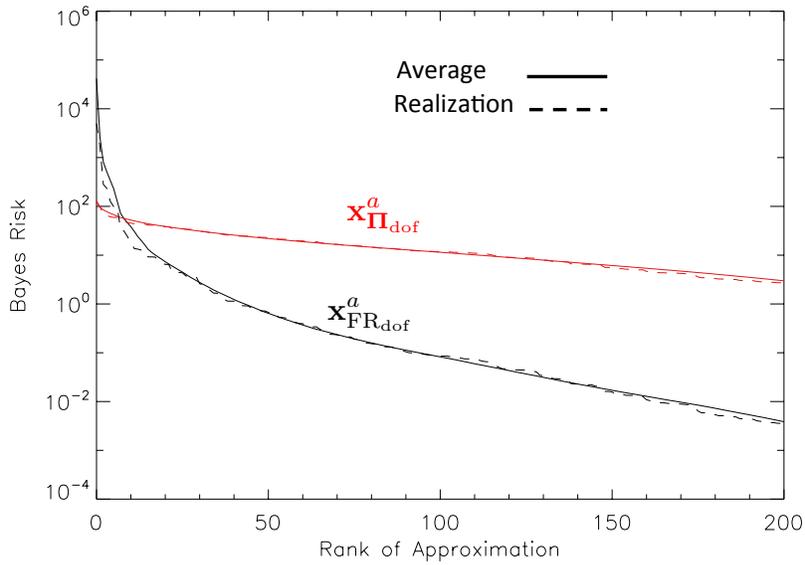}
\caption{Average $\mathbf{P}^a$-weighted error in the posterior mean approximation (or Bayes risk), for the solution of the maximum-DOFS projection $\mathbf{x}^a_{\mathbf{\Pi}_\mathrm{dof}}$ (red solid line) and for the full-rank posterior mean approximation $\mathbf{x}^{a}_{\mathrm{FR}_\mathrm{dof}}$ (black solid line), as a function of the rank of the approximation. Results for one single realization of the prior and the observations are also shown in dashed lines for both approximations.}
\label{fig:bayes_risk}
\end{figure} 

\begin{figure}[]
\centering
\includegraphics[width=35pc]{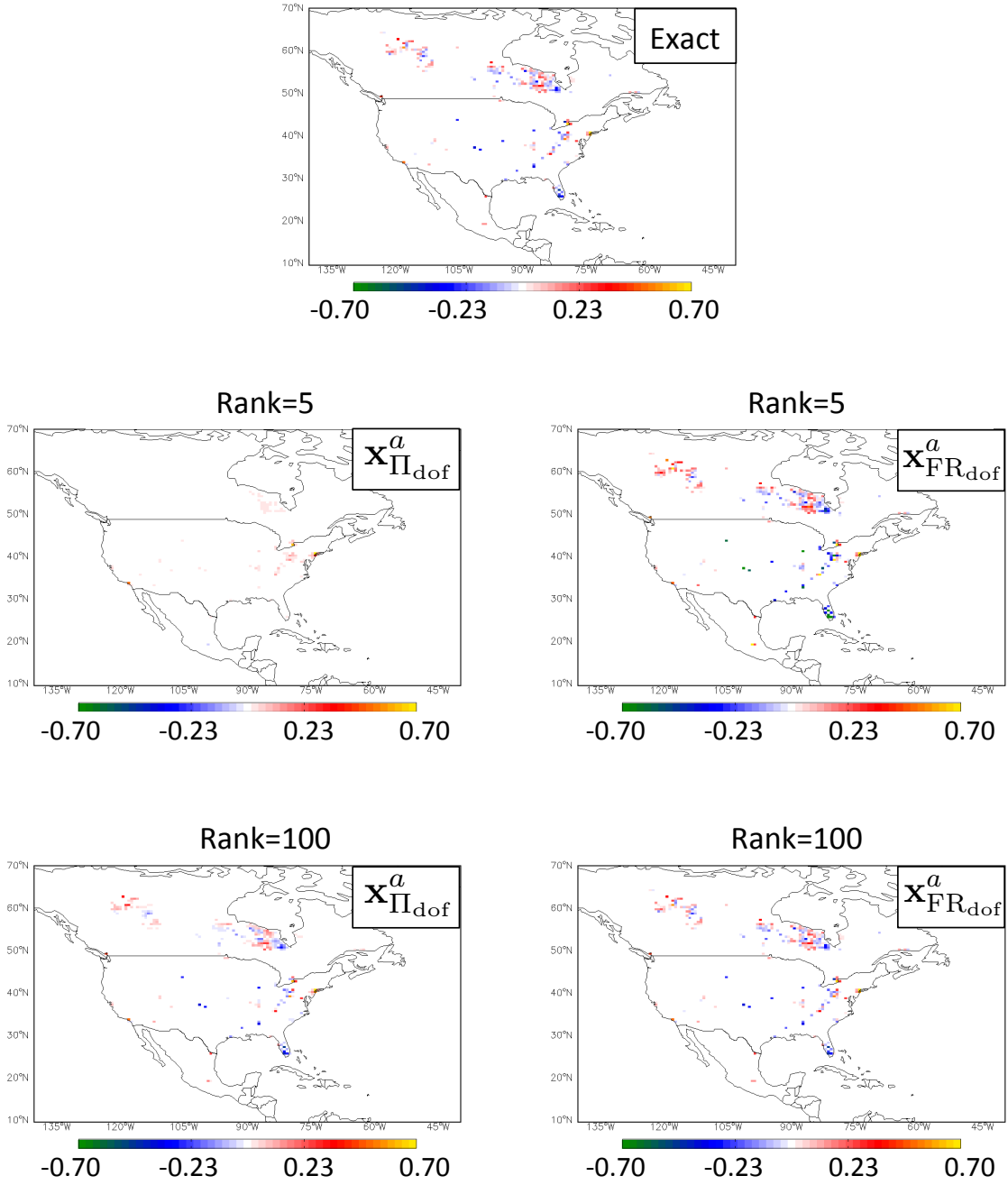}
\caption{Exact (top) and approximated posterior flux increments for the solution of the maximum-DOFS projection $\mathbf{x}^a_{\mathbf{\Pi}_\mathrm{dof}}$ and the full-rank posterior increment approximation $\mathbf{x}^{a}_{\mathrm{FR}_\mathrm{dof}}$, for $k=5$ and $k=100$.}
\label{fig:lr_incr_toy}
\end{figure} 

\begin{figure}[]
\centering
\includegraphics[width=35pc]{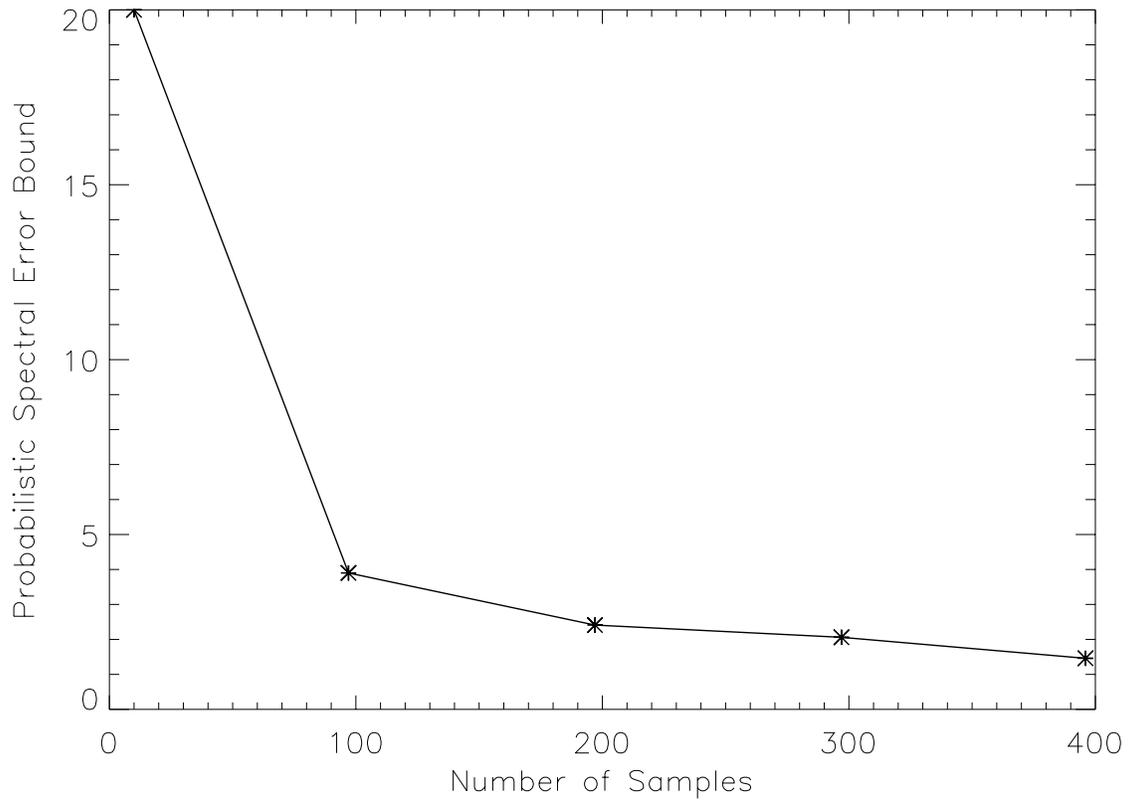}
\caption{Probabilistic spectral error bound as a function of the number of samples used in the randomized SVD estimate.}
\label{fig:prob_err}
\end{figure} 

\begin{figure}[]
\centering
\includegraphics[width=35pc]{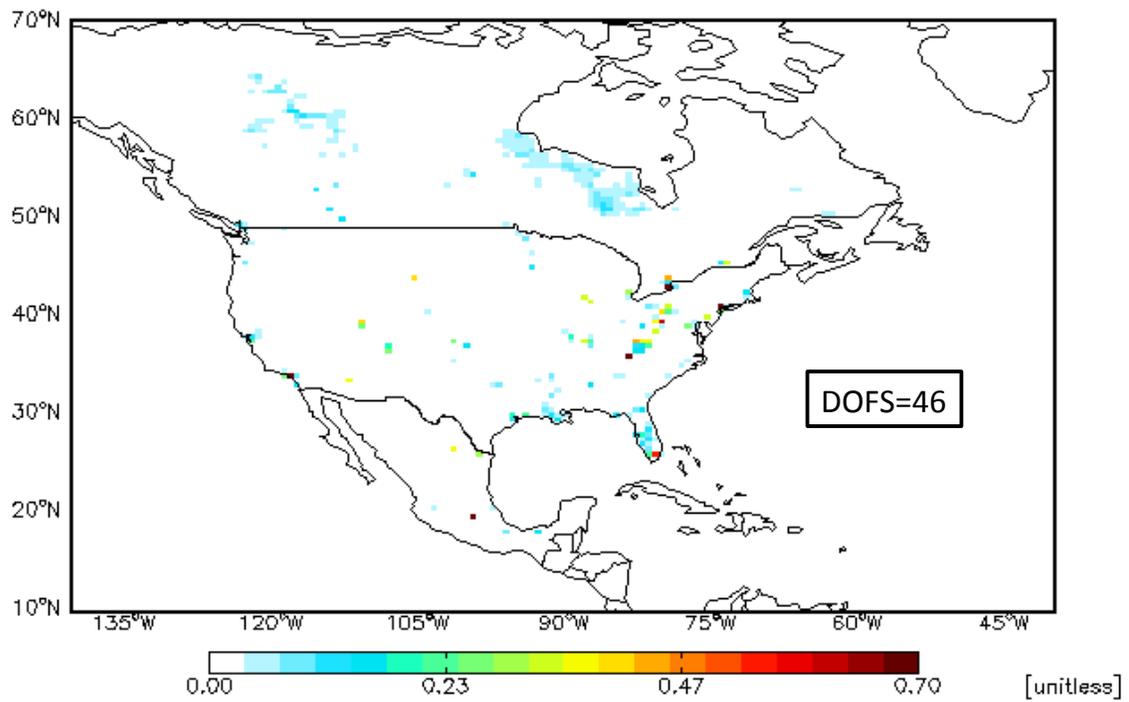}
\caption{Diagonal of the model resolution matrix of the rank-$400$ maximum-DOFs projection ($\mathbf{A}_{\Pi_\text{dof}}$) for the full-dimensional inverse problem. Each element of the diagonal is associated with the observational constraints on the flux in one single grid-cell. The value quantifies the relative contribution (from 0 to 1) of the observations to the total information content, with respect to the prior information. The trace of the model resolution matrix, or DOFS, is also indicated.
}
\label{fig:diag_avk}
\end{figure} 

\begin{figure}[]
\centering
\includegraphics[width=35pc]{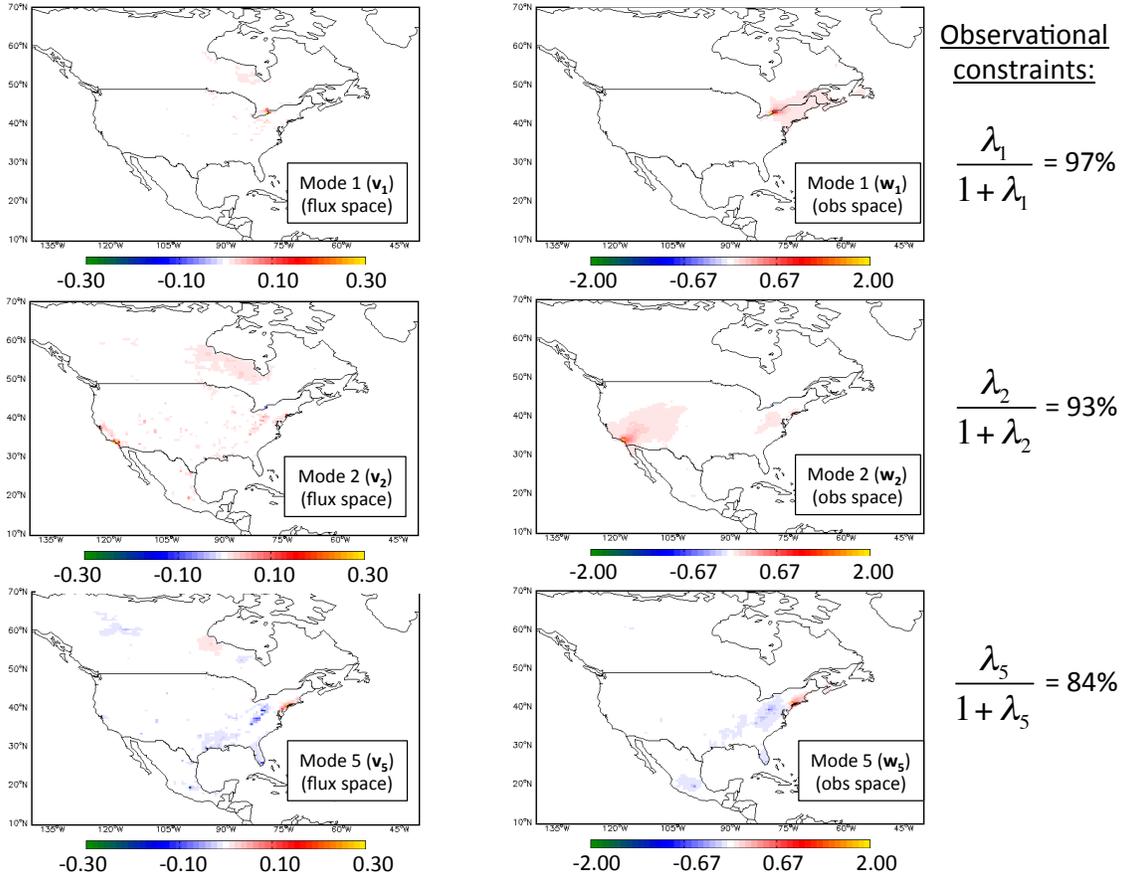}
\caption{Three of the principal modes of the prior-preconditioned Hessian of the inversion ($\widehat{\mathbf{H}}_p$), in control and observation space. Left panel: 1st, 2nd and 5th right singular vectors of the square-root of the prior-preconditioned Hessian $\widehat{\mathbf{H}}^{1/2}_p\equiv \mathbf{R}^{-1/2}\mathbf{H}^T\mathbf{B}^{1/2}$. Right panel: 1st, 2nd and 5th left singular vectors of $\widehat{\mathbf{H}}^{1/2}_p$. The relative contribution of the observations to the posterior information content (with respect to the prior) is also indicated on the right of the figures. }
\label{fig:princ_modes}
\end{figure} 

\begin{figure}[]
\centering
\includegraphics[width=35pc]{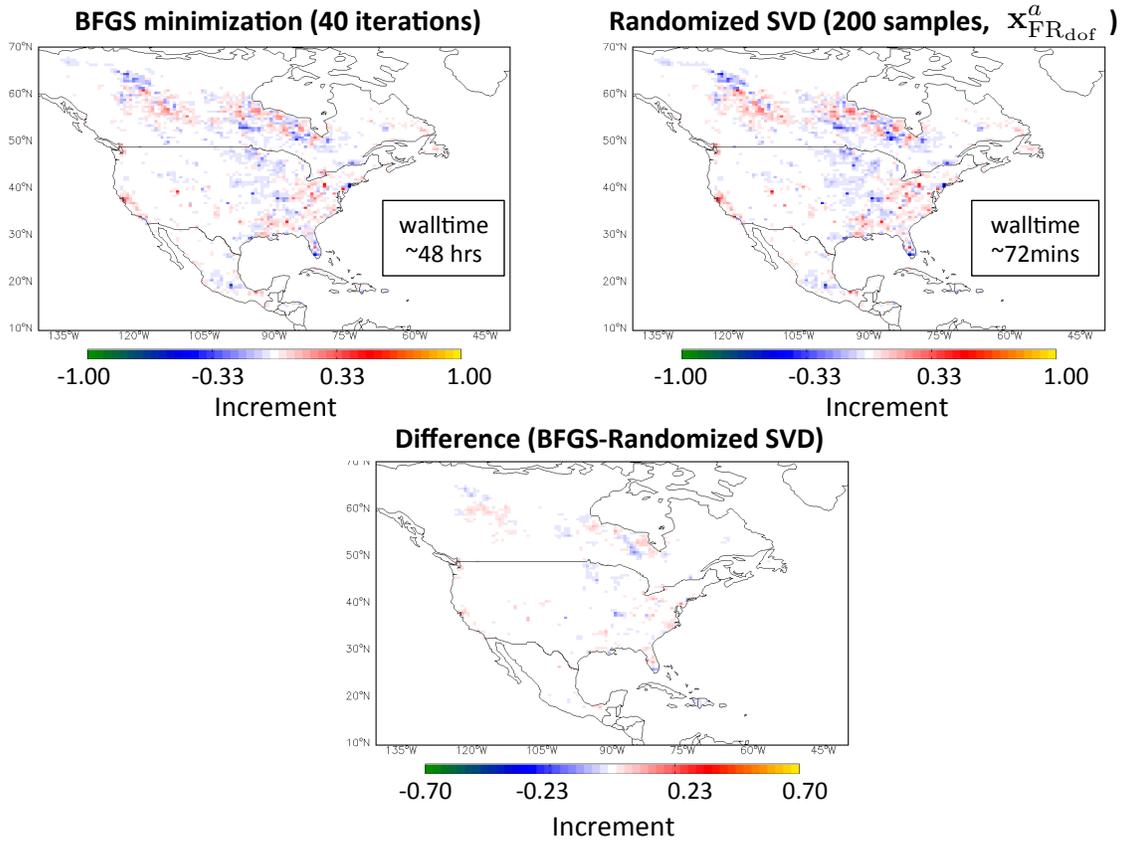}
\caption{Comparison of the computational efficiency of a standard BFGS minimization with the adaptive approximation approach using a randomized SVD. Top left: posterior scaling factor flux increments after 40 iterations of the BFGS algorithm; Top right: posterior scaling factor flux increments for the adaptive approximation using a randomized SVD with 200 samples. In this case $\lambda_{200}<1$, so the posterior increment is the full-rank approximation $\mathbf{x}^{a}_{\mathrm{FR}_\mathrm{dof}}$; Bottom: Difference between the posterior scaling factor increments obtained from the BFGS minimization and from the adaptive approximation using randomized SVD. The walltimes associated with the BFGS minimization and the adaptive approximation using randomized SVDs are also indicated on the top figures.}
\label{fig:comp_eff}
\end{figure} 
\bibliography{mybib}{}

\begin{thebibliography}{43}
\providecommand{\natexlab}[1]{#1}
\providecommand{\url}[1]{\texttt{#1}}
\providecommand{\urlprefix}{URL }
\expandafter\ifx\csname urlstyle\endcsname\relax
  \providecommand{\doi}[1]{doi:\discretionary{}{}{}#1}\else
  \providecommand{\doi}{doi:\discretionary{}{}{}\begingroup
  \urlstyle{rm}\Url}\fi

\bibitem[{Anderson(2001)}]{Anderson01}
Anderson J. 2001. An ensemble adjustment {K}alman filter for data assimilation.
  \emph{Monthly Weather Review} \textbf{129}(12): 2884--2903.

\bibitem[{Anderson and Lei(2013)}]{andersonlei13}
Anderson J, Lei L. 2013. Empirical localization of observation impact in
  ensemble {K}alman filters. \emph{Monthly Weather Review} \textbf{141}(11):
  4140--4153.

\bibitem[{Aulign{\'e} \emph{et~al.}(2016)Aulign{\'e}, M{\'e}n{\'e}trier, Lorenc
  and Buehner}]{auligne2016ensemble}
Aulign{\'e} T, M{\'e}n{\'e}trier B, Lorenc AC, Buehner M. 2016.
  Ensemble-variational integrated localized data assimilation. \emph{Monthly
  Weather Review} \textbf{144}(10): 3677--3696.

\bibitem[{Bannister(2008)}]{Bannister08}
Bannister RN. 2008. A review of forecast error covariance statistics in
  atmospheric variational data assimilation. \textrm{I} : Characteristics and
  measurements of forecast error covariances. \emph{Quarterly Journal of the
  Royal Meteorological Society} \textbf{134}(637): 1951--1970.

\bibitem[{Bennett(2005)}]{bennett2005inverse}
Bennett AF. 2005. \emph{Inverse modeling of the ocean and atmosphere}.
  Cambridge University Press.

\bibitem[{Bocquet(2016)}]{bocquet2016localization}
Bocquet M. 2016. Localization and the iterative ensemble {K}alman smoother.
  \emph{Quarterly Journal of the Royal Meteorological Society}
  \textbf{142}(695): 1075--1089.

\bibitem[{Bocquet \emph{et~al.}(2011)Bocquet, Wu and
  Chevallier}]{bocquet2011bayesian}
Bocquet M, Wu L, Chevallier F. 2011. Bayesian design of control space for
  optimal assimilation of observations. {P}art {I}: Consistent multiscale
  formalism. \emph{Quarterly Journal of the Royal Meteorological Society}
  \textbf{137}(658): 1340--1356.

\bibitem[{Bousserez \emph{et~al.}(2015)Bousserez, Henze, Perkins, Bowman, Lee,
  Liu, Deng and Jones}]{bousserez2015improved}
Bousserez N, Henze D, Perkins A, Bowman K, Lee M, Liu J, Deng F, Jones D. 2015.
  Improved analysis-error covariance matrix for high-dimensional variational
  inversions: application to source estimation using a 3{D} atmospheric
  transport model. \emph{Quarterly Journal of the Royal Meteorological Society}
  \textbf{141}(690): 1906--1921.

\bibitem[{Bousserez \emph{et~al.}(2016)Bousserez, Henze, Rooney, Perkins,
  Wecht, Turner, Natraj and Worden}]{bousserez2016constraints}
Bousserez N, Henze DK, Rooney B, Perkins A, Wecht KJ, Turner AJ, Natraj V,
  Worden JR. 2016. Constraints on methane emissions in north america from
  future geostationary remote-sensing measurements. \emph{Atmospheric Chemistry
  and Physics} \textbf{16}(10): 6175--6190.

\bibitem[{Buehner \emph{et~al.}(2010)Buehner, Houtekamer, Charette, Mitchell
  and He}]{buehner2010intercomparison}
Buehner M, Houtekamer P, Charette C, Mitchell HL, He B. 2010. Intercomparison
  of variational data assimilation and the ensemble {K}alman filter for global
  deterministic {NWP}. {P}art {I}: Description and single-observation
  experiments. \emph{Monthly Weather Review} \textbf{138}(5): 1550--1566.

\bibitem[{Bui-Thanh \emph{et~al.}(2012)Bui-Thanh, Burstedde, Ghattas, Martin,
  Stadler and Wilcox}]{bui2012extreme}
Bui-Thanh T, Burstedde C, Ghattas O, Martin J, Stadler G, Wilcox LC. 2012.
  Extreme-scale {UQ} for {B}ayesian inverse problems governed by {PDE}s. In:
  \emph{Proceedings of the International Conference on High Performance
  Computing, Networking, Storage and Analysis}. IEEE Computer Society Press,
  p.~3.

\bibitem[{Clayton \emph{et~al.}(2013)Clayton, Lorenc and
  Barker}]{clayton2013operational}
Clayton A, Lorenc AC, Barker DM. 2013. Operational implementation of a hybrid
  ensemble/4{D}-var global data assimilation system at the {M}et {O}ffice.
  \emph{Quarterly Journal of the Royal Meteorological Society}
  \textbf{139}(675): 1445--1461.

\bibitem[{Courtier \emph{et~al.}(1994)Courtier, Thepaut and
  Hollingsworth}]{courtier94}
Courtier P, Thepaut J, Hollingsworth A. 1994. A strategy for operational
  implementation of 4{D}-var, using an incremental approach. \emph{Quarterly
  Journal of the Royal Meteorological Society} \textbf{120}(519): 1367--1387.

\bibitem[{Cui \emph{et~al.}(2014)Cui, Martin, Marzouk, Solonen and
  Spantini}]{cui2014likelihood}
Cui T, Martin J, Marzouk YM, Solonen A, Spantini A. 2014. Likelihood-informed
  dimension reduction for nonlinear inverse problems. \emph{Inverse Problems}
  \textbf{30}(11): 114\,015.

\bibitem[{Desroziers \emph{et~al.}(2014)Desroziers, Camino and
  Berre}]{desroziers20144denvar}
Desroziers G, Camino JT, Berre L. 2014. 4{DE}n{V}ar: link with 4{D} state
  formulation of variational assimilation and different possible
  implementations. \emph{Quarterly Journal of the Royal Meteorological Society}
  \textbf{140}(684): 2097--2110.

\bibitem[{Friedland and Torokhti(2007)}]{friedland2007generalized}
Friedland S, Torokhti A. 2007. Generalized rank-constrained matrix
  approximations. \emph{SIAM Journal on Matrix Analysis and Applications}
  \textbf{29}(2): 656--659.

\bibitem[{Gaspari and Cohn(1999)}]{gaspari1999construction}
Gaspari G, Cohn SE. 1999. Construction of correlation functions in two and
  three dimensions. \emph{Quarterly Journal of the Royal Meteorological
  Society} \textbf{125}(554): 723--757.

\bibitem[{Golub and Van~Loan(2012)}]{golub2012matrix}
Golub GH, Van~Loan CF. 2012. \emph{Matrix computations}, vol.~3. JHU Press.

\bibitem[{Halko \emph{et~al.}(2011)Halko, Martinsson and
  Tropp}]{halko2011finding}
Halko N, Martinsson PG, Tropp JA. 2011. Finding structure with randomness:
  Probabilistic algorithms for constructing approximate matrix decompositions.
  \emph{SIAM review} \textbf{53}(2): 217--288.

\bibitem[{Henze \emph{et~al.}(2007)Henze, Hakami and Seinfeld}]{Henze07}
Henze DK, Hakami A, Seinfeld JH. 2007. Development of the adjoint of
  {GEOS}-{C}hem. \emph{Atmospheric Chemistry and Physics} \textbf{7}(9):
  2413--2433.

\bibitem[{Horn and Johnson(2012)}]{horn2012matrix}
Horn RA, Johnson CR. 2012. \emph{Matrix analysis}. Cambridge university press.

\bibitem[{Isotalo \emph{et~al.}(2008)Isotalo, Puntanen and
  Styan}]{isotalo2008blue}
Isotalo J, Puntanen S, Styan GP. 2008. The {BLUE}'s covariance matrix
  revisited: A review. \emph{Journal of Statistical Planning and Inference}
  \textbf{138}(9): 2722--2737.

\bibitem[{Jiang \emph{et~al.}(2011)Jiang, Jones, Kopacz, Liu, Henze and
  Heald}]{jiang2011quantifying}
Jiang Z, Jones D, Kopacz M, Liu J, Henze DK, Heald C. 2011. Quantifying the
  impact of model errors on top-down estimates of carbon monoxide emissions
  using satellite observations. \emph{Journal of Geophysical Research:
  Atmospheres} \textbf{116}(D15).

\bibitem[{Kopacz \emph{et~al.}(2009)Kopacz, Jacob, Henze, Heald, Streets and
  Zhang}]{kopacz2009comparison}
Kopacz M, Jacob DJ, Henze DK, Heald CL, Streets DG, Zhang Q. 2009. Comparison
  of adjoint and analytical {B}ayesian inversion methods for constraining asian
  sources of carbon monoxide using satellite ({MOPITT}) measurements of {CO}
  columns. \emph{Journal of Geophysical Research: Atmospheres}
  \textbf{114}(D4).

\bibitem[{Lanczos(1949)}]{lanczos50}
Lanczos C. 1949. An iteration method for solving the eigenvalue problem of
  linear differential operators. \emph{Bulletin of the American Mathematical
  Society} \textbf{55}(7): 717--718.

\bibitem[{Lorenc(2003)}]{lorenc2003potential}
Lorenc AC. 2003. The potential of the ensemble {K}alman filter for {NWP}: A
  comparison with 4{D}-{V}ar. \emph{Quarterly Journal of the Royal
  Meteorological Society} \textbf{129}(595): 3183--3203.

\bibitem[{M{\'e}n{\'e}trier and Aulign{\'e}(2015)}]{menetrier2015overlooked}
M{\'e}n{\'e}trier B, Aulign{\'e} T. 2015. An overlooked issue of variational
  data assimilation. \emph{Monthly Weather Review} \textbf{143}(10):
  3925--3930.

\bibitem[{M{\'e}n{\'e}trier \emph{et~al.}(2015)M{\'e}n{\'e}trier, Montmerle,
  Michel and Berre}]{menetrier2015linear}
M{\'e}n{\'e}trier B, Montmerle T, Michel Y, Berre L. 2015. Linear filtering of
  sample covariances for ensemble-based data assimilation. {P}art {I}:
  {O}ptimality criteria and application to variance filtering and covariance
  localization. \emph{Monthly Weather Review} \textbf{143}(5): 1622--1643.

\bibitem[{Meurant and Strako{\v{s}}(2006)}]{meurant2006lanczos}
Meurant G, Strako{\v{s}} Z. 2006. The lanczos and conjugate gradient algorithms
  in finite precision arithmetic. \emph{Acta Numerica} \textbf{15}: 471--542.

\bibitem[{M\"uller and Stavrakou(2005)}]{Muller05}
M\"uller J, Stavrakou T. 2005. Inversion of \uppercase{CO} and \uppercase{NO}x
  emissions using the adjoint of the \uppercase{IMAGES} model.
  \emph{Atmospheric Chemistry and Physics} \textbf{5}: 1157--1186.

\bibitem[{Nocedal and Wright(2006)}]{Nocedal06}
Nocedal J, Wright S. 2006. Numerical optimization, second edition.
  \emph{Numerical Optimization, Second Edition} : 1--664.

\bibitem[{Rabier and Courtier(1992)}]{Rabier92}
Rabier F, Courtier P. 1992. 4-dimensional assimilation in the presence of
  baroclinic instability. \emph{Quarterly Journal of the Royal Meteorological
  Society} \textbf{118}(506): 649--672.

\bibitem[{Rodgers(2000)}]{rodgers2000inverse}
Rodgers CD. 2000. \emph{Inverse methods for atmospheric sounding: theory and
  practice}, vol.~2. World scientific.

\bibitem[{Sherman and Morrison(1949)}]{sherman1949adjustment}
Sherman J, Morrison WJ. 1949. Adjustment of an inverse matrix corresponding to
  a change in one element of a given matrix. \emph{Annals of Mathematical
  Statistics} \textbf{20}: 317.

\bibitem[{Singh \emph{et~al.}(2011)Singh, Jardak, Sandu, Bowman, Lee and
  Jones}]{Singh11}
Singh K, Jardak M, Sandu A, Bowman K, Lee M, Jones D. 2011. Construction of
  non-diagonal background error covariance matrices for global chemical data
  assimilation. \emph{Geoscientific Model Development} \textbf{4}(2): 299--316.

\bibitem[{Spantini \emph{et~al.}(2015)Spantini, Solonen, Cui, Martin, Tenorio
  and Marzouk}]{spantini2015optimal}
Spantini A, Solonen A, Cui T, Martin J, Tenorio L, Marzouk Y. 2015. Optimal
  low-rank approximations of {B}ayesian linear inverse problems. \emph{SIAM
  Journal on Scientific Computing} \textbf{37}(6): 2451--2487.

\bibitem[{Tarantola(2005)}]{Tarantola05}
Tarantola A. 2005. \emph{Inverse problem theory and methods for model parameter
  estimation}. \uppercase{SIAM}.

\bibitem[{Thacker(1989)}]{thacker1989role}
Thacker WC. 1989. The role of the {H}essian matrix in fitting models to
  measurements. \emph{Journal of Geophysical Research: Oceans} \textbf{94}(C5):
  6177--6196.

\bibitem[{Tshimanga \emph{et~al.}(2008)Tshimanga, Gratton, Weaver and
  Sartenaer}]{Tshimanga08}
Tshimanga J, Gratton S, Weaver AT, Sartenaer A. 2008. Limited-memory
  preconditioners, with application to incremental four-dimensional variational
  data assimilation. \emph{Quarterly Journal of the Royal Meteorological
  Society} \textbf{134}(632): 751--769.

\bibitem[{Turner and Jacob(2015)}]{turner2015balancing}
Turner A, Jacob D. 2015. Balancing aggregation and smoothing errors in inverse
  models. \emph{Atmospheric Chemistry and Physics} \textbf{15}(12): 7039--7048.

\bibitem[{Wecht \emph{et~al.}(2014)Wecht, Jacob, Sulprizio, Santoni, Wofsy,
  Parker, B{\"o}sch and Worden}]{wecht2014spatially}
Wecht KJ, Jacob DJ, Sulprizio MP, Santoni G, Wofsy SC, Parker R, B{\"o}sch H,
  Worden J. 2014. Spatially resolving methane emissions in {C}alifornia:
  constraints from the {C}al{N}ex aircraft campaign and from present ({GOSAT},
  {TES}) and future ({TROPOMI}, geostationary) satellite observations.
  \emph{Atmospheric Chemistry and Physics} \textbf{14}(15): 8173--8184.

\bibitem[{Wells \emph{et~al.}(2015)Wells, Millet, Bousserez, Henze,
  Chaliyakunnel, Griffis, Luan, Dlugokencky, Prinn, O'Doherty
  \emph{et~al.}}]{wells2015simulation}
Wells K, Millet D, Bousserez N, Henze D, Chaliyakunnel S, Griffis T, Luan Y,
  Dlugokencky E, Prinn R, O'Doherty S, \emph{et~al.} 2015. Simulation of
  atmospheric {N2O} with {GEOS}-{C}hem and its adjoint: evaluation of
  observational constraints. \emph{Geoscientific Model Development}
  \textbf{8}(10): 3179--3198.

\bibitem[{Xu \emph{et~al.}(2013)Xu, Wang, Henze, Qu and
  Kopacz}]{xu2013constraints}
Xu X, Wang J, Henze DK, Qu W, Kopacz M. 2013. Constraints on aerosol sources
  using {GEOS}-{C}hem adjoint and {MODIS} radiances, and evaluation with
  multisensor ({OMI}, {MISR}) data. \emph{Journal of Geophysical Research:
  Atmospheres} \textbf{118}(12): 6396--6413.

\end{thebibliography}
\bibliographystyle{wileyqj}

\end{document}